\documentclass[12pt]{article}
\usepackage[english]{babel}
\usepackage[latin1]{inputenc}
\usepackage{amsfonts,amssymb,amsmath, epsfig}
\usepackage{color,graphicx,graphics,psfrag}
\usepackage{amsmath,amstext,amssymb,amsfonts, amscd}
\usepackage{amsthm}
\usepackage{amsmath,empheq}

\makeatletter

\@addtoreset{equation}{section}
\makeatother

\newtheorem{lemma}{Lemma}[section]  

\newtheorem{theorem}[lemma]{Theorem}

\newtheorem{definition}[lemma]{Definition}

\newtheorem{proposition}[lemma]{Proposition}

\newtheorem{remark}{Remark}[section]

\newtheorem{hypothesis}{Hypothesis}[section]

\title{Continuum model for linked fibers with alignment interactions} 
\author{P. Degond$^{1}$, F. Delebecque$^{2,3}$, D. Peurichard$^{2,3}$} 
\date{}
 
\begin{document}
\maketitle
    \begin{center}
1-Department of Mathematics, Imperial college London,\\
London SW7 2AZ, United Kingdom.\\
email:pdegond@imperial.ac.uk\\
2-Universite de Toulouse; UPS, INSA, UT1, UTM ;\\ 
Institut de Mathematiques de Toulouse ; \\
F-31062 Toulouse, France. \\
3-CNRS; Institut de Mathematiques de Toulouse UMR 5219 ;\\ 
F-31062 Toulouse, France.\\
email:fanny.delebecque@math.univ-toulouse.fr\\
email:diane.peurichard@math.univ-toulouse.fr
\end{center}
\maketitle

\begin{abstract}
We introduce an individual-based model for fiber elements having the ability to cross-link or unlink each other and to align with each other at the cross links. 
We first formally derive a kinetic model for the fiber and cross-links distribution functions. We then consider the fast linking/unlinking regime in which the model can be reduced to the fiber distribution function only  and investigate its diffusion limit. The resulting macroscopic model consists of a system of nonlinear diffusion equations for the fiber density and mean orientation. In the case of a homogeneous fiber density, we show that the model is elliptic.
\end{abstract}


\textbf{ccode:}{AMS Subject Classification: 82C31, 82C40, 82C70, 92C10, 92C17, }

\section{Introduction}
The topic of complex systems is attracting an increasingly abundant literature, due to its paramount importance in life and social sciences. Complex systems consist of a large number of agents interacting through local interactions only and yet able to self-organize into large-scale coherent structures and collective motion \cite{Vicsek_Zafeiris_PhysRep12}. Among examples of interactions leading to collective motion, the alignment interaction has been the subject of many studies since the seminal work of Vicsek and co-authors \cite{Vicsek_etal_PRL95}. In Vicsek's model, self-propelled point particles tend to align with their neighbors up to some noise. Vicsek's particles are polar: they carry a definite direction and orientation defined by the unit vector of their propulsion velocity. Their alignment interaction is also polar in the sense that a particle moving in an opposite direction to its neighbors will eventually reverse its direction of motion. However, other alignment rules have been studied as well. Polar particles can be subjected to nematic alignment. In this case, a particle moving in an opposite direction to its neighbors will not reverse its direction of motion, as opposed to the polar alignment case.  Nematic alignment has been used as a model for the volume exclusion interaction \cite{Baskaran_Marchetti_PRE08,Ginelli_etal_PRL10,Peruani_etal_PRE06} .Particles can also be apolar, for instance if they randomly reverse their direction of motion. Apolar particles interacting through nematic alignment have been proposed as a model for vibrating rods \cite{Bertin_etal_NewJPhys13}, or fiber networks \cite{Alonso_etal_CellMolBioeng14}. In the related field of nematic liquid crystals, volume exclusion interactions between rod-like particles are also modelled as an alignment force \cite{Doi_Edwards_Oxford99,Maier_Saupe_ZNaturforsch58,Onsager_AnnNYAcadSci49}. But additionally, the molecules are convected by the background solvent and are subjected to rotation by the fluid shear.  Additionally, they contribute to the fluid dynamics of the liquid solvent through an additional extra-stress tensor. Usually, the polymer chains are supposed of fixed length, although lately, models of polymer chains of variables lengths have appeared \cite{Ciuperca_etal_DCDS12}. 

In the present work, we are interested in a system consisting of fibers (or polymer chains) of variable lengths. This model aims to describe the network of collagen fibers in a fibrous tissue. We model fiber length variation (through polymerization / depolymerization) as well as the ability for the fibers to establish cross-links between them by the same basic rules described as follows. We assume the existence of a fiber unit element (or monomer) modeled as a line segment of fixed length $L$. We suppose that two fiber elements that cross each-other may form a link, thereby creating a longer fiber. There is no limit to the number of cross-links a given fiber can make. Therefore, the fibers have the ability to branch off and to achieve complex network topologies. We include fiber resistance to bending by assuming the existence of torque which, in the absence of any other force, makes the two linked fiber elements align with each other. Fibers are also subject to random positional and orientational noise and to external positional and orientational potential forces. Finally, cross-links may also be removed to model possible fiber breakage or depolymerization. 

Our model features apolar fiber particles (since they are not self-propelled), interacting through nematic alignment with the other fibers they are linked to. Thus, the model bears analogies with previous models of apolar particles interacting through nematic alignment \cite{Bertin_etal_NewJPhys13,Alonso_etal_CellMolBioeng14}. However, the interaction network topology (which keeps track of which fiber pairs are cross-linked) is different, as ours is determined by the distribution of cross-links. The fact that this network topology changes with time through dynamic cross-linking or unlinking processes is one specific feature of the present work. In the absence of cross-link remodeling, i.e.\ when the cross-links lifetime is infinite and no new cross-links is created, each connected component of the fiber network can be seen as an unstretchable elastic string since all connected fiber elements will spontaneously align with each other.  However, cross-link removal or creation events (supposed to occur at Poisson distributed random times) introduce a fluid-like component to the rheology of the fibers, thereby confering some visco-elastic character to the medium. Cross-link-governed statics and dynamics of fiber networks have been intensely studied in the literature \cite{Astrom_etal_PRE08,Broedersz_etal_PRL10,Buxton_etal_ExpPolLett09,Head_etal_PRE03,Oelz_etal_CellAdhMigr08} . However, most models consider passive cross-links which only act on the fibers by a spring-like attractive force. Here, our description introduces active links which tend to align the two fibers with each other. By doing so, we are also able to take into account fiber breakage, elongation and branching just in addition to and in the same way as fiber linking/unlinking because cross-linked fiber elements can be seen as two parts of the same fiber. Another difference from previous literature is that fibers in our model are subject to noise making the system more akin to a fluid or a gas than to a solid. By contrast to classical polymeric fluid studies, we do not assume that the fibers are transported by a fluid and modify its rheological properties but this feature could be added in future work. 

This model was first introduced in Ref. \cite{Peurichard_etal_preprint15} where it was coupled with the dynamics of spherical particles modelling cells. This model has been built to describe the self-organization of the adipose tissue, where spheres represent adipocytes and fibers, the surrounding collagen fibers. In this work, we demonstrated that the interaction between cells and fibers led to the spontaneous formation of cell clusters of ovoid shape akin to the adipose lobules that form the functional subunits of the adipose tissue. In Ref. \cite{Peurichard_etal_preprint15}, only a discrete Individual-Based Model (IBM) was considered. The present work focuses on the fibrous medium only and aims to derive meso and macroscopic models from the background IBM using techniques of kinetic theory. Indeed, the computational cost of an IBM scales polynomially with the number of agents, which makes them practically untractable for large systems. Continuum models allow to break this curse of scaling but they suppose that a suitable coarse-graining procedure which averages out the fine-scale structure has been applied to the IBM. In order to capture the correct effects of the fine-scale dynamics on the large-scale structures, it is of paramount importance to  perform this coarse-graining as rigorously as possible. This is the aim of the present work. 

The derivation of a continuum model from the fiber dynamics is done in two steps. We first derive a kinetic model from the underlying IBM and secondly, we perform a diffusion approximation of the latter to obtain the continuum model. The kinetic model provides a statistical mechanics description of the underlying IBM by investigating how the probability distribution of fibers in position and orientation space evolves in time. Here, we will show that the mere distribution of fibers is not sufficient to close the system and that the cross-link probability distribution needs to  be introduced. The cross-links provide correlations between the fibers and consequently their distribution can be viewed as similar to the two-particle fiber distribution. We will formally show that the knowledge of the one- and two-particle distributions is enough to provide a valid kinetic description of the system. Of course, this fact needs to be confirmed by numerical simulations and mathematical proofs. But if it proves correct, this model provides a unique example, to our knowledge, of a kinetic model which is closed at the level of the two-particle distribution function. Indeed, the question whether or not kinetic descriptions must include higher order distribution functions has been actively discussed in the recent years \cite{Carlen_etal_PhysicaD13,Carlen_etal_M3AS13,Mischler_Mouhot_InventMath13,Mischler_etal_PTRF15} . We also note that the introduction of the cross-link distribution functions provides an economic and efficient way of statistically tracking the fiber network topology. This methodology could prove interesting for other situations of dynamically evolving networks. 

The second step consists of a diffusion approximation of the previously derived kinetic model. It starts with changing the time and space units to macroscopic ones. The macroscopic space unit is large compared to the typical spatial scale of the fibers, e.g.\ their length and the macroscopic time unit is large to the typical time scale of the fibers, e.g.\ the time needed for two linked fibers to align with each other. A diffusive rescaling relates the time and space rescaling in such a way that the ratio of the microscopic to macroscopic time units is the square of that of the spatial units. This choice is made necessary by the absence of any polarization in the medium which makes diffusive behavior dominate. A key assumption that we make here is to assume that the linking/unlinking frequencies are very large: the typical linking/unlinking time measured in the macroscopic time unit scales like the square of the typical fiber alignment time (also measured in macroscopic unit), which is very small. This allows us to deduce an algebraic relation between the cross-link distribution function and the fiber distribution function, and to realize a closure of the kinetic equation at the level of the fiber distribution function alone. This assumption is questionable given the biological applications we have in mind, but it provides a first step towards a more complete theory involving finite linking/unlinking times. 

From these assumptions, we derive a singular perturbation problem for the fiber kinetic distribution function that has the form of a classical diffusion approximation problem \cite{Bardos_etal_TransAMS84,Degond_MasGallic_TTSP87,Poupaud_AsymptAnal91}, whose leading order collision operator comes from the nematic alignment of the fibers due to the alignment torque at the cross-links. This operator has equilibria in the form of generalized von Mises distributions of the fiber directions. The von Mises distribution extends Gaussian distributions to probabilities defined on the unit circle. It is peaked around a mean fiber direction angle $\theta_0$. The continuum model describes how the local fiber density $\rho$ and the local fiber direction $\theta_0$ vary as functions of position $x$ and time $t$. To obtain these evolution equations, we must integrate the kinetic equation against suitably chosen collision invariants. This operation cancels the singularly perturbed term. Here, the difficulty it that there exists only one such collision invariant in the classical sense, which allows us to find an equation for the density $\rho$ only. To find an equation for the mean fiber direction $\theta_0$, we use the recently developed theory of Generalized  Collision Invariants (GCI) \cite{Degond_etal_CMS15,Degond_etal_MAA13,Degond_Motsch_M3AS08,Frouvelle_M3AS12}. The resulting system is a nonlinear coupled system of diffusion equations for $\rho$ and $\theta_0$. In the case of a homogeneous fiber distribution, when the density is uniform in space and constant in time, we show that the resulting nonlinear diffusion model for $\theta_0$ is parabolic. In future work, it will be shown that this system is well-posed. Numerical simulations will demonstrate that the continuum model provides a consistent approximation of the underlying IBM for the fiber dynamics. Numerous macroscopic models for fibrous media have been previously considered in the literature but very few of them have been derived from an underlying IBM.  Most of them are heuristically derived from continuum theories such as mechano-chemical principles \cite{Alt_Dembo_MathBiosci99,Taber_etal_JMechMatStruct11}, thermodynamics \cite{Joanny_etal_NewJPhys07}, or viscous fluid mechanics \cite{Karsher_etal_BiophysJ03}.

The outline of this paper is as follows. In Section \ref{sec:IBM}, we start with the description of the IBM. Section \ref{sec:kinetic} is devoted to the derivation of the kinetic model. The scaling assumptions and the scaled kinetic equations are derived in Section \ref{sec:Scaling}. In Section \ref{sec:LSL}, we perform the large scale limit of the so-obtained equations. Finally, Section \ref{sec:homo} is devoted to the analysis of the model in the case of a homogeneous fiber density. Conclusions and perspectives are drawn in Section \ref{sec:conclu}. Some technical computations are detailed in Appendices. 

\setcounter{equation}{0}
\section{Individual Based Model for fibers interacting through alignment interactions}
\label{sec:IBM}

We intend to model a medium consisting of interconnected fibers. To simplify the geometric description of fibers, we decompose them into fiber elements of uniform fixed length and consider that a fiber consists of several connected fiber elements. The link between two connected fibers can be positionned at any point along the fibers (not only the extremities) and a given fiber can be connected to any number of other fibers, thereby allowing to model the branching off of a fiber into several branches. The links are not permanent. The topology of the fiber network is constantly remodelled through link creation/deletion processes.  To model fiber resistance to bending, we suppose that pairs of linked fibrs are subject to a torque that tends to align the two fibers with respect to each other. Finally, the fibers are subject to random positional and orientational noises to model the movements of the tissue and to positional and orientational potential forces to model the action of external elements. In the case of a fibrous tissue, these external elements may consist of cells or other tissues. 

In this paper, we restrict ourselves to a two-dimensional model. We consider a set of $N$ fiber elements modelled as small line segments of uniform and fixed length $L$, described by their center $X_i \in \mathbb{R}^2$ and their  angle $\theta_i$ with respect to a fixed reference direction. As the fiber elements are assumed apolar, $\theta_i$ is an angle of lines, i.e.\ $\theta_i \in [-\frac{\pi}{2},\frac{\pi}{2})$ modulo $\pi$.  We define energies related to each of the phenomena described above namely an energy for the maintenance of the links $W_{\mbox{\scriptsize{links}}}$, an energy for the alignment torque $W_{\mbox{\scriptsize{align}}}$, an energy for the action of the external elements $W_{\mbox{\scriptsize{ext}}}$, an energy for the noise contribution $W_{\mbox{\scriptsize{noise}}}$ and a total energy made of the sum of all these energies: 
\begin{equation}
\label{NRJtot}
W_{\mbox{\scriptsize{tot}}} = W_{\mbox{\scriptsize{links}}} + W_{\mbox{\scriptsize{ext}}} + W_{\mbox{\scriptsize{align}}}  + W_{\mbox{\scriptsize{noise}}},
\end{equation}  
All these energies are functions of the $N$ fiber positions $(X_i)_{i=1}^N$ and orientations $(\theta_i)_{i=1}^N$. Note that $W_{\mbox{\scriptsize{noise}}}$ is rather an entropy than an energy, so that $W_{\mbox{\scriptsize{tot}}}$ is indeed the total free energy of the system. Fiber motion and rotation during a time interval between two fiber linking-unlinking events is supposed to occur in the steepest descent direction to this free energy, namely according to:
\begin{empheq}[left=\empheqlbrace]{align} 
\frac{dX_i}{dt} = - \mu \, \nabla_{X_i} W_{\mbox{\scriptsize{tot}}}, \quad \forall i \in \{1, \ldots, N \}, \label{IBM1}\\
\frac{d\theta_i}{dt} = - \lambda \, \partial_{\theta_i} W_{\mbox{\scriptsize{tot}}} , \quad \forall i \in \{1, \ldots, N \}, \label{IBM2}.
\end{empheq}
Eqs.~\eqref{IBM1} and~\eqref{IBM2} express the motion and rotation of the individuals in an overdamped regime in which the forces due to friction are very large compared to the inertial forces. Fiber velocity and angular speed are proportional to the force exerted on the fiber through two mobility coefficients $\mu$ and $\lambda$ which are considered given. We now detail the expressions of the four energies involved in the expression~\eqref{NRJtot} of the total free energy of the system, as well as how Eqs.~\eqref{IBM1} and~\eqref{IBM2} are supplemented by Poisson jump processes when a linking/unlinking event occurs. 

To define the expression of $W_{\mbox{\scriptsize{links}}}$, we consider a time at which no linking/unlinking process occurs. Then, the set of links is well-defined and supposed to have $K$ elements. Let $k \in \{ 1, \ldots, K\}$ be a given link and denote by $(i(k),j(k))$ the pair of indices corresponding to the two fibers connected by this link. To make the labeling of the pair unique, we assume without loss of generality that the first element of the linked pair is always the one with lowest index, i.e.\ $i(k) < j(k)$. The link is supposed to connect two points $X_{i(k)}^k$ and  $X_{j(k)}^k$ on fibers $i(k)$ and $j(k)$ respectively. These points are determined by the algebraic distances $\ell_{i(k)}^k$ and $\ell_{j(k)}^k$ to the centers $X_{i(k)}$ and  $X_{j(k)}$ of the two fibers respectively; We thus have the relation:
$$ X_{i(k)}^k = X_{i(k)} + \ell_{i(k)}^k \omega_{i(k)} , \quad X_{j(k)}^k = X_{j(k)} + \ell_{j(k)}^k \omega_{j(k)},  $$
where $\ell_{i(k)}^k$, $\ell_{j(k)}^k \in [-L/2,L/2]$ and where, for any fiber $i$, we let $ \omega_i = (\cos \theta_i, \sin \theta_i)$ be the unit vector in the direction of the fiber. All along the link lifetime, the link places a spring-like restoring force that attracts $X_{i(k)}$ back to $X_{j(k)}$ (and vice-versa) as soon as their are displaced one with respect to each other. This restoring force gives rise to a potential energy $V(X_{i(k)}, \theta_{i(k)}, \ell_{i(k)}^k, X_{j(k)}, \theta_{j(k)}, \ell_{j(k)}^k)$, with 
\begin{equation}
\label{V}
V(X_1,\theta_1,\ell_1,X_2,\theta_2,\ell_2) = \frac{\kappa}{2} |X_1  + \ell_1 \omega(\theta_1) -  (X_2  + \ell_2 \omega(\theta_2))|^2, 
\end{equation}
where $\kappa$ is the intensity of the restoring force. Obviously, the larger $\kappa$, the better the maintainance of the link is ensured. The potential $W_{\mbox{\scriptsize{links}}}$ is then assumed to be the sum of all the linked fiber spring forces:
\begin{equation}\label{Wlinks}
W_{\mbox{\scriptsize{links}}} = \frac{1}{2}\sum_{k=1}^K V(X_{i(k)},\theta_{i(k)}, \ell_{i(k)}^k,X_{j(k)},\theta_{j(k)}, \ell_{j(k)}^k).
\end{equation}
We stress the fact that the quantities $\ell_{i(k)}^k$ and $\ell_{j(k)}^k$ remain constant throughout the link lifetime. They are determined at the time of the creation of the link (see below and Fig.~\ref{fig1}). 

\begin{figure}[ht]
\center
\includegraphics[width=12cm,height=3cm]{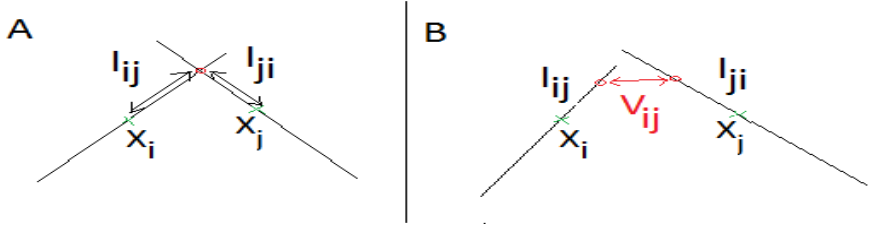}
\caption{Intersecting linked fibers. $l_{ij}$ and $l_{ji}$ refer to $\bar{\ell}(X_i, \theta_i, X_j, \theta_j)$  and $\bar{\ell}(X_j, \theta_j, X_i, \theta_i)$~\eqref{tfk}. A. Situation at linking time. B. Restoring potential $V_{ij}$~\eqref{V} after motion of the fibers. }
   \label{fig1}
\end{figure} 

The external potential $W_{\mbox{\scriptsize{ext}}}$ associated with the external forces is supposed to be the sum of potential forces $U(X_i,\theta_i)$ acting on each of the $N$ fibers:
\begin{equation}\label{potext}
W_{\mbox{\scriptsize{ext}}} = \sum_{i=1}^{N} U(X_i,\theta_i).
\end{equation}
\noindent Here, $U(x,\theta)$ is a given, possibly time-dependent smooth function. In the case where the system describes the collagen fibers in a tissue, $U$ aims to model the presence of cells or other organs.

Linked fibers are subjected to an alignment force at their junction to model fiber resistance to bending. This force tends to align linked fibers $i(k)$ and $j(k)$ and derives from the potential $b(\theta_{i(k)},\theta_{j(k)})$ which reads:
\begin{equation}\label{b}
b(\theta_1,\theta_2) = \alpha |\sin(\theta_1 - \theta_2)|^\beta,
\end{equation}
\noindent where $\alpha$ plays the role of a flexural modulus and $\beta$ is a modeling parameter. The binary alignment potential only depends on the angles $\theta_1$ and $\theta_2$, and the total alignment energy $W_{\mbox{\scriptsize{align}}}$ is supposed to be the sum of all the binary alignment interactions:
\begin{equation}\label{align}
W_{\mbox{\scriptsize{align}}}= \frac{1}{2}\sum_{k=1}^K b(\theta_{i(k)},\theta_{j(k)}). 
\end{equation}

We include random positional and orientational motion of the fiber elements which, in the context of tissue dynamics, originate from the random movements of the subject. With this aim, we introduce an entropy term:
\begin{equation}\label{noise}
W_{\mbox{\scriptsize{noise}}} = d \sum_{i=1}^N \log(\tilde{f})(X_i, \theta_i) ,
\end{equation}
\noindent where $\tilde{f}$ is a 'regularized density':
\begin{equation*}
\tilde{f}(x,\theta) = \frac{1}{N} \sum_{i=1}^N \xi^N(x-X_i) \, \eta^N(\theta-\theta_i).
\end{equation*}
\noindent Here, $\xi^N$ and $\eta^N$ are regularization functions which allow to define the logarithm of $\tilde{f}$ and have the following properties: 
\begin{empheq}[left=\empheqlbrace]{align*}
& \xi^N \in C^\infty({\mathbb R}^2), \quad \eta^N \in C_{per}^\infty([-\frac{\pi}{2},\frac{\pi}{2}]), \quad \xi^N \geq 0, \quad \eta^N \geq 0, \\
&\int\xi^N(x) dx = 1, \quad \int_{-\pi}^\pi \eta^N(\theta) \frac{d\theta}{2\pi} = 1,\\
&\text{Supp}(\xi^N) \subset B(0,R^N), \quad
 \text{Supp}(\eta^N) \subset [-M^N,M^N],
\end{empheq}
\noindent where $C^\infty({\mathbb R}^2)$ is the set of infinitely differentiable functions on ${\mathbb R}^2$, $C_{per}^\infty([-\frac{\pi}{2},\frac{\pi}{2}])$ the set of periodic $C^\infty$ functions of $[-\frac{\pi}{2},\frac{\pi}{2}]$ and Supp stands for the support of a function. Here, $R^N$ and $M^N$ are chosen such that $\sqrt{N}R^N$ and $N M^N \rightarrow \infty$ as $N \to \infty$. The mean interparticle distance in $x$ and $\theta$ are respectively of order $\frac{1}{\sqrt{N}}$ and $\frac{1}{N}$. This condition is equivalent to $\frac{1}{\sqrt{N}R^N} \rightarrow 0$ and $\frac{1}{N M^N} \rightarrow 0$, which means that as $N\rightarrow \infty$, the number of particles inside the support of a regularizing kernel tends to infinity. This way of modeling the influence of the noise is customary in polymer dynamics \cite{Bird_etal_Wiley87}.In the next section, we show that such an entropy term gives rise to  diffusion terms at the level of the mean-field kinetic model.

By inserting~\eqref{Wlinks},~\eqref{potext},~\eqref{align} and~\eqref{noise} into~\eqref{IBM1},~\eqref{IBM2}, we find the fiber equation of motion, during any time interval between two linking/unlinking events: 
\begin{align*}
\frac{d X_i}{dt} = &-\mu \bigg[  \nabla_x (U+\log \tilde{f}^N)(X_i,\theta_i)\\
&+\frac{1}{2}\sum_{k=1, i(k) = i}^K \nabla_{x_1} V(X_{i(k)},\theta_{i(k)},\ell^k_{i(k)},X_{j(k)},\theta_{j(k)},\ell^k_{j(k)}) \\
&+ \frac{1}{2}\sum_{k=1, j(k) = i}^K \nabla_{x_2} V(X_{i(k)},\theta_{i(k)},\ell^k_{i(k)},X_{j(k)},\theta_{j(k)},\ell^k_{j(k)})\bigg],
\end{align*}
\begin{align*}
\frac{d \theta_i}{dt} = &-\lambda \bigg[ \partial_\theta (U+\log \tilde{f}^N)(X_i,\theta_i)\\
&+\frac{1}{2}\sum_{k=1, i(k) = i}^K \partial_{\theta_1} V(X_{i(k)},\theta_{i(k)},\ell^k_{i(k)},X_{j(k)},\theta_{j(k)},\ell^k_{j(k)}) \\
&+ \frac{1}{2}\sum_{k=1, j(k) = i}^K \partial_{\theta_2} V(X_{i(k)},\theta_{i(k)},\ell^k_{i(k)},X_{j(k)},\theta_{j(k)},\ell^k_{j(k)})  \\
& + \frac{1}{2}\sum_{k=1, i(k) = i}^K \partial_{\theta_1} b(\theta_{i(k)},\theta_{j(k)}) + \frac{1}{2}\sum_{k=1, j(k) = i}^K \partial_{\theta_2} b(\theta_{i(k)},\theta_{j(k)})  \bigg],
\end{align*}
\noindent which we can write:
\begin{empheq}[left=\empheqlbrace]{align} 
\frac{d X_i}{dt} =& -\mu \bigg[ \bigg(\frac{1}{2}\sum_{k=1}^K \delta_{i(k)}(i) \nabla_{x_1} V + \frac{1}{2}\sum_{k=1}^K \delta_{j(k)}(i) \nabla_{x_2} V\bigg) (C^k_{i(k),j(k)}) \nonumber\\
&+ \nabla_x (U+\log \tilde{f}^N)(X_i,\theta_i)\bigg], \label{motionIBMX}\\
\frac{d \theta_i}{dt} =& -\lambda \bigg[ \partial_\theta (U+\log \tilde{f}^N)(X_i,\theta_i) \nonumber\\
&+\bigg( \frac{1}{2}\sum_{k=1}^K \delta_{i(k)}(i) \partial_{\theta_1} V + \frac{1}{2}\sum_{k=1}^K \delta_{j(k)}(i) \partial_{\theta_2} V \bigg)(C^k_{i(k),j(k)}) \label{motionIBMTETA}\\
& + \bigg(\frac{1}{2}\sum_{k=1}^K \delta_{i(k)}(i) \partial_{\theta_1} b + \frac{1}{2}\sum_{k=1}^K \delta_{j(k)}(i) \partial_{\theta_2} b\bigg)(\theta_{i(k)},\theta_{j(k)}) \bigg], \nonumber
\end{empheq}
\noindent with $C^k_{i(k),j(k)} = (X_{i(k)},\theta_{i(k)},\ell^k_{i(k)}, X_{j(k)},\theta_{j(k)},\ell^k_{j(k)})$ and $\delta_{i(k)}(i)$ is the Kronecker symbol, i.e.\ $\delta_{i(k)}(i) = 1$ if $i(k) = i$ and $\delta_{i(k)}(i) = 0$ otherwise. 

When two fibers $i$ and $j$ intersect each other, because of the continuity of their motion, they are going to intersect each other during a time interval $[t_*, t^*]$. We assume that, during this time span, the linking probability follows a Poisson process of parameter $\nu_f$, i.e.\ the probability that a link is formed during the interval $[t_*, t]$ with $t<t^*$ is $1-e^{-\nu_f (t-t_*)}$. Only one link can be formed between the two fibers of the same fiber pair. Supposing that a link, indexed by $k$ is formed between the fibers $i$ and $j$ (such that $i = i(k)$ and $j = j(k)$ if $i<j$) at a time $t_k \in [t_*, t^*]$, we denote by $X_k$ the attachment site of the link. The distance $\bar{\ell}(X_{i(k)}, \theta_{i(k)}, X_{j(k)}, \theta_{j(k)})$ between the center $X_{i(k)}$ of fiber $i(k)$ to the $k$-th link attachment site $X^k$ with fiber $j(k)$ (see Figure~\ref{fig1}.B) can be directly computed by:
\begin{equation}\label{tfk}
\bar{\ell}(X_{i(k)}, \theta_{i(k)}, X_{j(k)}, \theta_{j(k)}) = \frac{(x_{j(k)} - x_{i(k)})\sin \theta_{j(k)}  -  (y_{j(k)} - y_{i(k)})\cos\theta_{j(k)}}{\sin(\theta_{j(k)} - \theta_{i(k)})},
\end{equation}
\noindent where $X_{i(k)} = (x_{i(k)},y_{i(k)})$ are the coordinates of the center of fiber $i(k)$. For $X = (x,y)$ and $\omega = (\alpha,\beta)$, we denote by $X \times \omega = x\beta - y\alpha$. Then, $\bar{\ell}(X_{i(k)}, \theta_{i(k)}, X_{j(k)}, \theta_{j(k)})$ can be written: 
\begin{equation*}
\bar{\ell}(X_{i(k)}, \theta_{i(k)}, X_{j(k)}, \theta_{j(k)}) = \frac{|(X_{j(k)} - X_{i(k)}) \times \omega(\theta_{j(k)}|}{|\omega(\theta_{i(k)}) \times \omega(\theta_{j(k)})|},
\end{equation*}
\noindent where again, $\omega(\theta) = (\cos \theta , \sin \theta)$ is the directional vector associated to angle $\theta$. The fact that the two fibers are intersecting each other at time $t_k$ is written: 
$$
|\bar{\ell}(X_{i(k)}, \theta_{i(k)}, X_{j(k)}, \theta_{j(k)})|\leq \frac{L}{2}, \quad \mbox{and} \quad |\bar{\ell}(X_{j(k)}, \theta_{j(k)}, X_{i(k)}, \theta_{i(k)})|\leq \frac{L}{2}, 
$$
where $L$ is the fiber length and where all positions and angles are evaluated at time $t_k$. The quantities $\bar{\ell}(X_{i(k)}, \theta_{i(k)}, X_{j(k)}, \theta_{j(k)})$ and $\bar{\ell}(X_{j(k)}, \theta_{j(k)}, X_{i(k)}, \theta_{i(k)})$ at the time $t_k$ of the formation of the link set the positions of the attachment sites $X_{i(k)}^k$ and $X_{j(k)}^k$ of the link on fibers $i$ and $j$. Therefore, $\ell_{i(k)}^k$ and $\ell_{j(k)}^k$ remain constant throughout the link lifetime and equal to their value at the time $t_k$. So, we have 
$$ \frac{d}{dt} \ell_{i(k)}^k = \frac{d}{dt} \ell_{j(k)}^k = 0, $$
throughout the lifetime of the link. 

We also assume that existing links can disappear according to a Poisson random process of parameter $\nu_d$, i.e.\ the probability that the link disappears in the time interval $[t_k, t]$ is $1 - e^{-\nu_d (t - t^k)}$. 

The next section is devoted to the asymptotic limit $N,K \rightarrow \infty$ of this model.


\setcounter{equation}{0}
\section{Derivation of a kinetic model}
\label{sec:kinetic}

Here, the derivation of a kinetic model from the Individual Based Model of section~\ref{sec:IBM} is performed. The empirical measure  $f^N(x,\theta,t)$ of the fibers is introduced: 
\begin{equation*}
f^N(x,\theta,t) = \frac{1}{N} \sum_{i=1}^N \delta_{(X_i(t),\theta_i(t))}(x,\theta),
\end{equation*}
\noindent
where $\delta_{(X_i(t),\theta_i(t))}(x,\theta)$ denotes the Dirac delta located at $(X_i(t),\theta_i(t)$. It gives the probability to find a fiber at point $x$ and orientational angle $\theta$ at time $t$. The empirical measure $g^K(x_1,\theta_1, \ell_1,x_2,\theta_2,\ell_2,t)$ of the fiber links is given by:
\begin{equation*}
\begin{split}
g^K(x_1,\theta_1, \ell_1,x_2,\theta_2,\ell_2,t) =& \frac{1}{2K} \sum_{k=1}^K \delta_{(X_{i(k)},\theta_{i(k)},\ell^k_{i(k)}, X_{j(k)},\theta_{j(k)}, \ell^k_{j(k)})}(x_1,\theta_1,\ell_1,x_2,\theta_2,\ell_2)\\
&+ \delta_{(X_{j(k)},\theta_{j(k)},\ell^k_{j(k)}, X_{i(k)},\theta_{i(k)}, \ell^k_{i(k)})}(x_1,\theta_1,\ell_1,x_2,\theta_2,\ell_2),
 \end{split}
\end{equation*} 
\noindent 
with a similar definition of the Dirac deltas. It gives the probability of finding a link with associated lengths within a volume $d\ell_1 d\ell_2$ about $\ell_1$ and $\ell_2$, this link connecting a fiber located within a volume $dx_1 \frac{d\theta_1}{\pi}$ about $(x_1,\theta_1)$ with a fiber located within a volume $dx_2 \frac{d\theta_2}{\pi}$ about $(x_2,\theta_2)$. One notes that $(\ell_1,\ell_2)$ is defined in $[-\frac{L}{2},\frac{L}{2}]^2$.  Then, at the limit $N,K \rightarrow \infty$, $\frac{K}{N} \rightarrow \xi$, where $\xi>0$ is a fixed parameter, $f^N \rightarrow f$, $g^K \rightarrow g$ where $f$ and $g$ satisfy equations given in the following theorem:
\begin{theorem}\label{thm1}
The formal limit of Eqs.~\eqref{IBM1},~\eqref{IBM2} for $K,N \rightarrow \infty$, $\frac{K}{N} \rightarrow \xi$, where $\xi>0$ is a fixed parameter reads:
\begin{equation}\label{systtotf}
 \frac{df}{dt} - \mu \bigg( \nabla_x \cdot ((\nabla_x U) f) + \xi \nabla_x \cdot F_1 + d\Delta_x f \bigg) - \lambda \bigg(\partial_{\theta} ((\partial_{\theta} U) f) + \xi \partial_{\theta} F_2 + d \partial^2_{\theta} f \bigg)=0,
  \end{equation}
 \noindent and
 \begin{equation}\label{systtotg}
 \begin{split}
\frac{dg}{dt} - \mu &\bigg( \nabla_{x_1} \cdot \big(g\nabla_x U(x_1,\theta_1) + \xi \frac{g}{f(x_1,\theta_1)}F_1(x_1,\theta_1)\big)\\ &+\nabla_{x_2} \cdot \big(g\nabla_{x} U(x_2,\theta_2) +  \xi \frac{g}{f(x_2,\theta_2)}F_1(x_2,\theta_2)\big)\\
& +d\nabla_{x_1}\cdot ( \frac{g}{f(x_1,\theta_1)}\nabla_x f(x_1,\theta_1)) + d \nabla_{x_2}\cdot ( \frac{g}{f(x_2,\theta_2)}\nabla_{x} f(x_2,\theta_2))\bigg)\\
- \lambda &\bigg(\partial_{\theta_1} \big(g\partial_{\theta} U(x_1,\theta_1)+  \xi \frac{g}{f(x_1,\theta_1)}F_{2}(x_1,\theta_1)\big)\\ &+\partial_{\theta_2}  \big(g\partial_{\theta} U(x_2,\theta_2)+ \xi\frac{g}{f(x_2,\theta_2)}F_{2}(x_2,\theta_2) \big)\\
&+ d\partial_{\theta_1}( \frac{g}{f(x_1,\theta_1)} \partial_\theta f(x_1,\theta_1)) + d\partial_{\theta_2}(\frac{g}{f(x_2,\theta_2)} \partial_\theta f(x_2,\theta_2) )\bigg)=S(g) ,
\end{split}
\end{equation}
\noindent where
\begin{empheq}[left=\empheqlbrace]{align} 
F_1(x_1,\theta_1) &= \int\limits (g \nabla_{x_1} V)(x_1,\theta_1,\ell_1,x_2,\theta_2,\ell_2) d\ell_1 d\ell_2  \frac{d\theta_2}{\pi}dx_2, \label{F1}\\
F_2(x_1,\theta_1) &= \int\limits \big( g(\partial_{\theta_1} V + \partial_{\theta_1} b) \big)(x_1,\theta_1,\ell_1,x_2,\theta_2,\ell_2)d\ell_1 d\ell_2  \frac{d\theta_2}{\pi}dx_2,\label{F2}
\end{empheq}
\noindent 
and $S(g)$ is given by:
\begin{equation}\label{Sg}
S(g) = \nu_f f(x_1,\theta_1)f(x_2,\theta_2)\delta_{\bar{\ell}(x_1,\theta_1,x_2,\theta_2)}(\ell_1)\delta_{\bar{\ell}(x_2,\theta_2,x_1,\theta_1)}(\ell_2)- \nu_d g,
\end{equation}
\noindent where $\delta_{\bar{\ell}}(\ell_1)$ denotes the Dirac delta at $\bar{\ell}$, i.e.\ the distribution acting on test functions $\phi(\ell_1)$ such that $\langle \delta_{\bar{\ell}}(\ell_1),\phi(\ell_1)\rangle  = \phi(\bar{\ell})$
\end{theorem}

This kinetic model consists of two evolution equations. The first one (Eq.~\eqref{systtotf}) is an equation for the individual fibers  and describes the evolution of the one-particle distribution function $f$. Eq.~\eqref{systtotg} is an equation for the links between fiber pairs. The distribution function $g$ describes the correlations between fiber pairs brought by the presence of links. It can be viewed as a kind of  two-particle fiber distribution function. This model is, to our knowledge, a unique explicit example of a kinetic model written in terms of the one and two particle ditributions and closed at this level. Also, the distribution function $g$ can be seen as a way of describing the random graph of the fiber links, namely the graph where the nodes are the fibers and the edges are the links. This statistical description of a random graph could be useful to describe other kinds of random networks, notably in social sciences. As the links are tightly tied to the fibers, they are convected by them and follow their motion. Simultaneously, they constrain the linked fibers to move together, so they directly influence their motion. The action of the links on the individual fiber motion is contained in the third and sixth force terms $F_1$ and $F_2$ of Eq.~\eqref{systtotf} and are the kinetic counterparts of~\eqref{V}. The second and fith terms describe transport in physical and orientational spaces due to the external potential and are the kinetic counterparts of~\eqref{potext}. The fourth and seventh terms  are diffusion terms of amplitude $\lambda d$ and $\mu d$ respectively. They represent the random motion of the fibers and originate from the interactions described by Eq.~\eqref{noise}. The individual motion of the fibers is thus related to the motion of its linked neighbors. The left-hand side of Equation~\eqref{systtotg} describes the evolution of the links between fibers. Indeed, it is composed of the convective terms generated by the external potential and by the diffusion terms. The forces induced by the restoring potential generated by the links again gives rise to the non local terms $F_1$ and the first term of $F_{2}$. The kinetic counterpart of the alignment force between linked fibers (see Eq.~\eqref{align}) is encompassed in the second term of the force $F_2$ and only acts on the orientation of the fibers. The right hand side $S(g)$ of equation~\eqref{systtotg} describes the Poisson processes of linking/unlinking at frequencies $\nu_f$ and $\nu_d$, respectively. The first term describes the formation of the link and the Dirac deltas indicate that, at the link creation time, the link lengths $\ell_1$ and $\ell_2$ are set by the geometric configuration of the fibers at the attachment time. Also, because $\ell_1$ and $\ell_2$ are restricted to lie in the interval $[-L/2,L/2]$, we see that the link creation term is non-zero only when two fiber elements are intersecting each other. The second term just describes a decay of the link distribution at the rate set by the Poisson process, i.e.\ $\nu_d$. 

The formal proof of this result is inspired from Ref. \cite{Sone2002}, and the detailed computations can be found in appendix~\ref{proofIBMcont}. The rigorous proof of this result is an open question and is left for future work. 

\setcounter{equation}{0}
\section{Scaling}
\label{sec:Scaling}

\subsection{Dimensionless Equations}
We express the problem in dimensionless variables.  
For this purpose, let $t_0$ be the unit of time and $x_0$, $f_0=\frac{1}{x_0^2}$, $g_0 = \frac{1}{x_0^6}$ and $U_0 = \frac{x_0^2}{t_0^2}$ the units of space, distribution function and energy. The scaling of $f(x,\theta)$ and $g(x_1,\theta_1,\ell_1,x_2,\theta_2,\ell_2)$ comes from the fact that they are probability distribution functions on a 2D domain. The following dimensionless variables are defined: 
\begin{equation*}
\bar{x} = \frac{x}{x_0},\; \bar{\ell} =\frac{\ell}{x_0},\; \bar{f}=\frac{f}{f_0} = f x_0^2, \;  \bar{g} = \frac{g}{g_0} = g x_0^6,\;  \bar{U} = \frac{t_0^2 U}{x_0^2}.
\end{equation*}
\noindent and the following dimensionless parameters are introduced: 
\begin{equation*}
\mu' = \frac{\mu}{t_0},\;  \lambda' = \frac{\lambda x_0^2}{t_0}, \; \nu_f' = t_0 \nu_f,\;  \nu_d' = t_0 \nu_d, L'=\frac{L}{x_0}, \; d'=\frac{d t_0^2}{x_0^2}, \alpha'=\frac{\alpha t_0^2}{x_0^2}, \; \kappa' = \kappa t_0^2.
\end{equation*}
\noindent First of all, from the expression of $V$ (see Eq.~\eqref{V}), we get:
\begin{equation*}
\begin{split}
V(x_1,\theta_1,\ell_1,x_2,\theta_2,\ell_2) &= \frac{\kappa'}{2t_0^2}(x_1+\ell_1\omega(\theta_1) - x_2-\ell_2 \omega(\theta_2))^2 \\
&=  \frac{x_0^2}{t_0^2} \bar{V}(\bar{x}_1,\theta_1,\bar{\ell}_1,\bar{x}_2,\theta_2,\bar{\ell}_2), 
\end{split}
\end{equation*}
with 
$$ \bar{V}(\bar{x}_1,\theta_1,\bar{\ell}_1,\bar{x}_2,\theta_2,\bar{\ell}_2) =  \frac{\kappa'}{2} (\bar{x}_1+\bar{\ell}_1 \omega(\theta_1) - \bar{x}_2-\bar{\ell}_2 \omega(\theta_2))^2. $$

\noindent Now, from Eqs.~\eqref{F1}-\eqref{F2}, one notes that:
\begin{equation*}
F_1(x_1,\theta_1) = \frac{1}{x_0t_0^2} \bar{F}_1(\bar{x}_1,\theta_1),
\end{equation*}
\noindent where 
\begin{equation*}
\begin{split}
\bar{F}_1(\bar{x}_1,\theta_1)= \int_{\mathcal{L}'} \nabla_{\bar{x}_1} \bar{V}(\bar{x}_1,\theta_1,\bar{\ell}_1,\bar{x}_2,\theta_2,\bar{\ell}_2) \bar{g}(\bar{x}_1,\theta_1,\bar{\ell}_1,\bar{x}_2,\theta_2,\bar{\ell}_2) d\bar{\ell}_1 d\bar{\ell}_2 \frac{d\theta_2}{\pi} d\bar{x}_2,
\end{split}
\end{equation*}
\noindent with $\mathcal{L}' = {\mathbb{R}^2} \times [-\frac{\pi}{2},{\frac{\pi}{2}] \times [-\frac{L'}{2}},{\frac{L'}{2}}]\times [-\frac{L'}{2},{\frac{L'}{2}}]$. Similarly, $F_2(x_1,\theta_1) = \frac{1}{t_0^2} \bar{F}_2(\bar{x}_1,\theta_1)$, where:
\begin{empheq}[left=\empheqlbrace]{align*} 
\bar{F}_2(\bar{x}_1,\theta_1) &= \bar{F}_{al}(\bar{x}_1,\theta_1) + \bar{F}_{link}(\bar{x}_1,\theta_1), \\
\bar{F}_{link}(\bar{x}_1,\theta_1) &= \int_{\mathcal{L}'} \big(\bar{g} \partial_{\theta_1}V \big)(\bar{x}_1,\theta_1,\bar{\ell}_1,\bar{x}_2,\theta_2,\bar{\ell}_2) d\bar{\ell}_1 d\bar{\ell}_2 \frac{d\theta_2}{\pi} d\bar{x}_2 ,\\ 
\bar{F}_{al}(\bar{x}_1,\theta_1) &= \int_{\mathcal{L}'} \big(\bar{g}\partial_{\theta_1} \bar{b}\big)(\bar{x}_1,\theta_1,\bar{\ell}_1,\bar{x}_2,\theta_2,\bar{\ell}_2)d\bar{\ell}_1 d\bar{\ell}_2 \frac{d\theta_2}{\pi} d\bar{x}_2 ,
\end{empheq}
\noindent where $\bar{b}(\theta_1,\theta_2) = \alpha' \sin(\theta_1-\theta_2)^\beta$.  In this new set of variables, Eqs.~\eqref{systtotf}-\eqref{systtotg} become:
;\begin{align*}
 \partial_{t'} \bar{f}- \chi \lambda' \nabla_{\bar{x}}\cdot (\nabla_{\bar{x}} \bar{U}\bar{f}) -\lambda' \partial_{\theta} (\partial_{\theta} \bar{U}\bar{f}) 
  -  \xi \lambda' \partial_{\theta} \bar{F}_2 -& \chi  \xi \lambda' \nabla_{\bar{x}}\cdot \bar{F}_1 \\
  &- d' \lambda' \partial^2_{\theta} \bar{f} - d' \chi \lambda' \Delta_x \bar{f} = 0, 
\end{align*}
and 
\begin{align*}
\partial_{t'}\bar{g} - \chi \lambda'\nabla_{\bar{x}_1} \cdot (\bar{g}\nabla_{\bar{x}} U(\bar{x}_1,\theta_1) +   \xi\frac{\bar{g}}{\bar{f}(\bar{x}_1,\theta_1)} \bar{F}_1(\bar{x}_1,&\theta_1)) \\
- \lambda'\partial_{\theta_1} (\bar{g}\partial_{\theta} U(\bar{x}_1,\theta_1)+  \xi\frac{\bar{g}}{\bar{f}(\bar{x}_1,\theta_1)}\bar{F}_{2}(\bar{x}_1,\theta_1)) &\\
-\chi \lambda'\nabla_{\bar{x}_2} \cdot (\bar{g}\nabla_{\bar{x}} U(\bar{x}_2,\theta_2) +  \xi\frac{\bar{g}}{\bar{f}(\bar{x}_2,\theta_2)}&\bar{F}_1(\bar{x}_2,\theta_2))\\
- \lambda'\partial_{\theta_2}  (\bar{g}\partial_{\theta} U(\bar{x}_2,\theta_2)+ \xi\frac{\bar{g}}{\bar{f}(\bar{x}_2,\theta_2)} &\bar{F}_{2}(\bar{x}_2,\theta_2) )\\
-d'\chi \lambda'\nabla_{\bar{x}_1}\cdot (\frac{\bar{g}}{\bar{f}(\bar{x}_1,\theta_1)}\nabla_{\bar{x}} \bar{f}(\bar{x}_1,&\theta_1) ) \\
- d'\chi \lambda' \nabla_{\bar{x}_2}\cdot  (\frac{\bar{g}}{\bar{f}(\bar{x}_2,\theta_2)}\nabla_{\bar{x}} &\bar{f}(\bar{x}_2,\theta_2) ) \\
- d'\lambda'  \partial_{\theta_1}(\frac{\bar{g}}{\bar{f}(\bar{x}_1,\theta_1)} &\partial_{\theta} \bar{f}(\bar{x}_1,\theta_1) ) \\
- d'\lambda'\partial_{\theta_2}( &\frac{\bar{g}}{\bar{f}(\bar{x}',\theta_2)}\partial_{\theta} \bar{f}(\bar{x}_2,\theta_2))=\bar{S}(\bar{g}),
 \end{align*}
\noindent where $\chi =\frac{\mu'}{\lambda'}$ and:
 \begin{equation*}
 \begin{split}
\bar{S}(\bar{g})(\bar{x}_1,\theta_1,\bar{\ell}_1,\bar{x}_2,\theta_2,\bar{\ell}_2) = &\nu_f' \bar{f}(\bar{x}_1,\theta_1)\bar{f}(\bar{x}_2,\theta_2)\delta_{\bar{\ell}(\bar{x}_1,\theta_1,\bar{x}_2,\theta_2)}(\bar{\ell}_1)\delta_{\bar{\ell}(\bar{x}_2,\theta_2,\bar{x}_1,\theta_1)}(\bar{\ell}_2)\\
&- \nu_d'\bar{g}(\bar{x}_1,\theta_1,\bar{\ell}_1,\bar{x}_2,\theta_2,\bar{\ell}_2).
\end{split}
\end{equation*}
\noindent Finally, if the space and time scales $x_0$, $t_0$ are chosen such that $\lambda' = \chi = 1$, i.e:
$$
x_0^2 = \frac{\mu}{\lambda}, \quad t_0 = \mu, 
$$
\noindent the dimensionless equations for $\bar{f}$ and $\bar{g}$ read (dropping the primes and tildes for the sake of clarity):
\begin{empheq}[left=\empheqlbrace]{align}
 \partial_t f - \nabla_{x}\cdot (\nabla_{x} U f) - \partial_{\theta} (\partial_{\theta} U f)
  -  \xi \partial_{\theta} F_2& - \xi\nabla_{x}\cdot F_1 - d \partial^2_{\theta} f - d  \Delta_{x} f =0, \label{systadimf}\\
 \partial_t g - \nabla_{x_1} \cdot (g \nabla_{x_1} U(x_1,\theta_1) +   \xi\frac{g}{f(x_1,\theta_1)}&F_1(x_1,\theta_1)) \nonumber\\
 - \partial_{\theta_1} (g\partial_{\theta} U(x_1,\theta_1)+  \xi \frac{g}{f(x_1,\theta_1)}F_{2}(x_1,&\theta_1))\nonumber \\
-\nabla_{x_2} \cdot (g\nabla_{x} U(x_2,\theta_2) + \xi \frac{g}{f(x_2,\theta_2)}&F_1(x_2,\theta_2))\nonumber\\
- \partial_{\theta_2}  (g\partial_{\theta} U(x_2,\theta_2)+ \xi \frac{g}{f(x_2,\theta_2)}&F_{2}(x_2,\theta_2) )\label{systadimg}\\
-d\nabla_{x_1}\cdot ( \frac{g}{f(x_1,\theta_1)}\nabla_{x_1} f(x_1,&\theta_1)) \nonumber\\
- d \nabla_{x_2}\cdot ( \frac{g}{f(x_2,\theta_2)} \nabla_{x} f(x_2,&\theta_2)) \nonumber\\
- d\partial_{\theta_1}(\frac{g}{f(x_1,\theta_1)}\partial_{\theta} f(x_1,\theta_1&)) \nonumber\\
- d\partial_{\theta_2}( \frac{g}{f(x_2,\theta_2)} \partial_{\theta} f(x_2,&\theta_2))=S(g)(x_1,\theta_1,\ell_1,x_2,\theta_2,\ell_2),\nonumber
\end{empheq}  
\noindent with 
\begin{empheq}[left=\empheqlbrace]{align*}
F_1(x_1,\theta_1)&= \int_{\mathcal{L}}  \nabla_{x_1} V(x_1,\theta_1,\ell_1,x_2,\theta_2,\ell_2) g(x_1,\theta_1,\ell_1,x_2,\theta_2,\ell_2) d\ell_1 d\ell_2 \frac{d\theta_2}{\pi} dx_2,\\
F_2(x_1,\theta_1) &= F_{al}(x_1,\theta_1) + F_{link}(x_1,\theta_1), \\
F_{link}(x_1,\theta_1) &= \int_{\mathcal{L}}   \big(g \partial_{\theta_1}V \big)(x_1,\theta_1,\ell_1,x_2,\theta_2,\ell_2) d\ell_1 d\ell_2 \frac{d\theta_2}{\pi} dx_2 ,\\ 
F_{al}(x_1,\theta_1) &= \int_{\mathcal{L}}  \big(g\partial_{\theta_1} \bar{b}\big)(x_1,\theta_1,\ell_1,x_2,\theta_2,\ell_2)d\ell_1 d\ell_2 \frac{d\theta_2}{\pi} dx_2 ,
\end{empheq}
\noindent where $\mathcal{L} = \mathbb{R}^2 \times [-\frac{\pi}{2},\frac{\pi}{2}] \times [-L/2,L/2] \times [-L/2,L/2]$ and 
 \begin{equation*}
 \begin{split}
S(g)(x_1,\theta_1,\ell_1,x_2,\theta_2,\ell_2) = &\nu_f f(x_1,\theta_1)f(x_2,\theta_2)\delta_{\ell(x_1,\theta_1,x_2,\theta_2)}(\ell_1)\delta_{\ell(x_2,\theta_2,x_1,\theta_1)}(\ell_2)\\
&- \nu_dg(x_1,\theta_1,\ell_1,x_2,\theta_2,\ell_2).
\end{split}
\end{equation*}

\subsection{Scaled equations}\label{scaling}
So far, the chosen time and space scales are microscopic ones, and describe the system at the scale of the agent interactions. In order to describe the system at a macroscopic scale, a small parameter $\varepsilon \ll 1$ is introduced and the space and time units are set to $\tilde{x_0} = \varepsilon^{-1/2} x_0$, $\tilde{t_0} = \varepsilon^{-1} t_0$. The fiber length measured at scale $x_0$ is supposed to stay of order 1 as $\varepsilon \rightarrow 0$, i.e.\  $L=O(1)$. The variables $x$, $t$, $\ell$ and unknowns $f$ and $g$ are then correspondingly changed to $\tilde{x} = \sqrt{\varepsilon} x$, $\tilde{t} = \varepsilon t$, $\tilde{\ell}=\sqrt{\varepsilon} \ell$, $\tilde{f}(\bar{x},\theta) = \varepsilon^{-1} f(x,\theta)$ and $\tilde{g}(\tilde{x}_1,\theta_1,\tilde{\ell}_1,\tilde{x}_2,\theta_2,\tilde{\ell}_2) = \varepsilon^{-3} g(x_1,\theta_1,\ell_1,x_2,\theta_2,\ell_2)$. We suppose that the external potential $U(x,\theta)$ is decomposed into $U(x,\theta) = U^0(x) + U^1(\theta)$, where $U^0$ is acting on the space variable only and $U^1$ is a $\pi$-periodic potential acting on fiber orientation angles only. The external potential acting on the space variables is supposed to be one order of magnitude stronger than the one acting on the fiber rotations: $U^0 = O(1)$, $U^1=O(\varepsilon)$, i.e.\ $\tilde{U}^1 = \varepsilon^{-1} U^1$ with $\tilde{U}^1 = O(1)$. The strength of the alignment potential is supposed to be large $\alpha = O(\varepsilon^{-1})$, i.e.\ $\tilde{\alpha} = \varepsilon \alpha$ with $\tilde{\alpha} = O(1)$, and we choose the exposant $\beta = 1$. The intensity of the alignment potential between linked fibers is supposed to be small $\kappa=O(\varepsilon)$, i.e.\ $\tilde{\kappa} = \varepsilon^{-1}\kappa$ with $\tilde{\kappa} = O(1)$ and the diffusion coefficient and parameter $\xi$ are supposed to stay of order 1: $d,\xi = O(1)$.  In order to simplify the analysis of the system, the process of linking/unlinking is supposed to occur at a very fast time scale, i.e.\ $\tilde{\nu}_f = \varepsilon^2 \nu_f$ and $\tilde{\nu}_d = \varepsilon^2 \nu_d$, with $\tilde{\nu}_f$, $\tilde{\nu}_d = O(1)$. The macroscopic restoring potential $\tilde{V}$ is defined such that:
 \begin{equation*}
\tilde{V}(\tilde{x}_1,\theta_1,\tilde{\ell}_1,\tilde{x}_2,\theta_2,\tilde{\ell}_2) = \frac{\tilde{\kappa}}{2}|\tilde{x}_1+\tilde{\ell}_1 \omega(\theta_1) - \tilde{x}_2 - \tilde{\ell}_2 \omega(\theta_2)|^2,
 \end{equation*}
\noindent Then,
\begin{empheq}[left=\empheqlbrace]{align*}
 V(x_1,\theta_1,\ell_1,x_2,\theta_2,\ell_2) = \tilde{V}(\tilde{x}_1,\theta_1,\tilde{\ell}_1,\tilde{x}_2,\theta_2,\tilde{\ell}_2),\\
\partial_{\theta} V(x_1,\theta_1,\ell_1,x_2,\theta_2,\ell_2) = \partial_{\theta} \tilde{V}(\tilde{x}_1,\theta_1,\tilde{\ell}_1,\tilde{x}_2,\theta_2,\tilde{\ell}_2),\\
\nabla_{x} V(x_1,\theta_1,\ell_1,x_2,\theta_2,\ell_2) = \sqrt{\varepsilon} \nabla_{\bar{x}} \tilde{V}(\tilde{x}_1,\theta_1,\tilde{\ell}_1,\tilde{x}_2,\theta_2,\tilde{\ell}_2).
\end{empheq}
\noindent Similarly, we have
\begin{equation*}
b(\theta_1,\theta_2) = \alpha |\sin(\theta_1 - \theta_2)| = \frac{\tilde{\alpha}}{\varepsilon}|\sin(\theta_1 - \theta_2)| = \frac{1}{\varepsilon} \tilde{b}(\theta_1,\theta_2),
\end{equation*}
\noindent and consequently, 
\begin{equation*}
\partial_{\theta_1} b(\theta_1,\theta_2) = \frac{1}{\varepsilon} \partial_{\theta_1} \tilde{b}(\theta_1,\theta_2).
\end{equation*}
\noindent Then we have:
\begin{empheq}[left=\empheqlbrace]{align*}
 \nabla_{x_1} {F}_1 &=\sqrt{\varepsilon} \nabla_{\tilde{x}_1} \bigg( \int_{L^\varepsilon} \sqrt{\varepsilon} \nabla_{\tilde{x}_1} \tilde{V}(\tilde{x}_1,\theta_1,\tilde{\ell}_1,\tilde{x}_2,\theta_2,\tilde{\ell}_2) \varepsilon^3 \tilde{g}(\tilde{x}_1,\theta_1,\tilde{\ell}_1,\tilde{x}_2,\theta_2,\tilde{\ell}_2) \\
 & \hspace{8cm} \frac{d\tilde{x}_2 \frac{d\theta_2}{\pi}d\tilde{\ell}_1 d\tilde{\ell}_2}{\varepsilon^2}\bigg)\\
 & = \varepsilon^2 \nabla_{\tilde{x}_1} \tilde{F}_1,\\
 F_{link}(x_1,\theta_1) &= \int_{L^\varepsilon} \partial_{\theta_1} \tilde{V}(\tilde{x}_1,\theta_1,\tilde{\ell}_1,\tilde{x}_2,\theta_2,\tilde{\ell}_2) \varepsilon^3 \tilde{g}(\tilde{x}_1,\theta_1,\tilde{\ell}_1,\tilde{x}_2,\theta_2,\tilde{\ell}_2) \frac{d\tilde{x}_2 \frac{d\theta_2}{\pi}d\tilde{\ell}_1 d\tilde{\ell}_2}{\varepsilon^2}\\
 & = \varepsilon \tilde{F}_{link},\\ 
F_{al}(x_1,\theta_1) &= \int_{L^\varepsilon} \frac{1}{\varepsilon}\partial_{\theta_1} \tilde{b}(\theta_1,\theta_2)\varepsilon^3\tilde{g}(\tilde{x}_1,\theta_1,\tilde{\ell}_1,\tilde{x}_2,\theta_2,\tilde{\ell}_2) \frac{d\tilde{x}_2 \frac{d\theta_2}{\pi}d\tilde{\ell}_1 d\tilde{\ell}_2}{\varepsilon^2}= \tilde{F}_{al},
\end{empheq}
\noindent where $L^\varepsilon = \mathbb{R}^2 \times [-\frac{\pi}{2},\frac{\pi}{2}] \times [-\frac{\sqrt{\varepsilon} L}{2},\frac{\sqrt{\varepsilon} L}{2}]^2$. Finally, we define $X_1$ and $X_2$ such that: 
\begin{empheq}[left=\empheqlbrace]{align*}
X_1(x_1,\theta_1) &=  \sqrt{\varepsilon} \nabla_{\tilde{x}} \tilde{U}^0(\tilde{x}_1) + \xi \varepsilon^{\frac{3}{2}}  \frac{\tilde{F}_1}{\varepsilon \tilde{f}}(\tilde{x}_1,\theta_1) = \sqrt{\varepsilon} \tilde{X}_1(\tilde{x}_1,\theta_1),\\
X_2(x_1,\theta_1) &= \frac{\varepsilon d \partial_{\theta_1} \tilde{f} +  \xi \varepsilon \tilde{F}_{link}}{\varepsilon \tilde{f}} = \tilde{X}_2(\tilde{x}_1,\theta_1),
\end{empheq}
\noindent with $\tilde{X}_1(\tilde{x}_1,\theta_1)$ and $\tilde{X}_2(\tilde{x}_1,\theta_1)$ defined by
\begin{empheq}[left=\empheqlbrace]{align*}
\tilde{X}_1(\tilde{x}_1,\theta_1) &= \nabla_{\tilde{x}} \tilde{U}^0(\tilde{x}_1) + \xi \frac{\tilde{F}_1}{\tilde{f}}(\tilde{x}_1,\theta_1),\\
 \tilde{X}_2(\tilde{x}_1,\theta_1) &= \frac{d \partial_{\theta_1} \tilde{f} +  \xi \tilde{F}_{link}}{ \tilde{f}}. 
 \end{empheq}
\noindent The macroscopic fiber linking/unlinking operator $S(\tilde{g})$ is similar to the one defined Eq.~\eqref{Sg}. Indeed, from Eq.~\eqref{tfk}: $\bar{\ell}(x_1,\theta_1,x_2,\theta_2) = \varepsilon^{-1/2} \bar{\ell}(\tilde{x}_1,\theta_1,\tilde{x}_2,\theta_2)$ and thus:
\begin{equation*}
S(\tilde{g}) = \tilde{\nu_f} \tilde{f}(\tilde{x}_1,\theta_1) \tilde{f}(\tilde{x}_2,\theta_2) \delta_{\bar{\ell}(\tilde{x}_1,\theta_1,\tilde{x}_2,\theta_2)}(\tilde{\ell}) \delta_{\bar{\ell}(\tilde{x}_1,\theta_1,\tilde{x}_2,\theta_2)}(\tilde{\ell}_2) - \tilde{\nu}_d \tilde{g}.
\end{equation*} 
\noindent Altogether, the macroscopic version of Eqs.~\eqref{systadimf}-\eqref{systadimg} reads (dropping the tildes for the sake of clarity):
\begin{equation}\label{systepsf}
\begin{split}
 -\xi \partial_{\theta_1} F_{al}- & \varepsilon \bigg( \xi \partial_{\theta} {F}_{link} + d \partial^2_{\theta} {f} \bigg) \\
&+ \varepsilon^2 \bigg( \partial_{{t}} {f} - \nabla_{{x}} \cdot (\nabla_{{x}} {U}{f}) - \partial_{\theta} (\partial_{\theta} {U}{f}) - \xi \nabla_{{x}} \cdot {F}_1  - d  \Delta_{{x}} {f} \bigg) = 0,
\end{split}
\end{equation}
\noindent and
\begin{equation}\label{systgeps}
\begin{split}
-S(g) - \varepsilon  \xi \bigg( \partial_{\theta_1} ( g  F_{al}({x}_1,\theta_1)\big) & + \partial_{\theta_2} \big(g F_{al}(x_2,\theta_2)\big)\bigg) \\
- \varepsilon^2 \bigg(  \partial_{\theta_1} ( g  X_2({x}_1,\theta_1)\big)  &+ \partial_{\theta_2} \big(g X_2(x_2,\theta_2)\big)\bigg)\\
+\varepsilon^3 \bigg( \partial_{{t}} {g} - \nabla_{{x}_1}  \cdot  (g &X_1({x}_1,\theta_1)) -\nabla_{x_2}  \cdot  (g X_1(x_2,\theta_2)) \\
- \partial_{\theta_1} \big(g \partial_{\theta} {U}^1(\theta_1)&\big) - \partial_{\theta_2}  (g \partial_{\theta} {U}^1(\theta_2)\big)\\
- d \nabla_{{x}_1} \cdot (g &\frac{\nabla_{{x}} f}{f}({x}_1,\theta_1)) - d \nabla_{x_2} \cdot (g \frac{\nabla_{{x}} f}{f}(x_2,\theta_2)) \bigg) = 0.
 \end{split}
 \end{equation}
  
From now on, we note $f^\varepsilon = \tilde{f}$ and $g^\varepsilon = \tilde{g}$. The following proposition holds: 

\begin{proposition}\label{prop1}
Assuming $f^\varepsilon$ and $g^\varepsilon$ exist, then, formally, they satisfy:
\begin{equation}\label{MF1H3}
\begin{split}
- &\xi \partial_{\theta} \bigg( \partial_{\theta} \Phi[f^\varepsilon](x,\theta) f^\varepsilon\bigg)-  d \partial^2_{\theta} f^\varepsilon \\
&+ \varepsilon \bigg[\partial_t f^\varepsilon - \nabla_{x} \cdot  (\nabla_{x} U^0 f^\varepsilon)  - \partial_{\theta}\bigg( \bigg[\partial_{\theta} U^1 + \xi G[f^\varepsilon](x,\theta) \bigg] f^\varepsilon\bigg) - d \Delta_x f^\varepsilon\bigg] = O(\varepsilon^2),
 \end{split}
 \end{equation}
\noindent and 
\begin{equation}\label{g}
\begin{split}
g^\varepsilon(x_1,\theta_1,\ell_1,x_2,\theta_2,\ell_2) = \frac{\nu_f}{\nu_d} f^\varepsilon(x_1,\theta_1)f^\varepsilon(x_2,\theta_2)\delta_{\bar{\ell}(x_1,\theta_1,x_2,\theta_2)}(\ell_1)&\delta_{\bar{\ell}(x_2,\theta_2,x_1,\theta_1)}(\ell_2)\\
& + O(\varepsilon^2),
\end{split}
\end{equation}
\noindent with
\begin{empheq}[left=\empheqlbrace]{align}
\Phi[f^\varepsilon](x_1,\theta_1) &= C_{1} \int\limits_{-\frac{\pi}{2}}^{\frac{\pi}{2}} \sin^2(\theta-\theta_2) f^\varepsilon(x_1,\theta_2) \frac{d\theta_2}{\pi} \label{forcePhi}\\
G[f^\varepsilon](x_1,\theta_1) &= C_{2} \sum_{i,j=1}^2 \frac{\partial^2}{\partial x_i \partial x_j} \int\limits_{-\frac{\pi}{2}}^{\frac{\pi}{2}} f^\varepsilon(x_1,\theta_2)B_{ij}(\theta_1,\theta_2)\frac{d\theta_2}{\pi} \label{forceG}, \\
C_1 &= \frac{\alpha L^2\nu_f}{2 \nu_d} \; , \; C_{2} = \frac{\alpha L^4\nu_f}{48 \nu_d} \label{coeffC1C2},
\end{empheq}
\noindent and
\begin{equation}\label{Bteta}
B(\theta_1,\theta_2) =  \sin 2(\theta_1 - \theta_2) [\omega(\theta_1) \otimes \omega(\theta_1) + \omega(\theta_2) \otimes \omega(\theta_2)] = \big(B_{ij}(\theta_1,\theta_2)\big)_{i,j=1,2}.
\end{equation}
\end{proposition}

\begin{remark}
 In the proof of proposition~\ref{prop1}, we will show that
\begin{empheq}[left=\empheqlbrace]{align}
F^\varepsilon_1(x_1,\theta_1) &= O(\varepsilon^3), \label{F1SG0}\\
F^\varepsilon_{link}(x_1,\theta_1&) =O(\varepsilon^3),\label{F2SG0}\\
F^\varepsilon_{al}(x_1,\theta_1)&=\varepsilon \partial_{\theta_1} \Phi[f^\varepsilon](x_1,\theta_1) + \varepsilon^2 G[f^\varepsilon](x_1,\theta_1)+ O(\varepsilon^3).\nonumber
\end{empheq}
\end{remark}

\medskip
\noindent
The proof of this proposition is given in section~\ref{sectproof}. From these equations, one notes that the hypothesis of dominant creation/deletion of links makes the reaction forces $F_1$ and $F_{link}$ of order $O(\varepsilon^3)$. In this case, the process of linking/unlinking is so fast that the constraint is satisfied at all times. Moreover, under this assumption, the first contribution of the alignment force acting on a fiber is the sum of elementary alignment forces generated by its intersecting fibers, weighted by $\frac{\nu_f}{\nu_d}$. One also notes that the alignment force $F_{al}$ is local in space.

Under these scaling assumptions, the leading order of the left-hand side of Eq.~\eqref{MF1H3} takes the form of a collision operator of kinetic theory. It acts on the orientation vector $\theta$ only and it expresses that the alignment potential~\eqref{align} is counter-balanced by the diffusion term which tends to spread the particles isotropically on the sphere.  The other terms act at lower order $\varepsilon$. 

 As the large scale limit involves an expansion of the solution around a local equilibrium, the study of the local equilibria of the collision operator are of key importance. Therefore, section~\ref{sec:LSL} will be dedicated to the study of the properties of the left-hand side of~\eqref{MF1H3}. 

\subsection{Proof of proposition~\ref{prop1}}\label{sectproof}

\begin{proof}
From Eq.~\eqref{systgeps}, one notes that the source term $S(g^\varepsilon)$ is of order $O(\varepsilon)$. Thus:
\begin{equation}\label{ginit}
\begin{split}
g^\varepsilon(x_1,\theta_1,\ell_1,x_2,\theta_2,\ell_2) = \frac{\nu_f}{\nu_d} f^\varepsilon(x_1,\theta_1)f^\varepsilon(x_2,\theta_2)\delta_{\bar{\ell}(x_1,\theta_1,x_2,\theta_2)}(\ell_1)&\delta_{\bar{\ell}(x_2,\theta_2,x_1,\theta_1)}(\ell_2) \\
&+ O(\varepsilon).
\end{split}
\end{equation}
\noindent Inserting this expression into the relations for $F^\varepsilon_1$ and $F^\varepsilon_{link}$ and $F^\varepsilon_{al}$ (see Eqs.~\eqref{F1}-\eqref{F2}), one obtains (dropping the tildes for the new variables, and denoting $\tilde{V} = \tilde{V}(x_1,\theta_1,\ell_1,x_2,\theta_2,\ell_2)$ and $b^\varepsilon = b^\varepsilon(\theta_1,\theta_2)$): 
\begin{empheq}[left=\empheqlbrace]{align}
F^\varepsilon_1 = \frac{\nu_f f^\varepsilon(x_1,\theta_1)}{ \nu_d}  \int_{L^\varepsilon} \bigg(\nabla_{x_1} \tilde{V} &f^\varepsilon(x_2,\theta_2)\delta_{\bar{\ell}(x_1,\theta_1,x_2,\theta_2)}(\ell_1)\nonumber\\
&\delta_{\bar{\ell}(x_2,\theta_2,x_1,\theta_1)}(\ell_2) 
+ O(\varepsilon) \bigg)dx_2 \frac{d\theta_2}{\pi} d\ell_1 d\ell_2, \nonumber\\
F^\varepsilon_{link} = \frac{\nu_f f^\varepsilon(x_1,\theta_1)}{\nu_d} \int_{L^\varepsilon} \bigg( \partial_{\theta_1} \tilde{V}& f^\varepsilon(x_2,\theta_2)\delta_{\bar{\ell}(x_1,\theta_1,x_2,\theta_2)}(\ell_1)\nonumber\\
&\delta_{\bar{\ell}(x_2,\theta_2,x_1,\theta_1)}(\ell_2)
+ O(\varepsilon) \bigg)dx_2 \frac{d\theta_2}{\pi} d\ell_1 d\ell_2, \label{3F}\\
F^\varepsilon_{al} = \frac{\nu_f f^\varepsilon(x_1,\theta_1)}{\nu_d} \int_{L^\varepsilon} \bigg( \partial_{\theta_1} b^\varepsilon &f^\varepsilon(x_2,\theta_2)\delta_{\bar{\ell}(x_1,\theta_1,x_2,\theta_2)}(\ell_1)\nonumber\\
&\delta_{\bar{\ell}(x_2,\theta_2,x_1,\theta_1)}(\ell_2) + O(\varepsilon) \bigg) dx_2 \frac{d\theta_2}{\pi} d\ell_1 d\ell_2\nonumber
\end{empheq}
\noindent We note that if $\phi(x_1,\theta_1,\ell_1,x_2,\theta_2,\ell_2) \in L^\infty(L^\varepsilon)$ with sufficient decay at infinity, then 
$$
\int_{L^\varepsilon} \nabla_{x_1} \tilde{V} \, \phi \, dx_2 \, \frac{d\theta_2}{\pi}d\ell_1d\ell_2 \leq \varepsilon C, 
$$ 
since the measure of $L^\varepsilon$ intersected with any compact set of $\mathbb{R}^2 \times [-\pi/2,\pi/2] \times \mathbb{R}^2$ is of order $\varepsilon$. Indeed, the domain of integration with respect to $\ell_1$ or $\ell_2$ has a measure of order $\varepsilon$. Thus, assuming that the $O(\varepsilon)$ remainder in (\ref{ginit}) is an $L^\infty$ function, which is legitimate in view of the diffusive character of (\ref{systgeps}), we get:
\begin{empheq}[left=\empheqlbrace]{align}
F^\varepsilon_1 = \bigg( \int\limits_{K^\varepsilon(x_1,\theta_1)} \nabla_{x_1} \tilde{V}&(x_1,\theta_1,\bar{\ell}(x_1,\theta_1,x_2,\theta_2),x_2,\theta_2,\bar{\ell}(x_2,\theta_2,x_1,\theta_1))\nonumber\\
& f^\varepsilon(x_2,\theta_2)) dx_2 \frac{d\theta_2}{\pi}\bigg) \frac{\nu_f f^\varepsilon(x_1,\theta_1)}{\nu_d}
+ O(\varepsilon^2),\nonumber\\
\nonumber\\
F^\varepsilon_{link} = \bigg( \int\limits_{K^\varepsilon(x_1,\theta_1)}  \partial_{\theta_1} \tilde{V}&(x_1,\theta_1,\bar{\ell}(x_1,\theta_1,x_2,\theta_2),x_2,\theta_2,\bar{\ell}(x_2,\theta_2,x_1,\theta_1))\nonumber\\
& f^\varepsilon(x_2,\theta_2) dx_2 \frac{d\theta_2}{\pi}\bigg)\frac{\nu_f f^\varepsilon(x_1,\theta_1)}{\nu_d}+ O(\varepsilon^2), \label{3Fsuite}\\
\nonumber\\
F^\varepsilon_{al} =\frac{\nu_f f^\varepsilon(x_1,\theta_1)}{\nu_d}\int\limits_{K^\varepsilon(x_1,\theta_1)}& \partial_{\theta_1} b(\theta_1,\theta_2)f^\varepsilon(x_2,\theta_2) dx_2 \frac{d\theta_2}{\pi}+ O(\varepsilon^2),\nonumber
\end{empheq}
\noindent where $K^\varepsilon(x_1,\theta_1)$ is the set of  fibers intersecting fiber in $(x,\theta)$, given by:
\begin{equation}\label{setK}
K^\varepsilon(x_1,\theta_1) = \{(x_2,\theta_2) \; | \; |\bar{\ell}(x_1,\theta_1,x_2,\theta_2)|\leq \sqrt{\varepsilon} L/2 \; , |\bar{\ell}(x_2,\theta_2,x_1,\theta_1)|\leq \sqrt{\varepsilon} L/2 \}.
\end{equation}
\noindent From the fact that $\tilde{V}$ is a quadratic function of $x_1 + \ell_1 \omega(\theta_1) - x_2 - \ell_2 \omega(\theta_2)$ and the fact that setting $\ell_1 = \bar{\ell}(x_1,\theta_1,x_2,\theta_2)$ and $\ell_2 = \bar{\ell}(x_2,\theta_2,x_1,\theta_1)$ just cancels this expression, one immediately notes that:
\begin{empheq}[left=\empheqlbrace]{align*}
\nabla_{x_1} \tilde{V}(x_1,\theta_1,\bar{\ell}(x_1,\theta_1,x_2,\theta_2),x_2,\theta_2,\bar{\ell}(x_2,\theta_2,x_1,\theta_1)) &= 0,\\
 \partial_{\theta_1} \tilde{V}(x_1,\theta_1,\bar{\ell}(x_1,\theta_1,x_2,\theta_2),x_2,\theta_2,\bar{\ell}(x_2,\theta_2,x_1,\theta_1))& =0.
 \end{empheq}
\noindent So, finally:
\begin{equation}\label{F1F2s}
F^\varepsilon_1 = O(\varepsilon^2),\quad
F^\varepsilon_{links,2}=O(\varepsilon^2). 
\end{equation}
\noindent We are left with:
\begin{equation}\label{Fals}
F^\varepsilon_{al} =  \frac{\nu_f}{\nu_d} f^\varepsilon(x_1,\theta_1)\int\limits_{K^\varepsilon(x_1,\theta_1)} \partial_{\theta_1} b(\theta_1,\theta_2) f^\varepsilon(x_2,\theta_2) dx_2 \frac{d\theta_2}{\pi} + O(\varepsilon^2).
\end{equation}
\noindent From now on, we write $\omega_1 = \omega(\theta_1)$ and $\omega_2 = \omega(\theta_2)$. By the change of variables $x_2 \mapsto (s_1,s_2)$ defined by
$$
x_2=x_1+\frac{\sqrt{\varepsilon} L}{2}s_1\omega_1 - \frac{\sqrt{\varepsilon} L}{2} s_2 \omega_2, 
$$ 
with associated Jacobian 
\begin{equation*}
J_{x_2} = \frac{L \sqrt{\varepsilon}}{2}\begin{pmatrix}
 \cos\theta_1 & -\cos\theta_2  \\ \sin \theta_1 & -\sin \theta_2 \end{pmatrix},
 \end{equation*}
 \noindent and $|\det(J_{x_2})| = \frac{L^2 \varepsilon}{4} |\sin(\theta_1-\theta_2)|$, we have: 
\begin{equation*}
\begin{split}
\displaystyle  F^\varepsilon_{al}(x_1,\theta_1) = \varepsilon C(x_1,\theta_1)\int\limits_{-\frac{\pi}{2}}^{\frac{\pi}{2}} \int\limits_{|s_1|,|s_2|\leq 1}  |\sin(\theta_1-\theta_2)|  \partial_{\theta_1} b(\theta_1,\theta_2)&\\
 f^\varepsilon(x_1 + \frac{\sqrt{\varepsilon} L}{2}s_1\omega_1 - \frac{\sqrt{\varepsilon} L}{2}&s_2\omega_2,\theta_2)
ds_1 ds_2 \frac{d\theta_2}{\pi}\\
& + O(\varepsilon^2),
\end{split}
\end{equation*}
\noindent where $C(x_1,\theta_1) = \frac{L^2 \nu_f f^\varepsilon(x_1,\theta_1)}{4\nu_d} $. Thanks to~\eqref{b} with $\beta=1$, one notes that\\
$\partial_{\theta_1} b(\theta_1, \theta_2) = \alpha \partial_{\theta_1} |\sin(\theta_1-\theta_2)| $, and then, $|\sin(\theta_1 - \theta_2)| \partial_{\theta_1} b(\theta_1, \theta_2) =  \frac{\alpha}{2}\partial_{\theta_1} \sin^2(\theta_1-\theta_2)$. Then,
\begin{equation}\label{F2taylor}
\begin{split}
\displaystyle  F^\varepsilon_{al}(x_1,\theta_1) =   \frac{\varepsilon\alpha}{2} C(x_1,\theta_1) \int\limits_{-\frac{\pi}{2}}^{\frac{\pi}{2}} \int\limits_{|s_1|,|s_2|\leq 1} \partial_{\theta_1} \sin^2(\theta_1-\theta_2) &\\ 
 f^\varepsilon(x_1 + \frac{\sqrt{\varepsilon} L}{2}s_1\omega_1 - \frac{\sqrt{\varepsilon} L}{2}&s_2\omega_2,\theta_2)
 ds_1 ds_2 \frac{d\theta_2}{\pi}\\
 &+O(\varepsilon^2). 
\end{split}
\end{equation}
\noindent By Taylor expansion, we have:
\begin{equation*}
\begin{split}
f^\varepsilon(x_1+ \frac{\sqrt{\varepsilon} L}{2} s_1\omega_1 - \frac{\sqrt{\varepsilon} L}{2} &s_2\omega_2, \theta_2) = f^\varepsilon(x_1,\theta_2) + \frac{\sqrt{\varepsilon}L}{2} \nabla_{x} f^\varepsilon(x_1, \theta_2).(s_1\omega_1-s_2\omega_2) \\
+ \frac{\varepsilon L^2}{4} &(s_1\omega_1-s_2\omega_2)^T \nabla_{x}^2 f^\varepsilon(x_1, \theta_2) (s_1\omega_1-s_2\omega_2) \\
&+ O((\frac{\sqrt{\varepsilon}L}{2}|s_1\omega_1-s_2\omega_2|)^3),
\end{split}
\end{equation*}
\noindent where $\nabla_{x}^2 f^\varepsilon$ is the spatial-hessian matrix of $f^\varepsilon$ ($(\nabla_{x}^2 f)_{ij} = \frac{\partial^2 f}{\partial x_i \partial x_j}$), and for any vector $a$ of $\mathbb{R}^2$ and any $2 \times 2$ matrix B : $a^T B a = \sum_{(i,j)\in [1,2]^2} B_{ij} a_j a_i$. Integrating over $s_1,s_2 \in [-1,1]$, the odd terms with respect to either $s_1$ or $s_2$ vanish. Therefore: 
\begin{equation}\label{taylor1}
\begin{split}
\displaystyle \int\limits_{-\frac{\pi}{2}}^{\frac{\pi}{2}}& \int\limits_{|s_1|,|s_2| \leq 1} \partial_{\theta_1} \sin^2(\theta_1-\theta_2) f^\varepsilon(x_1+ \frac{\sqrt{\varepsilon}L}{2} s_1\omega_1 - \frac{\sqrt{\varepsilon}L}{2} s_2\omega_2, \theta_2) ds_1ds_2 \frac{d\theta_2}{\pi}\\
= 4 &\int\limits_{-\frac{\pi}{2}}^{\frac{\pi}{2}}  \partial_{\theta_1} \sin^2(\theta_1-\theta_2) f^\varepsilon(x_1,\theta_2) \frac{d\theta_2}{\pi} \\ 
&+ \frac{\varepsilon L^2}{6} \int\limits_{-\frac{\pi}{2}}^{\frac{\pi}{2}}  \partial_{\theta_1} \sin^2(\theta_1-\theta_2) \nabla_{x}^2{f^\varepsilon}(x_1,\theta_2):[\omega_1 \otimes \omega_1 + \omega_2 \otimes \omega_2] \frac{d\theta_2}{\pi} + O(\varepsilon^2), 
\end{split}
\end{equation}
\noindent where $\forall A,B \in \mathbb{R}^2\; ,\;  A:B = \sum_{i,j \in [1,2]} A_{ij} B_{ij}$ and for any vectors $\omega,\omega' \in \mathbb{R}^2$, we write $(\omega \otimes \omega')_{ij} = \omega_{i} \omega'_{j} $. Then: 
\begin{equation}\label{taylor2}
\begin{split}
\int\limits_{-\frac{\pi}{2}}^{\frac{\pi}{2}}  \partial_{\theta_1} \sin^2(\theta_1-\theta_2) \nabla_{x_1}^2 &f^\varepsilon(x_1,\theta_2):[\omega_1 \otimes \omega_1 + \omega_2 \otimes \omega_2]  \frac{d\theta_2}{\pi}\\
& = \sum_{(i,j)=1}^2 \frac{\partial^2}{\partial x_i x_j} \int\limits_{-\frac{\pi}{2}}^{\frac{\pi}{2}} f^\varepsilon(x_1,\theta_2) B_{ij}(\theta_1,\theta_2) \frac{d\theta_2}{\pi}, 
\end{split}
\end{equation} 
\noindent where:
\begin{equation*}
B_{ij}(\theta_1,\theta_2) = [\omega_i(\theta_1) \omega_j(\theta_1) + \omega_i(\theta_2)\omega_j(\theta_2)]\sin(2(\theta_1-\theta_2)).
\end{equation*}
\noindent  A first consequence of what precedes is that $F^\varepsilon_{al} = O(\varepsilon)$. Therefore, $S(g^\varepsilon) = O(\varepsilon^2)$ (instead of formally $O(\varepsilon)$ as seen from Eq.~\eqref{systgeps}). As a consequence, the remainder in~\eqref{ginit} is $O(\varepsilon^2)$ instead of being $O(\varepsilon)$, and the same is true for the remainders in~\eqref{3F}. Consequently, the remainders in~\eqref{3Fsuite} are $O(\varepsilon^3)$ instead of being $O(\varepsilon^2)$ as before. It follows that the remainders in~\eqref{F1F2s}-\eqref{Fals} are $O(\varepsilon^3)$ as well. Then, inserting~\eqref{taylor1} and~\eqref{taylor2} into~\eqref{Fals} (with remainder $O(\varepsilon^2)$), we get~\eqref{g}-\eqref{MF1H3}, which ends the proof.
\end{proof}

From now on, we focus on Eq.~\eqref{MF1H3} in which we neglect the $O(\varepsilon^2)$ terms, namely
\begin{equation}\label{MF1H3b}
\begin{split}
- \xi \partial_{\theta} \bigg(& \partial_{\theta} \Phi[f^\varepsilon](x,\theta) f^\varepsilon\bigg)-  d \partial^2_{\theta} f^\varepsilon \\
&+ \varepsilon \bigg[\partial_t f^\varepsilon - \nabla_{x} \cdot  (\nabla_{x} U^0 f^\varepsilon)  - \partial_{\theta}\bigg( \bigg[\partial_{\theta} U^1 + \xi G[f^\varepsilon](x,\theta) \bigg] f^\varepsilon\bigg) - d \Delta_x f^\varepsilon\bigg] = 0,
 \end{split}
 \end{equation}
\noindent where $\Phi$ and $G$ are given by~\eqref{forcePhi}-\eqref{forceG} respectively, and we investigate the limit $\varepsilon \rightarrow 0$. This is the object of the next section.

\setcounter{equation}{0}
\section{Large scale limit}
\label{sec:LSL}

In this section, the limit $\varepsilon \rightarrow 0$ of the solution $f^\varepsilon$ to~\eqref{MF1H3b} is explored. For this purpose, Eq.~\eqref{MF1H3b} is rewritten
\begin{equation}\label{pbf}
\partial_t f^\varepsilon - \nabla_x  \cdot  (\nabla_xU^0 f^\varepsilon) - \partial_\theta \big((\partial_\theta U^1 +  \xi G[f^\varepsilon])f^\varepsilon\big) - d \Delta_x f^\varepsilon = \frac{1}{\varepsilon} Q(f^{\varepsilon}),
\end{equation} 
\noindent where the collision operator $Q(f^\varepsilon)$ is defined by
\begin{align}
Q(f) = d\partial^2_{\theta} f + \xi \partial_{\theta} (\partial_{\theta} \Phi[f])f), \label{Q}\\
\Phi[f] = C_1 \int_{-\frac{\pi}{2}}^{\frac{\pi}{2}} \sin^2(\theta - \theta_2) f \frac{d\theta_2}{\pi}, \label{Phi}
\end{align}
 \noindent and where we recall that $C_1$ and $G[f]$ are defined by~\eqref{forceG} and~\eqref{coeffC1C2} respectively. The operator $Q$ is a non linear operator on $f$ which acts on $\theta$ only and leaves $x$ and $t$ as parameters. For each function $\Phi(\theta)$, we define  $M_\Phi(\theta)$ by: 
\begin{equation}\label{Mphi}
M_{\Phi}(\theta) = \frac{1}{Z} e^{- \xi \Phi(\theta)/d},
\end{equation}
\noindent where $Z$ is a normalization factor such that $Z = \int_{-\frac{\pi}{2}}^{\frac{\pi}{2}} e^{- \xi \Phi(\theta)/d} \frac{d\theta}{\pi}$. Thus, $M_{\Phi}(\theta)$ is a probability distribution of $\theta$. Such functions are called generalized Von Mises distributions (the Von Mises distribution being the case of $\Phi(\theta) = -\cos \theta$). The next section is devoted to the analysis of the properties of $Q(f)$ and follows closely Ref. \cite{Degond_etal_JSP13}.
\subsection{Properties of $Q$ }
\subsubsection{Equilibria}
In this section, the equilibria of the operator $Q$ are studied, and the following proposition is proven:
\begin{proposition}\label{prop2}
Here, we restrict ourselves to functions of $\theta$ only.

\noindent (i) The operator $Q$ can be written: 
\begin{equation}\label{qf}
Q(f) = d \partial_{\theta} \bigg(M_{\Phi[f]} \partial_{\theta} (\frac{f}{M_{\Phi[f]}})\bigg).
\end{equation}

\noindent (ii) The equilibrium solutions of $Q$, i.e.\ the functions $f$ such that $Q(f) = 0$ are of the form $f(\theta) = \rho M_{\Phi[f]}$, where $M_{\Phi[f]}$ is defined by Eq.~\eqref{Mphi} and $\rho$ is a positive constant. 
\end{proposition}

This proposition shows that the equilibria of operator $Q$ are generalized Von Mises distributions of $\theta$, weighted by the particle density.

\begin{proof} To prove (i), one can note that:   
\begin{equation*}
\begin{split}
d \partial_{\theta} \bigg(M_{\Phi[f]} \partial_{\theta} (\frac{f}{M_{\Phi[f]}})\bigg) = d \partial_{\theta} \bigg(\partial_{\theta} f - f \partial_{\theta} (\log(M_{\Phi[f]}))\bigg) &= \partial_\theta \bigg(d \partial_{\theta} f + \xi \partial_{\theta} \Phi[f] f\bigg) \\
&= Q(f). 
\end{split}
\end{equation*}
\noindent To prove (ii), note that $f=\rho M_{\Phi[f]}$ is solution of~\eqref{qf}. Conversely, suppose that $f$ is such that 
$$d \partial_{\theta}\bigg(M_{\Phi[f]} \partial_{\theta} (\frac{f}{M_{\Phi[f]}})\bigg)=0. $$ 
We define the sets $H_f$ and $V_f$ by:
\begin{equation*}
 H_f = \{\phi \text{ measurable on } [-\frac{\pi}{2},\frac{\pi}{2}] \,  \, | \, \int_{-\frac{\pi}{2}}^{\frac{\pi}{2}} \bigg|\frac{\phi}{M_{\Phi[f]}}\bigg|^2 M_{\Phi[f]} \frac{d\theta}{\pi}< +\infty \},
 \end{equation*}
 \noindent and
 \begin{equation*}
  V_f = \{\phi \in H \, | \,  \int_{-\frac{\pi}{2}}^{\frac{\pi}{2}} \bigg|\partial_\theta (\frac{\phi}{M_{\Phi[f]}})\bigg|^2 M_{\Phi[f]} \frac{d\theta}{\pi}< +\infty \}.
  \end{equation*}
  \noindent The norms  $\| \cdot \|_{H_f}$, $\| \cdot \|_{V_f}$ on $H_f$ and $V_f$ are then defined such that: 
\begin{equation*}
\|\phi\|^2_{V_f} = \|\phi\|^2_{H_f} + |\phi|^2_{V_f}.
\end{equation*}
\noindent where 
\begin{equation*}
\|\phi\|_{H_f} = \int_{-\pi/2}^{\pi/2} \bigg|\frac{\phi}{M_{\Phi[f]}}\bigg|^2 M_{\Phi[f]} \frac{d\theta}{\pi},
\end{equation*}
\noindent and
\begin{equation*}
|\phi|_{V_f} = \int_{-\pi/2}^{\pi/2} \bigg|\partial_\theta (\frac{\phi}{M_{\Phi[f]}})\bigg|^2 M_{\Phi[f]} \frac{d\theta}{\pi}.
\end{equation*}
\noindent For $f \in V_f$ using Green's formula, we get: 
\begin{equation*}
\int_{-\pi/2}^{\pi/2} d \partial_{\theta}\bigg(M_{\Phi[f]} \partial_{\theta} (\frac{f}{M_{\Phi[f]}})\bigg) \frac{f}{M_{\Phi[f]}} \frac{d\theta}{\pi}= -d \int_{-\pi/2}^{\pi/2} M_{\Phi[f]} \bigg|\partial_{\theta} (\frac{f}{M_{\Phi[f]}})\bigg|^2 \frac{d\theta}{\pi}= 0,\\
\end{equation*}
\noindent and thus, $\partial_{\theta} (\frac{f}{M_{\Phi[f]}}) = 0$. Then, $f = \rho M_{\Phi[f]}$, with $\rho > 0$, which ends the proof.
\end{proof}

Now, the following lemma is proven:
\begin{lemma}\label{lem2}
For any function $f(\theta)$, the potential function $\Phi[f](\theta)$ of Eq.~\eqref{Phi} can be written: 
\begin{equation}\label{phi}
\displaystyle\Phi[f](\theta) = C - \frac{C_1}{2} \eta_f \cos 2(\theta - \theta_f),
\end{equation}
\noindent where $C_1$ is given by~\eqref{coeffC1C2}, $C=\frac{C_1 \rho_f}{2}$, $\rho_f = \int_{-\pi/2}^{\pi/2} f \frac{d\theta}{\pi}$ and $(\eta_f, \theta_f) \in \mathbb{R}^+ \times [-\frac{\pi}{2},\frac{\pi}{2})$ are uniquely defined by:
\begin{equation*}
\begin{split}
\eta_f \begin{pmatrix}
\cos 2\theta_f \\
\sin 2\theta_f 
\end{pmatrix} = \int_{-\frac{\pi}{2}}^{\frac{\pi}{2}}  \begin{pmatrix}
\cos 2\theta' \\ \sin 2\theta' \end{pmatrix} f(\theta') \frac{d\theta'}{\pi}, 
\end{split}
\end{equation*}
\noindent or equivalently by:
\begin{equation}\label{rftetaf}
  \int_{-\frac{\pi}{2}}^{\frac{\pi}{2}} \cos 2(\theta' - \theta_f) f(\theta')  \frac{d\theta'}{\pi} = \eta_f \; , \hspace{1cm} \int_{-\frac{\pi}{2}}^{\frac{\pi}{2}} \sin 2(\theta' - \theta_f) f(\theta')  \frac{d\theta'}{\pi} = 0.
\end{equation}
\noindent Remark that the second condition is equivalent to saying that
\begin{equation*}
\theta_f = \frac{1}{2} \tan^{-1} \big(\frac{\int \sin 2\theta' f(\theta') d\theta'}{\int \cos 2\theta' f(\theta') d\theta'}\big),
\end{equation*}
\noindent and this defines $\theta_f$ uniquely modulo $\pi$.
\end{lemma}
\begin{proof}
As $ \sin^2(\theta - \theta') = \frac{1}{2}(1 - \cos 2\theta\cos 2\theta' - \sin 2\theta \sin 2\theta')$, $\Phi[f]$ can be decomposed into: 
\begin{equation*}
\begin{split}
\Phi[f](\theta) &= C_1 \int_{-\frac{\pi}{2}}^{\frac{\pi}{2}} \sin^2(\theta - \theta') f(\theta') \frac{d\theta'}{\pi} \\
= \frac{C_1 }{2}&\bigg(\int_{-\frac{\pi}{2}}^{\frac{\pi}{2}} f(\theta') \frac{d\theta'}{\pi} - \cos 2\theta\int_{-\frac{\pi}{2}}^{\frac{\pi}{2}} \cos 2\theta' f(\theta') \frac{d\theta'}{\pi} - \sin 2\theta\int_{-\frac{\pi}{2}}^{\frac{\pi}{2}} \sin 2\theta' f(\theta') \frac{d\theta'}{\pi}\bigg)\\
= \frac{C_1 }{2}&\bigg(\rho - \eta_f\cos 2(\theta-\theta_f) \bigg),
\end{split}
\end{equation*}
\noindent The result follows.
\end{proof}

Let us now suppose that $\frac{\nu_f}{\nu_d}$ depends on  $\eta_f$:

\begin{hypothesis}\label{hyp2}
The parameter $\frac{\nu_f}{\nu_d}$ is supposed to be inversely proportional to the local fiber density: $\frac{\nu_f}{\nu_d} = \frac{\gamma}{\eta_f}$, with $\gamma$ a constant.
\end{hypothesis}

\medskip
Note that, thanks to Hypothesis~\ref{hyp2}, we have
\begin{equation}\label{r}
\frac{\xi C_1 \eta_f}{2d} = \alpha \frac{\xi  L^2 \nu_f}{2\nu_d}\eta_f \frac{1}{2d} = \frac{\xi \alpha L^2 \gamma}{4d} = r,
\end{equation}
\noindent where $r$ is a constant depending only on the data of the problem. 
\begin{proposition}\label{prop3}
Here, we restrict ourselves to functions of $\theta$ only. Under Hypothesis~\ref{hyp2}, the equilibrium solutions of $Q$, i.e.\ the functions $f_{eq}$ such that $Q(f_{eq}) = 0$ are of the form: 
\begin{equation}\label{feq}
f_{eq}(\theta) = \rho M_{\theta_0}(\theta),
\end{equation}
\noindent for arbitrary $\rho \in [0,\infty)$ and $\theta_0 \in [-\frac{\pi}{2},\frac{\pi}{2})$ and where:
\begin{empheq}[left=\empheqlbrace]{align}
M_{\theta_0} &= \frac{e^{r \cos 2(\theta-\theta_0)}}{Z},\label{Mr}\\
Z &= Z(r) = \int_{-\frac{\pi}{2}}^{\frac{\pi}{2}} e^{r \cos 2(\theta-\theta_0)} \frac{d\theta}{\pi},\nonumber
\end{empheq}
\noindent with  $r$ given by \eqref{r}. We have $\eta_f = \rho c(r)$ with 
\begin{equation}\label{cr}
c(r)= \frac{\int_{-\frac{\pi}{2}}^{\frac{\pi}{2}} \cos 2 \theta e^{r\cos 2 \theta } \frac{d\theta}{\pi}}{\int_{-\frac{\pi}{2}}^{\frac{\pi}{2}} e^{r\cos 2 \theta } \frac{d\theta}{\pi}}.
\end{equation}
\end{proposition}

\noindent
Proposition~\ref{prop3} gives a precise description of the equilibria of $Q$, in terms of classical von Mises-Fisher distributions.  
 
\begin{proof}[Proof of proposition~\ref{prop3}]
From Proposition~\ref{prop2}, the equilibria of the collision operator $Q(f)$ are of the form 
$$
f = \rho \frac{e^{-\xi\frac{\Phi[f](\theta)}{d}}}{\int_{-\frac{\pi}{2}}^{\frac{\pi}{2}} e^{-\xi\frac{\Phi[f](\theta)}{d}} \frac{d\theta}{\pi}} .
$$ 
Thanks to Eq.~\eqref{coeffC1C2}, Lemma~\ref{lem2},Eqs.~\eqref{r} and~\eqref{cr}, we get:  
\begin{equation}\label{f}
\begin{split}
f(\theta) &= \rho \frac {e^{-\frac{\xi C}{d}+\frac{\xi C_1}{2d}\eta_f\cos 2(\theta-\theta_f)}}{\int_{-\frac{\pi}{2}}^{\frac{\pi}{2}} e^{-\frac{\xi C}{d}+\frac{\xi C_1}{2d} \eta_f\cos 2(\theta'-\theta_f)} \frac{d\theta'}{\pi}} = \rho(x) \frac { e^{r\cos 2(\theta-\theta_f)}}{\int_{-\frac{\pi}{2}}^{\frac{\pi}{2}} e^{r\cos 2(\theta'-\theta_f)} \frac{d\theta'}{\pi}},
\end{split}
\end{equation}
\noindent where $(\eta_f,\theta_f)\in \mathbb{R}^+ \times [-\frac{\pi}{2},\frac{\pi}{2})$ satisfy Eq.~\eqref{rftetaf}. Therefore, $f$ is of the form~\eqref{feq} with $r = \frac{\xi C_1 \eta_f}{2d}$. By Hypothesis~\ref{hyp2} and~\eqref{coeffC1C2}, $r = \frac{\xi \alpha L^2 \gamma}{4d}$. Conversely, let $f$ be given by~\eqref{feq}. Then, by~\eqref{phi},and~\eqref{r}, $\phi[f] = C - r\frac{d}{\xi}\cos 2 (\theta- \theta_f)$ with $\theta_f$ uniquely determined by $\int_{-\pi/2}^{\pi/2} \sin 2(\theta - \theta_f) f(\theta) \frac{d\theta}{\pi} = 0$. But $\int_{-\pi/2}^{\pi/2} \sin 2(\theta - \theta_0) f(\theta) \frac{d\theta}{\pi} = 0$ by symmetry, showing that $\theta_f = \theta_0 \; \mbox{mod}(\pi)$. Therefore, $M_{\phi[f]} = M_{\theta_0}$ and $f = \rho M_{\phi[f]}$ showing that $f$ is an equilibrium, which ends the proof.
\end{proof}

Thanks to Eq.~\eqref{cr}, Hypothesis~\ref{hyp2} amounts to supposing that the ratio $\frac{\nu_f}{\nu_d}$ is inversely proportional to the fiber density. 

Since there is no obvious conservation relation other than the conservation of the local fiber density, the only collision invariants in this model are the constants. The integration of equation~\eqref{MF1H3} against these invariants does not allow us to find the evolution equation for the mean orientation. In order to obtain an equation on $\theta_0$, inspired from Ref. \cite{Degond_Motsch_M3AS08}, the concept of Generalized Collision Invariants (GCI), i.e.\ of collision invariants when acting on a restricted subset of functions $f$, is introduced.

\subsubsection{Collision invariant}
A collision invariant is a function $\Psi$ such that for all function $f$ of $\theta$, $\int Q(f) \Psi d\theta = 0$. However, due to the lack of momentum conservation, the only collision invariants are the constants. This is not enough to determine both $\rho$ and $\theta_0$. To this aim, following Refs.~\cite{Frouvelle_M3AS12} and~\cite{Degond_Motsch_M3AS08}, we introduce the notion of GCI. For any $\theta_0 \in [-\frac{\pi}{2} \frac{\pi}{2})$, we define $L_{\theta_0}$ as the following linear operator:
\begin{equation*}
L_{\theta_0} f = d \partial_\theta \bigg(M_{\theta_0} \partial_{\theta}(\frac{f}{M_{\theta_0}})\bigg).
\end{equation*}
\noindent Note that $Q(f) = L_{\theta_f} f$ where $\theta_f$ satisfies Eq.~\eqref{rftetaf}. 
\begin{definition}
For a given $\theta_0 \in [-\frac{\pi}{2}, \frac{\pi}{2})$ a GCI associated to $\theta_0$ is a function $\Psi$ such that:
\begin{equation}\label{GCIset}
 \int_{-\frac{\pi}{2}}^{\frac{\pi}{2}} L_{\theta_0} f \Psi \frac{d\theta}{\pi} = 0 \; \;  \forall f \text{such that} \; \theta_f = \theta_0 \; \text{mod}(\pi).
\end{equation}
\noindent The set of the GCI associated to a given $\theta_0 \in [-\frac{\pi}{2},\frac{\pi}{2})$ is a linear space denoted by $\mathcal{G}_{\theta_0}$.
\end{definition}
\begin{lemma}
$\Psi \in \mathcal{G}_{\theta_0}$  if and only if $\exists \beta \in \mathbb{R}$ such that:
\begin{equation}\label{adjointL}
L^*_{\theta_0} \Psi = \beta \sin 2(\theta - \theta_0),
\end{equation}
\noindent where $L^*_{\theta_0}$ is the $L^2$ formal adjoint of $L_{\theta_0}$, i.e.\
\begin{equation*}
L^*_{\theta_0} \Psi = -\frac{d}{M_{\theta_0}} \partial_\theta \bigg( M_{\theta_0} \partial_{\theta} \Psi \bigg).
\end{equation*}
\end{lemma}

\begin{proof}
By~\eqref{rftetaf}, the condition $\theta_f = \theta_0 \; \mbox{mod}(\pi)$ is equivalent to the linear constraint: 
\begin{equation*}
\int_{-\frac{\pi}{2}}^{\frac{\pi}{2}} f \sin 2(\theta-\theta_0) \frac{d\theta}{\pi} =0.
\end{equation*}
\noindent By a classical duality argument \cite{Degond_Motsch_M3AS08}, we deduce that $\Psi \in \mathcal{G}_{\theta_0}$ if and only if:
\begin{equation*}
\exists \beta \in \mathbb{R} \text{  such that } \; \int_{-\frac{\pi}{2}}^{\frac{\pi}{2}} L_{\theta_0} f \Psi \frac{d\theta}{\pi}= \beta\int_{-\frac{\pi}{2}}^{\frac{\pi}{2}} f \sin 2(\theta - \theta_0) \frac{d\theta}{\pi} \; \forall f. 
\end{equation*}
\noindent Note that now, there are no more constraints on $f$. Therefore, we can eliminate $f$ and get~\eqref{adjointL}.
\end{proof}
\begin{proposition}\label{prop4}
Any GCI $\Psi_{\theta_0}$ associated to $\theta_0$ can be written: 
\begin{equation}
\Psi_{\theta_0}(\theta) = C+ \beta g(\theta-\theta_0), 
\end{equation}
\noindent with arbitrary $C$, $\beta \in \mathbb{R}$ and  with $g$ an odd $\pi$ periodic function belonging to $H^1_0(0,\frac{\pi}{2})$, whose expression is: 
\begin{equation}
g(\theta) = \frac{1}{2r}\bigg( \theta - \frac{\pi}{2} \frac{\int_0^\theta e^{- r\cos 2\theta'}\frac{d\theta'}{\pi}}{\int_0^{\frac{\pi}{2}} e^{-r\cos 2\theta'} \frac{d\theta'}{\pi}}\bigg).
\end{equation}
\end{proposition}

\begin{proof}
Following Refs.~\cite{Frouvelle_M3AS12},~\cite{Degond_Motsch_M3AS08}, using Lax-Milgram's theorem and Poincar\'e's inequality, it is easy to show that the problem $L^*_{\theta_0}(\Psi)  = \frac{d}{\xi} \beta \sin 2(\theta-\theta_0) $ has a unique solution in the space $\dot{H}^1(-\frac{\pi}{2},\frac{\pi}{2})$ of functions $H^1(-\frac{\pi}{2},\frac{\pi}{2})$ with zero mean. Then, the change of variables  $\theta' = \theta - \theta_0$ is performed, and functions of the form $\Psi(\theta) = \beta g(\theta)$ with $g$ odd are searched. Then, $\Psi \in \dot{H}^1([-\frac{\pi}{2},\frac{\pi}{2}])$ if and only if g belongs to $H_0^1(0,\frac{\pi}{2})$. Straightforward computations show that $\Psi$ is a solution of~\eqref{adjointL} if and only if $g$ is a solution of
\begin{equation}\label{dPsi}
(M_0 g')' = -\sin 2\theta M_0.
\end{equation}
\noindent As $M_0(\theta) = \frac{e^{r\cos 2\theta}}{Z}$ and as we search for $g \in H^1_0(0,\frac{\pi}{2})$, an analytic expression for $g$ can be found. Indeed, since $-\sin 2\theta M_0 = \frac{1}{2r} M_0$, integrating~\eqref{dPsi} with respect to $\theta$ once, we get:
\begin{equation*}
 g'(\theta) = \frac{1}{2r} + CZ e^{-r\cos 2\theta} ,
\end{equation*}
\noindent for an appropriate constant $C$. Then, since $g\in H^1_0(0,\frac{\pi}{2})$,
\begin{equation*}
 g(\theta) = \frac{\theta}{2r} + CZ \int_0^\theta e^{-r\cos 2\theta'} d\theta'.
\end{equation*}
\noindent Finally, as $g\in H^1_0(0,\pi)$, $g(0)=g(\pi)=0$ and $C$ can be determined: 
\begin{equation*}
 C = -\frac{\pi}{4rZ \int_0^{\frac{\pi}{2}} e^{-r\cos 2\theta'} d\theta'} =-\frac{1}{2 rZ \int_{-\pi/2}^{\pi/2} e^{-r\cos 2\theta'} d\theta'} = -\frac{1}{2r Z^2}.
\end{equation*}
\noindent Indeed, we have:
\begin{equation*}
\int_{-\pi/2}^{\pi/2} e^{-r\cos 2\theta'} d\theta' = \int_{-\pi/2}^{\pi/2} e^{r\cos 2\theta'} d\theta',
\end{equation*}
\noindent by the change of variable $\theta \rightarrow \frac{\pi}{2} - \theta$ for $\theta>0$ and $\theta \rightarrow \frac{-\pi}{2} - \theta$ for $\theta<0$. This yields the result. For further usage, we note that 
\begin{equation}\label{dg2}
g'(\theta) = \frac{1}{2r}(1 - \frac{1}{M_0 Z^2}).
\end{equation}
\end{proof}

\subsection{Limit $\varepsilon \rightarrow 0$}
In this section, the formal limit $\varepsilon \rightarrow 0$ of Eq.~\eqref{MF1H3} is studied. We aim to prove the following theorem:

\begin{theorem}\label{thm4}
Under the scaling~\ref{scaling} and~\ref{hyp2}, the solution $f^\varepsilon$ of eq.~\eqref{pbf} formally converges to $f(x,\theta,t)$ given by 
\begin{equation}\label{flimit}
f(x,\theta,t) = \rho(x,t) M_{\theta_0(x,t)}(\theta),
\end{equation}
\noindent where $M_{\theta_0}$ is given by~\eqref{Mr} and $\rho(x,t)$ and $\theta_0(x,t)$ satisfy the following system:
\begin{equation}\label{Eqrho}
\partial_t \rho - \nabla_x \cdot  (\nabla_x U^0 \rho) - d \Delta_x \rho=0,
 \end{equation}
 \noindent and
\begin{equation}\label{Eqtheta}
\begin{split}
\rho \partial_t \theta_0& - \rho \nabla_x U^0 \cdot \nabla_x \theta_0 - 2 \alpha_2 \nabla_x\rho \cdot \nabla_x \theta_0 -\alpha_2 \rho\Delta_x \theta_0 \\
+ & \alpha_3 (\rho \nabla_x^2\theta_0 +\nabla_x\theta_0 \otimes \nabla_x \rho + \nabla_x \rho \otimes \nabla_x \theta_0) : [\omega_0\otimes \omega_0 - \omega_0^\perp \otimes \omega_0^\perp]\\
& + \big(2 \rho \alpha_3 \nabla_x\theta_0 \otimes \nabla_x \theta_0 - \alpha_4 \nabla_x^2 \rho\big):[\omega_0\otimes \omega_0^\perp + \omega_0^\perp \otimes \omega_0] + \alpha_5 \rho \langle \partial_\theta U^1\rangle  = 0,
\end{split}
\end{equation}
\noindent where $\langle h \rangle = \int\limits_{-\pi/2}^{\pi/2} h(\theta) M_{\theta_0}(\theta) \frac{d\theta}{\pi}$ for any function $h$ of $\theta \in [-\frac{\pi}{2}, \frac{\pi}{2})$, and where the coefficients $\alpha_2,\alpha_3,\alpha_4,\alpha_5$ are given by:
\begin{empheq}[left=\empheqlbrace]{align}
\alpha_2&=\frac{d}{\alpha_1} (\alpha_1+\frac{\xi \alpha L^4 \gamma c(r)}{24 d}), \nonumber\\
\alpha_3&=\frac{\xi \alpha L^4 \gamma}{24\alpha_1} (\frac{1}{4Z^2} -1+\frac{6d c(r)}{\xi \alpha L^2\gamma}), \nonumber\\
\alpha_4 &= \frac{\xi \alpha L^4 \gamma}{192 Z^2 \alpha_1},\label{coeffalpha}\\
\alpha_5 &= \frac{1}{\alpha_1},\nonumber
\end{empheq}
\noindent with $\alpha_1$ given by:
\begin{equation}
\alpha_1=1 - \frac{1}{Z^2}.
\end{equation}
\end{theorem}

\begin{proof}
Suppose that all the functions are as smooth as needed and that all convergences are as strong as needed. In the limit $\varepsilon \rightarrow 0$, let $f^{\varepsilon} \rightarrow f$. As $Q(f^\varepsilon) = O(\varepsilon)$, then $Q(f) = 0$. By proposition~\ref{prop4}, we deduce that $f$ is given by \eqref{flimit} with $\rho \geq 0$ and $\theta_0 \in [-\frac{\pi}{2},\frac{\pi}{2})$ to be determined. In order to find the equations for $\rho$ and $\theta_0$, we use the set of GCI given by Prop.~\ref{prop3}. 

\paragraph{Equation for $\rho$}
The use of the constant GCI amounts to integrating Eq.~\eqref{pbf} over $[-\frac{\pi}{2},\frac{\pi}{2})$. This gives: 
\begin{equation*}
\int_{-\frac{\pi}{2}}^{\frac{\pi}{2}} \bigg\{ \partial_t f^\varepsilon - \nabla_x \cdot  (\nabla_x U^0 f^{\varepsilon})  - \partial_\theta\bigg( \bigg[\partial_{\theta} U^1 + \xi G[f^{\varepsilon}](\theta) \bigg] f^{\varepsilon}\bigg) - d \Delta_x f^{\varepsilon}\bigg\} \, \frac{d\theta}{\pi} = 0, \\
\end{equation*}
\noindent which leads to the continuity equation for $\rho^\varepsilon$:
 \begin{equation*}
 \partial_t \rho^\varepsilon -\nabla_x \cdot  (\nabla_x U^0 \rho^{\varepsilon}) - d \Delta_x \rho^{\varepsilon}=0.
 \end{equation*}
 \noindent In the limit $\varepsilon \rightarrow 0$, $\rho^\varepsilon \rightarrow \rho$ which leads to Eq.~\eqref{Eqrho}.

\paragraph{Equation for $\theta_0$}
We multiply Eq.~\eqref{pbf} by the GCI $\Psi_{\theta_{f^\varepsilon}}$ associated with the direction $\theta_{f^\varepsilon}$ of $f^\varepsilon$, namely $\Psi_{\theta_{f^\varepsilon}} = g(\theta - \theta_{f^\varepsilon})$ where $g$ is the function defined in Prop.~\ref{prop4}. We integrate with respect to $\theta$ and first note that:
\begin{equation*}
\int_{-\frac{\pi}{2}}^{\frac{\pi}{2}} Q(f^{\varepsilon}) \Psi_{\theta_{f^\varepsilon}} d\theta = \int_{-\frac{\pi}{2}}^{\frac{\pi}{2}} L_{\theta_{f^{\varepsilon}}}f^{\varepsilon} \Psi_{\theta_{f^\varepsilon}} d\theta=0,
\end{equation*}
\noindent by~\eqref{GCIset}. Since $f^{\varepsilon} \rightarrow \rho M_{\theta_0}$, we have $\theta_{f^\varepsilon} \rightarrow \theta_0$ and $\Psi_{\theta_{f^\varepsilon}} \rightarrow \Psi_{\theta_0}$. Therefore, in the limit $\varepsilon \rightarrow 0$, we get:
\begin{equation}\label{inhomo2}
\int_{-\frac{\pi}{2}}^{\frac{\pi}{2}}  \bigg(\partial_t (\rho M_{\theta_0}) - \nabla_x \cdot  (\nabla_x U^0 \rho M_{\theta_0})  - \partial_\theta\bigg( \bigg[\partial_{\theta} U^1 + \xi G[\rho M_{\theta_0}](\theta) \bigg] \rho M_{\theta_0}\bigg) - d \Delta_x (\rho M_{\theta_0})\bigg) \Psi_{\theta_0} d\theta= 0.
\end{equation}
\noindent For simplicity, we denote $M_{\theta_0} = M$. We have:
\begin{equation*}
\begin{split}
\Delta_x (\rho M) &= M \Delta_x \rho + \rho \Delta_x M + 2\nabla_x\rho  \cdot  \nabla_x M,\\
\nabla_x  \cdot  (\nabla_x U^0 \rho M) &= M \nabla_x \cdot  (\nabla_x U^0  \rho) + \rho \nabla_x U^0  \cdot  \nabla_x M. 
\end{split}
\end{equation*}
\noindent Using the continuity equation~\eqref{Eqrho}, we have:
\begin{equation*}
\partial_t(\rho M) = \rho \partial_t M + M \partial_t \rho = \rho \partial_t M + (\nabla_x \cdot  (\nabla_x U^1 \rho) + d\Delta_x \rho) M .
\end{equation*}
\noindent So:
\begin{equation*}
\partial_t (\rho M) - \nabla_x  \cdot  (\nabla_x U^0 \rho M) - d \Delta_x (\rho M) = \rho \partial_t M - \rho \nabla_x U^0  \cdot  \nabla_x M - d\rho \Delta_x M - 2d\nabla_x \rho \cdot  \nabla_x M.
\end{equation*}
\noindent Therefore, Eq.~\eqref{inhomo2} reads: 
\begin{equation}\label{inhomo3}
 \rho \int_{-\frac{\pi}{2}}^{\frac{\pi}{2}} \partial_t M \Psi \frac{d\theta}{\pi} - X_1 - X_2 - X_3 - X_4 =0,
\end{equation} 
 \noindent where: 
 \begin{empheq}[left=\empheqlbrace]{align}
 X_1 &= \int_{-\frac{\pi}{2}}^{\frac{\pi}{2}}  \bigg( \rho (\nabla_x U^0 + 2d\frac{\nabla_x \rho}{\rho})  \cdot  \nabla_x M\bigg) \Psi \frac{d\theta}{\pi} \label{X1},\\
 X_2 &= \int_{-\frac{\pi}{2}}^{\frac{\pi}{2}} \partial_\theta\bigg( \partial_{\theta} U^1 \rho M\bigg) \Psi \frac{d\theta}{\pi} \label{X2},\\
  X_3 &=\xi \int_{-\frac{\pi}{2}}^{\frac{\pi}{2}} \partial_\theta\bigg(G[\rho M](\theta) \rho M\bigg) \Psi \frac{d\theta}{\pi}\label{X3},\\
X_4 &= d\rho \int_{-\frac{\pi}{2}}^{\frac{\pi}{2}} \Delta_x M \Psi \frac{d\theta}{\pi}\label{X4}.
 \end{empheq}
\noindent We now turn to the development of each term of Eq.~\eqref{inhomo3}. We have:
\begin{equation}\label{nablaxM}
\nabla_x M = 2 r \sin 2(\theta-\theta_0) M \nabla_x \theta_0.
\end{equation} 
\noindent Then, 
\begin{equation*}
\begin{split}
(\nabla_x U^0 + 2d\frac{\nabla_x \rho}{\rho}) \cdot \nabla_xM =2r\sin 2(\theta-\theta_0) M \bigg(\nabla_x U^0 + 2d\frac{\nabla_x \rho}{\rho}\bigg)  \cdot  \nabla_x\theta_0,
\end{split}
\end{equation*}
\noindent and thus, $X_1$ can be written:
\begin{equation*}
 X_1 =2r \rho \bigg(\nabla_x U^0 \cdot \nabla_x \theta_0 + 2 d\frac{\nabla_x \rho \cdot \nabla_x \theta_0}{\rho}\bigg) \langle \sin 2(\theta-\theta_0) \Psi \rangle .
\end{equation*} 
\noindent From integration by parts, the following relations can be written: 
 \begin{equation}\label{sinmean}
 \begin{split}
 \langle \sin 2(\theta-\theta_0) \Psi \rangle  = \frac{1}{4r^2} (1 - \frac{1}{Z^2}) = \frac{1}{4r^2}\alpha_1.
 \end{split}
 \end{equation}
\noindent Therefore, we have:
\begin{equation}\label{valueX1}
X_1 = \frac{\rho \alpha_1}{2 r} (\nabla_x U^0 \cdot \nabla_x \theta_0 + 2 d\frac{\nabla_x \rho \cdot \nabla_x \theta_0}{\rho}).
\end{equation}
\noindent Since $X_2$ is the integral of a $\pi$-periodic function over a period, we can write 
\begin{equation*}
X_2 = \int_{\theta_0 - \pi/2}^{\theta_0 + \pi/2} \partial_\theta \big(\partial_\theta U^1 \rho M \big) \Psi \frac{d\theta}{\pi}.
\end{equation*}
\noindent Now, by construction, (see prop~\ref{prop4}), $\Psi(\theta_0-\frac{\pi}{2}) = \Psi(\theta_0) = \Psi(\theta_0 + \frac{\pi}{2}) = 0$. So, integrating by parts, we have 
\begin{equation*}
X_2 = - \int_{\theta_0 - \pi/2}^{\theta_0 + \frac{\pi}{2}} \rho M \; \partial_\theta U^1 \; \partial_\theta \Psi \frac{d\theta}{\pi}.
\end{equation*}
\noindent Now, by construction again (see~\eqref{dg2}), we have
\begin{equation}\label{partialPsi}
\partial_\theta \Psi = \frac{1}{2r}(1 - \frac{1}{M Z^2}).
\end{equation}
\noindent Using again the $\pi$-periodicity of $U^1$, we obtain:
\begin{equation}\label{valueX2}
 X_2 =-\frac{ \rho}{2r} \langle \partial_{\theta} U^1 (1 - \frac{1}{M Z^2}) \rangle  = -\frac{\rho}{2r}  \langle \partial_{\theta} U^1 \rangle  .
 \end{equation}
\noindent Now, let us turn to $X_3$. The details of this computation are postponed to appendix~\ref{app3}. We find:
\begin{equation}\label{valueX3}
\begin{split}
\displaystyle X_3 = -\frac{dL^2}{12} \bigg[& -  c(r)(\rho \Delta_x \theta_0 + 2 \nabla_x \theta_0  \cdot  \nabla_x \rho )\\
&+\big( 2 \rho \gamma_1 \nabla_x\theta_0 \otimes \nabla_x \theta_0 -\frac{1}{8Z^2}\nabla_x^2 \rho \big) : [\omega_0\otimes \omega_0^\perp + \omega_0^\perp \otimes \omega_0]\\
& +  \gamma_1 (\rho \nabla_x^2\theta_0 +\nabla_x\theta_0 \otimes \nabla_x \rho + \nabla_x \rho \otimes \nabla_x \theta_0) : [\omega_0\otimes \omega_0 - \omega_0^\perp \otimes \omega_0^\perp]\bigg],
\end{split}
\end{equation}
\noindent where, using~\eqref{r}, 
\begin{equation*}
\gamma_1 = \frac{1}{4Z^2} -1 + \frac{3c(r)}{2r}  = \frac{1}{4Z^2} - 1 + \frac{6 d c(r)}{\alpha L^2 \xi \gamma}.
\end{equation*}
\noindent We note that $\frac{\alpha L^4 \gamma}{48 r }\gamma_1 = \frac{1}{2r} \alpha_3$. Finally, let us explicit the last term $X_4$. A direct computation gives:  
\begin{equation*}
\begin{split}
\Delta_x M &= M \bigg[4r\big[r\sin^2 2(\theta-\theta_0) - \cos 2(\theta-\theta_0)\big]|\nabla_x \theta_0|^2 + 2r\sin 2(\theta-\theta_0) \Delta_x\theta_0 
\bigg].
\end{split}
\end{equation*}
\noindent Then, we deduce that 
\begin{equation*}
\begin{split}
X_4 =d\rho\bigg[&2r \Delta_x\theta_0  \langle \sin 2(\theta-\theta_0) \Psi \rangle \\
&+|\nabla_x\theta_0|^2 4r[-\langle \cos 2(\theta-\theta_0) \Psi \rangle +r \langle \sin^2 2(\theta-\theta_0) \Psi \rangle ]\bigg].
\end{split}
\end{equation*}
 \noindent By symmetry, we have:
 \begin{equation*}
\langle \sin^2 2(\theta - \theta_0) \Psi \rangle  = \frac{1}{r} \langle \cos 2 (\theta - \theta_0) \Psi \rangle .
 \end{equation*}
\noindent Therefore, with~\eqref{sinmean}, we get:
\begin{equation}\label{valueX4}
X_4 =\frac{d\rho}{2 r} (1 - \frac{1}{Z^2}) \Delta_x\theta_0  = \frac{d\rho}{2r} \alpha_1 \Delta_x \theta_0.
\end{equation}
\noindent Now, $\partial_t M = 2r \sin 2(\theta-\theta_0) M \partial_t \theta_0$, and 
\begin{equation}\label{intdtM}
\rho \int_{-\frac{\pi}{2}}^{\frac{\pi}{2}} \partial_t M \;  \Psi = 2r\rho \langle \sin 2(\theta - \theta_0)\Psi \rangle \partial_t \theta_0 = \frac{\rho}{2r}(1-\frac{1}{Z^2}) \partial_t \theta_0 = \frac{\alpha_1 \rho}{2r} \partial_t \theta_0
\end{equation}
\noindent Collecting~\eqref{valueX1} to~\eqref{valueX4} and inserting them into~\eqref{inhomo3} leads to~\eqref{Eqtheta}.
\end{proof}

\setcounter{equation}{0}
\section{Case of a homogeneous fiber distribution: stationary solutions}
\label{sec:homo}

In this section, we study the stationary solutions of~\eqref{Eqrho}-\eqref{Eqtheta} in the case of a spatially homogeneous fiber distribution and consequently no external spatial potential $U^0=0$. We make the following assumption: 

\begin{hypothesis}\label{hyp4}
The fiber spatial distribution is supposed to be homogeneous, i.e.\ there exists a constant $\rho_0 >0$ such that  $\rho(x,t) = \rho_0$ for all $(x,t) \in {\mathbb R}^2 \times [0, \infty)$. We also suppose that there are no external spatial forces, i.e.\ $U^0=0$.
\end{hypothesis}

We first note that in the absence of external forces, a uniform and constant density $\rho_0$ is a solution of Eq.~\eqref{Eqrho}. Now, we are interested in the stationary solutions for the fiber orientation equation (\ref{Eqtheta}). Noting that the terms involving the spatial derivatives of $\rho$, we find that such stationary solutions satisfy the following equation:
\begin{equation}\label{TetHomo}
\begin{split}
 \alpha_2 \Delta_x \theta_0 - &\alpha_3 [ \omega_0\otimes\omega_0 - \omega_0^\perp\otimes\omega_0^\perp]:  \nabla_x^2 \theta_0 \\
 &- 2 \alpha_3[\omega_0\otimes\omega_0^\perp + \omega_0^\perp\otimes\omega_0] :  \nabla_x\theta_0 \otimes \nabla_x \theta_0 = \alpha_5 \langle \partial_{\theta} U^1 \rangle .
  \end{split}
 \end{equation}
In this equation, the coefficients~$r$, $\alpha_1$, $\alpha_2$ and $\alpha_3$ are constants thanks to~\eqref{r}. Moreover, using~\eqref{coeffalpha}, they can be written as functions of $d$, $L^2$ and $r$ as follows:
\begin{empheq}[left=\empheqlbrace]{align}
\alpha_1(r)&=1 - \frac{1}{Z(r)^2} \label{alpha1},\\
\alpha_2(d,r,L^2)&=d\bigg(1+\frac{L^2 r c(r)}{6 \alpha_1(r)} \bigg) \label{alpha2},\\
\alpha_3(d,r,L^2)& = \frac{d L^2 r}{6\alpha_1(r)} \mathcal{A}(r) \label{alpha3}.
\end{empheq}
with 
\begin{equation}
 \mathcal{A}(r)= \bigg(\frac{1}{4Z(r)^2} -1+\frac{3}{2}\frac{c(r)}{r}\bigg). 
\label{eq:curlyA}
\end{equation}
We now show that (\ref{TetHomo}) is an elliptic equation. We first introduce some definitions. 

Given a function $\mathbf{f}(x,E)$ smooth in its arguments $x \in \Omega$, $E \in {\mathbb R} \times {\mathbb R}^2 \times {\mathcal S}_2({\mathbb R})$, where ${\mathcal S}_2({\mathbb R})$ is the space of $2 \times 2$ symmetric matrices with real coefficients, we define the  non linear differential operator $F: C^{\infty}({\mathbb R}^2) \rightarrow C^{\infty}({\mathbb R}^2)$ such that for any $x\in {\mathbb R}^2$ and any $u\in C^\infty({\mathbb R}^2)$, we have
\begin{equation*}
 F(u(x)) = \mathbf{f}(x,D^2 u(x)),
 \end{equation*}
 \noindent where $D^2 u = \{D^\alpha u, \alpha \in \mathbb{N}^2,|\alpha|\leq 2\}$ and where, for a multi-index $\alpha = (\alpha_1,\alpha_2)\in \mathbb{N}^2$, $|\alpha|= \alpha_1+\alpha_2$ and $D^\alpha u = \frac{\partial^{|\alpha|} u}{\partial_{x_1}^{\alpha_1} \partial_{x_2}^{\alpha_2}}$. The operator $F$ is said to be elliptic at $u_1 \in C^\infty({\mathbb R}^2)$ (see Ref. \cite{Taylor1996}) if its linearization $DF(u_1)$ is an elliptic, linear differential operator.
We state the following proposition:

\begin{proposition}
Eq. (\ref{TetHomo}) can be put in the form
\begin{equation}
\label{TetHomoDevStatio}
\mathbf{f}(x,D^2 \theta_0(x)) = 0, \quad x \in {\mathbb R}^2, 
\end{equation}
\noindent where $\mathbf{f}(x,D^2 \theta_0)$ is the following operator, quasi linear in $\theta_0$:
\begin{equation}\label{f2}
\mathbf{f}(x, D^2 \theta_0) = \sum_{i,j=1}^2 \partial_{x_i} \big( a_{ij} (\theta_0) \partial_{x_j} \theta_0 \big) - \alpha_5 h(\theta_0).
\end{equation}
\noindent Here, $h(\theta_0) = \langle \partial_\theta U^1 \rangle $ and $A(\theta_0)=(a_{ij}(\theta_0))_{i,j=1,2}$ is a $2 \times 2$ matrix such that:
\begin{equation}\label{aij}
A(\theta_0) = \begin{pmatrix} \alpha_2 - \alpha_3 \cos 2\theta_0 & -\alpha_3 \sin 2 \theta_0 \\ -\alpha_3 \sin 2 \theta_0 & \alpha_2 + \alpha_3 \cos 2 \theta_0\end{pmatrix}.
\end{equation}
\noindent Moreover, if the following condition is satisfied for all $r\in \mathbb{R}^+$:
\begin{equation}\label{cond}
\mathcal{A}(r)+c(r) \geq 0,
\end{equation}
\noindent where $\mathcal{A}(r)$ is given by (\ref{eq:curlyA}), then $F(\theta) = \mathbf{f}(x,D^2 \theta)$ is elliptic at $\theta_1$ for all $\theta_1 \in C^2({\mathbb R}^2)$. 
\end{proposition}

\begin{proof} 
 For any $\theta \in [-\frac{\pi}{2} ,\frac{\pi}{2})$, letting $\omega(\theta) = (\cos \theta,\sin \theta)$ and $\omega^\perp(\theta) = (-\sin \theta,\cos \theta)$, we have:
\begin{equation*}
\frac{d}{d\theta} [\omega(\theta) \otimes \omega(\theta) - \omega^\perp(\theta) \otimes \omega^\perp(\theta)] = 2 [\omega(\theta) \otimes \omega^\perp(\theta) + \omega^\perp(\theta) \otimes \omega(\theta)].
\end{equation*}
\noindent Let $F: C^\infty({\mathbb R}^2) \rightarrow C^\infty({\mathbb R}^2)$ be the non linear differential operator defined by:
\begin{equation*}
F(\theta_0) = \mathbf{f}(x,D^2 \theta_0),
\end{equation*}
\noindent for $\mathbf{f}$ defined by~\eqref{f2}. Let $DF(\theta_1)$ denote its linearization at $\theta_1$. Then, $DF(\theta_1)$ is a linear map from $C^2({\mathbb R}^2)$ to $C^0({\mathbb R}^2)$ and reads, for $v \in C^2({\mathbb R}^2)$:
\begin{equation}
DF(\theta_1) v = \frac{\partial F(\theta_1 + s v)}{\partial s}\bigg|_{s=0} = \sum_{i,j=1}^2 \bigg(a_{ij}(\theta_1) \partial_{x_i} \partial_{x_j} v\bigg) + L v,
\end{equation}
\noindent where $L$ is a linear differential operator of order 1 the coefficients of which depend on $D\theta_1$:
\begin{equation*}
\begin{split}
L v =& \sum_{i,j=1}^2 \bigg( a'_{ij}(\theta_1) (\partial_{x_i}\theta_1 \partial_{x_j} v + \partial_{x_i} v \partial_{x_j} \theta_1)\bigg) + \sum_{i=1}^2 \partial_{x_i} U^0 \partial_{x_i} v \\
&+ \sum_{i,j=1}^2 \bigg( a''_{ij}(\theta_1) \partial_{x_i}\theta_1 \partial_{x_j} \theta_1 + a'_{ij}(\theta_1) \partial_{x_i x_j} \theta_1 - \alpha_5 h'(\theta_1)\bigg) v,
\end{split}
\end{equation*}
\noindent where $a'_{ij}(\theta_1)$ and $a''_{ij}(\theta_1)$ are the first and second order derivatives of the coefficients of matrix $A$ which read:
\begin{empheq}[left=\empheqlbrace]{align*}
(a'_{ij}(\theta_1))_{i,j=1,2} &= 2 \alpha_3 \begin{pmatrix} \sin 2\theta_1 & -\cos 2 \theta_1 \\
 - \cos 2 \theta_1 & -\sin 2 \theta_1  \end{pmatrix},\\ (a''_{ij}(\theta_1))_{i,j=1,2} &= 4 \alpha_3 \begin{pmatrix} \cos 2\theta_1 & \sin 2 \theta_1 \\  \sin 2 \theta_1 & -\cos 2 \theta_1  \end{pmatrix}.
\end{empheq}
\noindent Therefore, the linearization of $F$ at $\theta_1$ is elliptic provided that the matrix $A(\theta_1) = \big(a_{ij}(\theta_1)\big)_{i,j=1,2}$ is positive-definite.
 
Note that the determinant of the matrix $A(\theta_1) = (a_{ij}(\theta_1))_{i,j=1,2}$ is given by $\det(A(\theta_1)) = \alpha_2^2 - \alpha_3^2$ and does not depend on $\theta_1$. Moreover, $\det(A(\theta_1)) >0$ provided that $|\frac{\alpha_2}{\alpha_3}|>1$. The eigenvalues of the matrix $A(\theta_1)$ solve \begin{equation*}
\det(A(\theta_1) - \lambda I) = \lambda^2 - 2 \lambda \alpha_2 + \alpha_2^2 - \alpha_3^2=0
\end{equation*}
\noindent and the determinant $\Delta = 4 \alpha_3^2$ is strictly positive as long as $\alpha_3 \neq 0$. In this case, the matrix $A(\theta_1)$ has two distinct real eigenvalues given by: 
\begin{equation*}
\lambda ^\pm = (\alpha_2 \pm \alpha_3).
\end{equation*}
\noindent Therefore, the matrix $A(\theta_1)$ is positive definite if and only if $\alpha_2 > |\alpha_3|$.

\noindent We now analyse the sign of each coefficient $\alpha_1,\alpha_2,\alpha_3$. First of all (see Eq.  \eqref{Mr}), the $p$-th derivative $Z^{(p)}$ of $Z$ with respect to $r$ reads:
\begin{equation*}
Z^{(p)}(r) =  \int\limits_{-\pi/2}^{\pi/2} (\cos 2\theta)^{p} e^{r\cos 2\theta}\frac{d\theta}{\pi},
\end{equation*}
\noindent and we have $Z^{(2k)}(r) \geq 0$ for all $k\in \mathbb{N}^+$ and all $r\in\mathbb{R}^+$ as the functions $\theta \rightarrow (\cos 2\theta)^{2k}e^{r\cos 2\theta}$ are positive for any $r\in \mathbb{R}^+$. We deduce that $Z^{(2k+1)}(r)$ are increasing functions of $r$ for any $k\in \mathbb{N}^+$. Note that from the symmetry of the function $\cos 2\theta$, we have for any $k\in \mathbb{N}^+$:
\begin{equation*}
 \int\limits_{-\pi/2}^{\pi/2} (\cos 2\theta)^{2k+1} \frac{d\theta}{\pi}= Z^{(2k+1)}(0) = 0.
\end{equation*}
\noindent Therefore, we also have that $Z^{(2k+1)}(r) \geq Z^{(2k+1)}(0) \geq 0$ for any $k\in \mathbb{N}$. We thus obtain that for any $p \in \mathbb{N}$ and any $r\in \mathbb{R}^+$:
\begin{equation*}
Z^{(p)}(r) \geq Z^{(p)}(0) \geq 0,
\end{equation*}
\noindent and we note that $Z(r) \rightarrow \infty$ as $r \rightarrow \infty$. Moreover, as $Z(0) =  1$ we deduce $Z(r) \geq 1$ for any $r\in [0,+\infty)$. We also note that:
\begin{equation*}
c(r) = \frac{Z^{(1)}(r)}{Z(r)} \geq \frac{Z^{(1)}(0)}{Z(r)} \geq 0,
\end{equation*}
\noindent and we have:
\begin{equation*}
\alpha_1(r) \geq 0, \; \; \alpha_2(d,r,L^2) \geq 0 \; \; \forall (r,L,d)\in \mathbb{R}^+ \times  \mathbb{R} \times \mathbb{R}^+.
\end{equation*}
\noindent Now, by integration by parts, we can write:
\begin{equation}\label{crr}
\begin{split}
\frac{c(r)}{r} &= \frac{1}{r Z(r)}\int\limits_{-\pi/2}^{\pi/2} \cos 2\theta e^{r\cos 2\theta}\frac{d\theta}{\pi} = \frac{1}{Z(r)} \int\limits_{-\pi/2}^{\pi/2} \sin^2 2\theta e^{r\cos 2\theta}\frac{d\theta}{\pi}\\
&= 1 -\frac{1}{ Z(r)} \int\limits_{-\pi/2}^{\pi/2} \cos^2 2\theta e^{r\cos 2\theta}\frac{d\theta}{\pi} = 1 - \frac{Z^{(2)}(r)}{Z(r)}.
\end{split}
\end{equation}
\noindent We now show that
 \begin{equation*}
\frac{Z^{(2)}(r)}{Z(r)} \geq \frac{1}{2},
\end{equation*}
\noindent or, equivalently, that
 \begin{equation*}
Z(r) \leq 2 Z^{(2)}(r).
\end{equation*}
\noindent Indeed,
\begin{equation}\label{eq:ZZ2}
Z(r) = \int\limits_{-\frac{\pi}{2}}^{\frac{\pi}{2}} e^{r\cos 2\theta} \frac{d\theta}{\pi} = \int\limits_{-\frac{\pi}{2}}^{\frac{\pi}{2}} \cos^2 2\theta e^{r\cos 2\theta} \frac{d\theta}{\pi} + \int\limits_{-\frac{\pi}{2}}^{\frac{\pi}{2}} \sin^2 2\theta e^{r\cos 2\theta} \frac{d\theta}{\pi},
\end{equation}
\noindent and, by integration by parts, we have:
\begin{equation*}
\int\limits_{-\frac{\pi}{2}}^{\frac{\pi}{2}} \sin^2 2\theta e^{r\cos 2\theta} \frac{d\theta}{\pi} = \int\limits_{-\frac{\pi}{2}}^{\frac{\pi}{2}} \cos^2 2\theta e^{r\cos 2\theta} \frac{d\theta}{\pi} - r\int\limits_{-\frac{\pi}{2}}^{\frac{\pi}{2}} \sin^2 2\theta \cos 2 \theta e^{r\cos 2\theta} \frac{d\theta}{\pi}.
\end{equation*}
\noindent To show that $\int\limits_{-\frac{\pi}{2}}^{\frac{\pi}{2}} \sin^2 2\theta \cos 2 \theta e^{r\cos 2\theta} \frac{d\theta}{\pi}$ is positive, we can note that it is an increasing function of $r$ and that for $r=0$ we have $\int\limits_{-\frac{\pi}{2}}^{\frac{\pi}{2}} \sin^2 2\theta \cos 2\theta \frac{d\theta}{\pi} = 0$. Indeed, the derivative of this term with respect to $r$ reads:
\begin{equation*}
\frac{d}{dr} \big(\int\limits_{-\frac{\pi}{2}}^{\frac{\pi}{2}} \sin^2 2\theta \cos 2 \theta e^{r\cos 2\theta} \frac{d\theta}{\pi}\big) =  \int\limits_{-\frac{\pi}{2}}^{\frac{\pi}{2}} \sin^2 2\theta \cos^2 2 \theta e^{r\cos 2\theta} \frac{d\theta}{\pi},
\end{equation*}
\noindent which is positive for any $r\geq 0$. Therefore:
\begin{equation*}
\int\limits_{-\frac{\pi}{2}}^{\frac{\pi}{2}} \sin^2 2\theta e^{r\cos 2\theta} \frac{d\theta}{\pi} \leq  \int\limits_{-\frac{\pi}{2}}^{\frac{\pi}{2}} \cos^2 2\theta e^{r\cos 2\theta} \frac{d\theta}{\pi},
\end{equation*}
\noindent for any $r \geq 0$, and inserting this expression into Eq. \eqref{eq:ZZ2}, we obtain:
\begin{equation*}
Z(r) \leq 2  \int\limits_{-\frac{\pi}{2}}^{\frac{\pi}{2}} \cos^2 2\theta e^{r\cos 2\theta} \frac{d\theta}{\pi} = 2 Z^{(2)}(r).
\end{equation*}
\noindent All together, we have:
\begin{equation*}
\frac{c(r)}{r} \leq \frac{1}{2},
\end{equation*}
\noindent for any $r\in \mathbb{R}^+$. This relation together with the fact that $1 - \frac{1}{4Z^2(r)} \geq \frac{3}{4}$ leads to:
\begin{equation*}
\alpha_3(d,r,L^2) = \frac{dL^2r}{6\alpha_1(r)} ( \frac{3c(r)}{2r} - (1 - \frac{1}{4Z^2})) \leq 0.
\end{equation*}
\noindent Now, we can write:
\begin{equation*}
|\frac{\alpha_2}{\alpha_3}| > 1 \; \Leftrightarrow \alpha_2 > -\alpha_3,
\end{equation*}
\noindent or equivalently, using Eqs.~\eqref{alpha2} and \eqref{alpha3}:
\begin{equation*}
|\frac{\alpha_2}{\alpha_3}| > 1 \; \Leftrightarrow \frac{6 \alpha_1(r)}{drL^2} > -(\mathcal{A}(r)+c(r)).
\end{equation*}
\noindent Therefore, if \eqref{cond} holds, then $\alpha_2 > |\alpha_3|$ and the matrix $A(\theta_1)$ is positive definite for all $r \in [0,1], \;  L \in \mathbb{R}^+, \; d \in \mathbb{R}$, independently of $\theta_1 \in C^2({\mathbb R}^2)$. We conclude that $F$ is elliptic at $\theta_1$ for all $\theta_1 \in C^2({\mathbb R}^2)$, provided \eqref{cond} holds.
\end{proof}

\begin{remark} As shown by Fig.~\ref{figcoeff}
, $\mathcal{A}(r)+c(r)$ is positive for any $r \in \mathbb{R}^+$. The rigorous proof of this fact will be the subject of future work. 
\begin{figure}[h!]
\centering
\includegraphics[scale=0.5]{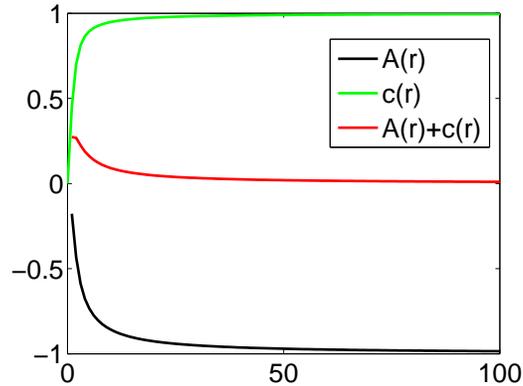}
\caption{Functions $\mathcal{A}(r)$ (black), $c(r)$ (green) and $\mathcal{A}(r)+c(r)$ (red) as functions of $r\in [0,100]$. }
   \label{figcoeff}
\end{figure}
\end{remark}

\setcounter{equation}{0}
\section{Conclusion}
\label{sec:conclu}

In this paper, we have formally derived a macroscopic model for temporarily linked fibers interacting through alignment at the links. We have shown that the corresponding kinetic model involves two distribution functions: the fiber distribution function and the cross-link distribution function. The latter can be seen as a joint two-particle fiber distribution function. This model provides a unique explicit example of a kinetic model closed at the level of the two particle distribution function. We then considered the regime of a fast fiber linking/unlinking process, where the link distribution function can be expressed simply in terms of the fiber distribution function. We studied the diffusive limit of the resulting equation and obtained a system of two coupled nonlinear diffusion equations for the fiber density and mean orientation. In the homogeneous fiber density case, we showed that the resulting quasilinear problem is elliptic. Future works will deeper investigate the mathematical properties of the models, such as rigorously proving the mean-field kinetic limit of the particle system or proving existence and uniqueness of smooth solutions for the macroscopic diffusion system. Numerical simulations will be performed to validate the macroscopic model by comparison with the individual based model. Further perspectives are the removal of the fast fiber linking/unlinking hypothesis, in order to understand how a finite lifetime of the cross-links affects the macroscopic dynamics. 
\newpage
\appendix
\section{Proof of Theorem~\ref{thm1}}\label{proofIBMcont}
\subsection{Evolution equation for the fibers} 
 For all observable functions $\Phi(x,\theta)$, we define:
\begin{equation*}
\begin{split}
 \langle f^N,\Phi \rangle  = \int\limits \Phi(x,\theta) f^N(t,x,\theta)  dx_1d\theta = \frac{1}{N} \sum_{i=1}^N \Phi(X_i(t),\theta_i(t)).
\end{split}
\end{equation*}
\noindent Similarly, for all two-particle observable functions $\Psi (x_1,\theta_1,\ell_1,x_2,\theta_2,\ell_2)$, we define: 
\begin{equation*}
\begin{split}
\langle \hspace{-0.8mm} \langle g^K,\Psi \rangle \hspace{-0.8mm} \rangle = \int &\Psi(x_1,\theta_1,\ell_1,x_2,\theta_2,\ell_2) g^K(x_1,\theta_1,\ell_1,x_2,\theta_2,\ell_2) dx_1dx_2 \frac{d\theta_1}{\pi} \frac{d\theta_2}{\pi}d\ell_1d\ell_2\\
= \frac{1}{2K} \sum_{k=1}^K \bigg(&\Psi(X_{i(k)},\theta_{i(k)},\ell^k_{i(k)}, X_{j(k)}, \theta_{j(k)},\ell^k_{j(k)}) \\
&+ \Psi(X_{j(k)},\theta_{j(k)},\ell^k_{j(k)}, X_{i(k)}, \theta_{i(k)},\ell^k_{i(k)}) \bigg),
\end{split}
\end{equation*}
\noindent where integrals over $x$ are carried over ${\mathbb R}^2$, in $\theta$ over $(-\frac{\pi}{2},\frac{\pi}{2})$ and in $\ell$ over $(-\frac{L}{2},\frac{L}{2})$. We recall the notations $C^k_{i(k),j(k)}=(X_{i(k)},\theta_{i(k)},\ell^k_{i(k)},X_{j(k)},\theta_{j(k)},\ell^k_{j(k)})$ (resp. $C^k_{j(k),i(k)} = (X_{j(k)},\theta_{j(k)},\ell^k_{j(k)},X_{i(k)},\theta_{i(k)},\ell^k_{i(k)})$). Then: 
\begin{equation*}
 \frac{d}{dt}\langle f^N,\Phi \rangle  =  \frac{1}{N} \sum_{i=1}^N \bigg(\nabla_x \Phi(X_i(t),\theta_i(t)) \cdot \frac{d X_i(t)}{dt} +\partial_\theta \Phi(X_i(t),\theta_i(t)) \frac{d\theta_i(t)}{dt} \bigg).
\end{equation*}
\noindent Using~\eqref{motionIBMX} and~\eqref{motionIBMTETA}, we obtain: 
\begin{equation*}
\begin{split}
 \frac{d}{dt}\langle f^N,\Phi \rangle & \\
 =  -\frac{1}{N} \sum_{i=1}^N \bigg[ (\mu &\nabla_x \Phi  \cdot \nabla_x U + \lambda \partial_{\theta} \Phi \partial_{\theta} U)(X_i,\theta_i)\\
  + d (\mu \nabla_x \Phi &\cdot \nabla_{x} \log(\tilde{f}^N) + \lambda \partial_{\theta} \Phi  \partial_{\theta} \log(\tilde{f}^N) )(X_i,\theta_i) \\
 + \mu \nabla_x \Phi&(X_i,\theta_i) \cdot \frac{1}{2} \sum_{k=1}^K  (\nabla_{x_1}V \delta_{i(k)}(i) + \nabla_{x_2}V \delta_{j(k)}(i))(C^k_{i(k),j(k)}) \\
 + \lambda & \partial_{\theta} \Phi(X_i,\theta_i) \frac{1}{2} \sum_{k=1}^K (\partial_{\theta_1} V \delta_{i(k)}(i) + \partial_{\theta_2}V \delta_{j(k)}(i))(C^k_{i(k),j(k)}) \\
&+  \lambda \partial_{\theta} \Phi(X_i,\theta_i) \frac{1}{2} \sum_{k=1}^K( \partial_{\theta_1} b \; \delta_{i(k)}(i) + \partial_{\theta_2} b \; \delta_{j(k)}(i))(\theta_{i(k)},\theta_{j(k)})\bigg].
\end{split}
\end{equation*}
\noindent We get, using the definition of a distributional derivative:
\begin{equation*}
\begin{split}
 \frac{d}{dt}\langle f^N,\Phi \rangle= \langle \mu \nabla_x \cdot & \big(f^N \nabla_x (U+d\log{\tilde{f}^N})\big) + \lambda \partial_\theta \big(f^N\partial_\theta(U + d\log{\tilde{f}^N})\big),\Phi \rangle \\
 -  \frac{\mu}{2N} \sum_{i=1}^N \bigg(\nabla_x &\Phi(X_i,\theta_i) \cdot \sum_{k=1}^K (\nabla_{x_1}V \delta_{i(k)}(i) + \nabla_{x_2}V \delta_{j(k)}(i))(C^k_{i(k),j(k)})\bigg) \\
 - \frac{\lambda}{2N} \sum_{i=1}^N& \bigg(\partial_{\theta} \Phi(X_i,\theta_i)\sum_{k=1}^K (\partial_{\theta_1} V \delta_{i(k)}(i) + \partial_{\theta_2}V \delta_{j(k)}(i))(C^k_{i(k),j(k)})\bigg) \\
-  \frac{\lambda}{2N} &\sum_{i=1}^N \bigg(\partial_{\theta} \Phi(X_i,\theta_i)\sum_{k=1}^K( \partial_{\theta_1} b \; \delta_{i(k)}(i) + \partial_{\theta_2} b \; \delta_{j(k)}(i))(\theta_{i(k)},\theta_{j(k)})\bigg).
\end{split}
\end{equation*}
\noindent Now, exchanging the sums in $i$ and $k$ in the previous equation, one obtains:
\begin{equation*}
\begin{split}
 \frac{d}{dt} \langle f^N,\Phi \rangle =  \langle\mu \nabla_x \cdot \big(f^N \nabla_x &(U+d\log{\tilde{f}^N})\big) + \lambda \partial_\theta \big(f^N\partial_\theta(U + d\log{\tilde{f}^N})\big),\Phi \rangle  \\
 - \frac{\mu}{2N} \sum_{k=1}^K \bigg(  \nabla_{x_1}V(C^k_{i(k),j(k)}) \cdot \nabla_x &\Phi(X_{i(k)},\theta_{i(k)}) + \nabla_{x_2}V(C^k_{i(k),j(k)}) \cdot \nabla_x \Phi(X_{j(k)},\theta_{j(k)}))\bigg)\\
 -\frac{\lambda}{2N}\sum_{k=1}^K \bigg( \partial_{\theta_1} V(C^k_{i(k),j(k)})& \partial_\theta \Phi(X_{i(k)},\theta_{i(k)})  + \partial_{\theta_2}V(C^k_{i(k),j(k)}) \partial_{\theta}\Phi(X_{j(k)},\theta_{j(k)})) \bigg) \\
-  \frac{\lambda}{2N}\sum_{k=1}^K\bigg( \partial_{\theta_1} b(\theta_{i(k)},\theta_{j(k)}&) \partial_{\theta} \Phi(X_{i(k)},\theta_{i(k)}) + \partial_{\theta_2} b(\theta_{i(k)},\theta_{j(k)}) \partial_{\theta} \Phi(X_{j(k)},\theta_{j(k)})\bigg).
\end{split}
\end{equation*}
\noindent From the symmetry of $V$ (see Eq.~\eqref{V}), the following expressions hold:
$$ 
\nabla_{x_2}V(C^k_{i(k),j(k)}) = \nabla_{x_1} V(C^k_{j(k),i(k)}),\quad \partial_{\theta_2}V(C^k_{i(k),j(k)}) = \partial_{\theta_1} V(C^k_{j(k),i(k)}), 
$$
 \noindent and from the symmetry of $b$, we have:
 \begin{equation*}
 \partial_{\theta_2} b(\theta_{i(k)},\theta_{j(k)}) = \partial_{\theta_1} b(\theta_{j(k)},\theta_{i(k)}),
 \end{equation*}
 \noindent leading to:
\begin{equation*}
\begin{split}
 \frac{d}{dt}\langle f^N,\Phi \rangle =  \langle \mu \nabla_x \cdot &\big(f^N \nabla_x (U+d\log{\tilde{f}^N})\big) + \lambda \partial_\theta \big(f^N\partial_\theta(U + d\log{\tilde{f}^N})\big),\Phi \rangle  \\
 - \frac{\mu}{2N} \sum_{k=1}^K \bigg(  \nabla_{x_1}V&(C^k_{i(k),j(k)}) \cdot  \nabla_x \Phi(X_{i(k)},\theta_{i(k)})  + \nabla_{x_1}V(C^k_{j(k),i(k)}) \cdot \nabla_x \Phi(X_{j(k)},\theta_{j(k)}))\bigg)\\
 -\frac{\lambda}{2N}\sum_{k=1}^K \bigg(& \partial_{\theta_1} V(C^k_{i(k),j(k)}) \partial_\theta \Phi(X_{i(k)},\theta_{i(k)})  + \partial_{\theta_1}V(C^k_{j(k),i(k)}) \partial_{\theta}\Phi(X_{j(k)},\theta_{j(k)})) \bigg) \\
-  \frac{\lambda}{2N}\sum_{k=1}^K&\bigg( \partial_{\theta_1} b(\theta_{i(k)},\theta_{j(k)}) \partial_{\theta} \Phi(X_{i(k)},\theta_{i(k)}) + \partial_{\theta_1} b(\theta_{j(k)},\theta_{i(k)}) \partial_{\theta} \Phi(X_{j(k)},\theta_{j(k)})\bigg),
\end{split}
\end{equation*}
\noindent or again:
\begin{equation*}
\begin{split}
 \frac{d}{dt}\langle f^N,\Phi \rangle =  \langle \mu \nabla_x \cdot \big(f^N \nabla_x (U+d\log{\tilde{f}^N})\big) +& \lambda \partial_\theta \big(f^N\partial_\theta(U + d\log{\tilde{f}^N})\big),\Phi \rangle  \\
 - \frac{K}{N}  \langle  \hspace{-0.8mm} \langle  \frac{1}{2K}\sum_{k=1}^K \big(\delta_{(C^k_{i(k),j(k)})} + \delta_{(C^k_{j(k),i(k)})} \big)& (x_1,\theta_1,\ell_1,x_2,\theta_2,\ell_2), \\
 \big( \mu \nabla_{x_1} &V(x_1,\theta_1,\ell_1,x_2,\theta_2,\ell_2) \cdot \nabla_x \Phi(x_1,\theta_1)\\
  +&\lambda \partial_{\theta_1} V(x_1,\theta_1,\ell_1,x_2,\theta_2,\ell_2) \partial_{\theta_1} \Phi(x_1,\theta_1)
 \\
 & +  \lambda \partial_{\theta_1} b(\theta_1,\theta_2) \partial_{\theta_1} \Phi(x_1,\theta_1) \big) \rangle  \hspace{-0.8mm} \rangle .
\end{split}
\end{equation*}
\noindent Therefore, we obtain:
\begin{equation*}
\begin{split}
 \frac{d}{dt}\langle f^N,\Phi \rangle &= \langle \mu \nabla_x \cdot \big(f^N \nabla_x (U+d\log{\tilde{f}^N})\big) + \lambda \partial_\theta \big(f^N\partial_\theta(U + d\log{\tilde{f}^N})\big),\Phi \rangle  \\
 - \frac{K}{N}& \langle  \hspace{-0.8mm} \langle  g^K,\mu \nabla_{x_1} V(x_1,\theta_1,\ell_1,x_2,\theta_2,\ell_2) \cdot \nabla_{x} \Phi(x_1,\theta_1) \rangle  \hspace{-0.8mm} \rangle  \\
 &- \frac{K}{N}  \langle  \hspace{-0.8mm} \langle  g^K,\lambda \bigg(\partial_{\theta_1} V(x_1,\theta_1,\ell_1,x_2,\theta_2,\ell_2) + \partial_{\theta_1} b(\theta_1,\theta_2) \bigg) \partial_\theta\Phi(x_1,\theta_1) \rangle  \hspace{-0.8mm} \rangle .
 \end{split}
\end{equation*} 
\noindent Finally, we get:
\begin{equation}\label{kinetf}
\begin{split}
 \frac{d}{dt}\langle f^N,\Phi \rangle =& \langle \mu \nabla_x \cdot \big(f^N \nabla_x (U+d\log{\tilde{f}^N})\big) + \lambda \partial_\theta \big(f^N\partial_\theta(U + d\log{\tilde{f}^N})\big),\Phi \rangle  \\
 &+ \mu \frac{K}{N} \langle  \hspace{-0.8mm} \langle  \nabla_{x_1} \cdot \big(g^K\nabla_{x_1} V \big), \Phi(x_1,\theta_1)  \rangle  \hspace{-0.8mm} \rangle  \\
 &+ \lambda \frac{K}{N}  \langle  \hspace{-0.8mm} \langle  \partial_{\theta_1} \big( g^K \partial_{\theta_1} V + \partial_{\theta_1} b \big) , \Phi(x_1,\theta_1) \rangle  \hspace{-0.8mm} \rangle .\\
 =& \langle \mu \nabla_x \cdot \big(f^N \nabla_x (U+d\log{\tilde{f}^N})\big) + \lambda \partial_\theta \big(f^N\partial_\theta(U + d\log{\tilde{f}^N})\big)\\
 & +  \frac{K}{N}[[\mu \nabla_{x_1} \cdot \big(g^K\nabla_{x_1} V \big)+ \lambda  \partial_{\theta_1} \big( g^K (\partial_{\theta_1} V + \partial_{\theta_1} b)\big)]](x_1,\theta_1),\Phi(x_1,\theta_1) \rangle ,
 \end{split}
\end{equation} 
\noindent where, for a distribution $T$ acting on functions of $(x_1,\theta_1,\ell_1,x_2,\theta_2,\ell_2)$, we denote by $[[T]](x_1,\theta_1)$ the distribution which to any function $\Phi(x_1,\theta_1)$ associates 
\begin{equation*}
\langle [[T]](x_1,\theta_1),\Phi(x_1,\theta_1) \rangle  =  \langle  \hspace{-0.8mm} \langle  T,\Phi \textbf{1} \rangle  \hspace{-0.8mm} \rangle ,
\end{equation*}
\noindent and where $\textbf{1}$ is the constant function of the variables $(x_1,\theta_1,\ell_1,x_2,\theta_2,\ell_2)$ equal to 1. 
In the formal limit $N \rightarrow \infty$, $\frac{K}{N} \rightarrow \xi$ and given the assumptions on the regularizing sequences $\xi^N$, $\eta^N$, we get that $f^N \rightarrow f$, $\tilde{f}^N \rightarrow f$. Then, $\nabla_x \cdot  (f\nabla_x \log f) = \Delta_x f$ and $\partial_\theta (f \partial_\theta f) = \partial_\theta^2 f$ and we obtain:
\begin{equation}\label{systf}
 \frac{df}{dt} - \mu \nabla_x \cdot ((\nabla_x U) f) - \lambda \partial_{\theta} ((\partial_{\theta} U) f) - \mu \xi \nabla_x \cdot F_1 - \lambda \xi \partial_{\theta} F_2- d\mu\Delta_x f - d \lambda \partial^2_{\theta} f=0 ,
 \end{equation}
 \noindent where,
\begin{empheq}[left=\empheqlbrace]{align*}
F_1(x,\theta) = \int\limits (&g \nabla_x V)(x_1,\theta_1,\ell_1,x_2,\theta_2,\ell_2)  dx_2 \frac{d\theta_2}{\pi}d\ell_1d\ell_2 ,\\
F_2(x,\theta) = \int\limits ( (&g\partial_{\theta_1} V)(x_1,\theta_1,\ell_1,x_2,\theta_2,\ell_2) \\
&+  g(x_1,\theta_1,\ell_1,x_2,\theta_2,\ell_2)\partial_{\theta_1} b(\theta_1,\theta_2))dx_2 \frac{d\theta_2}{\pi}d\ell_1d\ell_2.
\end{empheq}

\subsection{Evolution equation for the fiber links}
 Following the same principle as for $f^N$ and given that the links are maintained over time, i.e.\ $\frac{d\ell^k_{i(k)}}{dt} =\frac{d\ell^k_{j(k)}}{dt}= 0$,   $\forall \; k \in [1,K]$, one can write:
\begin{equation}\label{dg1}
\begin{split}
\frac{d}{dt} \langle  \hspace{-0.8mm} \langle  g^K,\Psi  \rangle  \hspace{-0.8mm} \rangle  =\frac{1}{2K} \sum_{k=1}^k& \bigg[ \nabla_{x_1} \Psi(C^k_{i(k),j(k)})  \cdot \frac{dX_{i(k)}}{dt} + \nabla_{x_1} \Psi(C^k_{j(k),i(k)}) \cdot \frac{dX_{j(k)}}{dt}\\
 & + \nabla_{x_2} \Psi(C^k_{i(k),j(k)}) \cdot \frac{dX_{j(k)}}{dt} + \nabla_{x_2} \Psi(C^k_{j(k),i(k)}) \cdot \frac{dX_{i(k)}}{dt} \\
 &+ \partial_{\theta_1} \Psi(C^k_{i(k),j(k)}) \frac{d\theta_{i(k)}}{dt} + \partial_{\theta_1} \Psi(C^k_{j(k),i(k)}) \frac{d\theta_{j(k)}}{dt} \\
 &+ \partial_{\theta_2} \Psi(C^k_{i(k),j(k)}) \frac{d\theta_{j(k)}}{dt} + \partial_{\theta_2} \Psi(C^k_{j(k),i(k)}) \frac{d\theta_{i(k)}}{dt}  \bigg]\\
  = E_1 + \hdots& + E_4, 
 \end{split}
 \end{equation}
 \noindent where $E_k$ corresponds to the $k$-th line of~\eqref{dg1}. For the sake of simplicity, the computation of $E_1$ only is developed here. The computation of the other ones are similar and omitted. From Eqs.~\eqref{IBM1},~\eqref{IBM2}, one obtains: 
 \begin{equation*}
 \begin{split}
 E_1  = &\frac{1}{2K} \sum_{K=1}^k \bigg[\nabla_{x_1} \Psi(C^k_{i(k),j(k)}) \cdot \frac{dX_{i(k)}}{dt} + \nabla_{x_1} \Psi(C^k_{j(k),i(k)}) \cdot \frac{dX_{j(k)}}{dt}\bigg]\\
  = &- \frac{\mu}{2K} \sum_{k=1}^K \bigg[\bigg(
 \nabla_{x_1} \Psi(C^k_{i(k),j(k)}) \cdot \nabla_{x} \big( U + d \log \tilde{f}^N\big)(X_{i(k)},\theta_{i(k)}) \\
 &+ \nabla_{x_1} \Psi(C^k_{j(k),i(k)}) \cdot \nabla_{x} \big( U + d \log \tilde{f}^N\big) (X_{j(k)},\theta_{j(k)})\\
 &+ \frac{1}{2} \nabla_{x_1} \Psi(C^k_{i(k),j(k)}) \cdot \sum_{k'=1}^K  \bigg(\nabla_{x_1} V \delta_{(i(k'),i(k))} + \nabla_{x_2} V \delta_{(j(k'),i(k))}\bigg)(C^k_{i(k'),j(k')})\\
  &+ \frac{1}{2} \nabla_{x_1} \Psi(C^k_{j(k),i(k)}) \cdot  \sum_{k'=1}^K  \bigg(\nabla_{x_1} V \delta_{(i(k'),j(k))} + \nabla_{x_2} V \delta_{(j(k'),j(k))}\bigg)(C^k_{i(k'),j(k')})\bigg],
 \end{split}
 \end{equation*}
 \noindent where we write $V = V(x_1,\theta_1,\ell_1,x_2,\theta_2,\ell_2)$. Now, exchanging the sums in $k$ and $k'$ and using the symmetry of $V$, one obtains:
  \begin{equation}\label{dg}
 \begin{split}
  E_1= -\mu  \langle  \hspace{-0.8mm} \langle  g^K,\nabla_{x_1} \Psi(x_1,\theta_1,\ell_1,x_2,\theta_2,\ell_2) \cdot &\big(\nabla_x U + d\log \tilde{f}\big)(x_1,\theta_1) \rangle  \hspace{-0.8mm} \rangle \\
 - \frac{\mu}{4K} \sum_{k'=1}^K \bigg(\nabla_{x_1} V(C^k_{i(k'),j(k')}) \cdot \sum_{k=1}^K & \big(\nabla_{x_1} \Psi(C^k_{i(k),j(k)}) \delta_{(i(k),i(k'))} \\
 &+ \nabla_{x_1} \Psi(C^k_{j(k),i(k)}) \delta_{(j(k),i(k'))} \big)\bigg) \\
  - \frac{\mu}{4K} \sum_{k'=1}^K \bigg(\nabla_{x_1} V(C^k_{j(k'),i(k')}) \cdot  \sum_{k=1}^K & \big(\nabla_{x_1} \Psi(C^k_{i(k),j(k)}) \delta_{(i(k),j(k'))} \\
  &+ \nabla_{x_1} \Psi(C^k_{j(k),i(k)}) \delta_{(j(k),j(k'))} \big)\bigg).
  \end{split}
 \end{equation}
\noindent Because there is no restriction on the number of links per fiber, the sums over $k$ cannot be simplified in this case. In order to express the third and fourth terms, the number $C_i^{k'}$ (resp. $C_j^{k'}$) of fibers linked to fiber $i(k')$ (resp. $j(k')$) is introduced:
\begin{empheq}[left=\empheqlbrace]{align*}
C_i^{k'} &= \mbox{Card}(\{ k \; | \; i(k)=i(k') \; or \; j(k) = i(k') \},\\ 
C_j^{k'} &= \mbox{Card}(\{ k \; | \; i(k)=j(k') \; or \; j(k) = j(k') \}, 
\end{empheq}
\noindent where Card denote the cardinal of a set. Then, as $K \rightarrow \infty$, the following expression holds for any chosen fiber $k'$: 
 \begin{equation*}
 \begin{split}
 \frac{1}{2C_i^{k'}} \sum_{k=1}^K \big( \Psi(C^k_{i(k),j(k)})&\delta_{i(k),i(k')} + \Psi(C^k_{j(k),i(k)}) \delta_{j(k),i(k')})\big)\\
&  \underset{K \rightarrow \infty}{\rightarrow} \int (\Psi P)(X_{i(k')},\theta_{i(k')},\ell_1,x_2,\theta_2,\ell_2)dx_2 \frac{d\theta_2}{\pi} d\ell_1d\ell_2,
\end{split}
 \end{equation*}
 \noindent where 
 \begin{equation*}
P(X_{i(k')},\theta_{i(k')},\ell,x_2,\theta_2,\ell_2)= \frac{g(X_{i(k')},\theta_{i(k')},\ell,x_2,\theta_2,\ell_2)}{\int g(X_{i(k')},\theta_{i(k')},\ell_1,x_2,\theta_2,\ell_2) dx_2 \frac{d\theta_2}{\pi} d\ell_1d\ell_2}  ,
 \end{equation*}
 \noindent is the conditional probability of finding a link conditioned on the fact that one of the fibers of this link has the same location and orientation as $i(k')$. Then, as $N \rightarrow \infty, K \rightarrow \infty $ such that $\frac{K}{N} \rightarrow \xi>0$ , $C_i^{k'}$ is the mean number of links per fiber. The mean number of links in the volume $dX_{i(k')}d\theta_{i(k')}$ is $K\int g(X_{i(k')},\theta_{i(k')},\ell,x_2,\theta_2,\ell_2) dx_2 \frac{d\theta_2}{\pi} d\ell_1d\ell_2$ and the mean number of fibers in $dX_{i(k')}d\theta_{i(k')}$ is $Nf(X_{i(k')},\theta_{i(k')})$. Thus:
 \begin{equation*}
 C_i^{k'} \underset{\underset{\frac{K}{N} \rightarrow \xi>0}{\underset{K \rightarrow \infty}{ N \rightarrow \infty }}}{\rightarrow} \xi \frac{\int g(X_{i(k')},\theta_{i(k')},\ell_1,x_2,\theta_2,\ell_2) dx_2 \frac{d\theta_2}{\pi} d\ell_1d\ell_2}{f(X_{i(k')},\theta_{i(k')})}.
 \end{equation*}
 \noindent So, we get:
 \begin{equation*}
 \begin{split}
\sum_{k=1}^K \big( \Psi(C^k_{i(k),j(k)})&\delta_{i(k),i(k')} +\Psi(C^k_{j(k),i(k)})\delta_{j(k),i(k')} \big) \\ &\underset{\underset{\frac{K}{N} \rightarrow \xi>0}{\underset{K \rightarrow \infty}{ N \rightarrow \infty }}}{\rightarrow}
 \frac{2\xi}{f(X_{i(k')},\theta_{i(k')})} \int (\Psi g)(X_{i(k')},\theta_{i(k')},\ell_1,x_2,\theta_2,\ell_2) dx_2 \frac{d\theta_2}{\pi} d\ell_1d\ell_2.
 \end{split}
 \end{equation*}
 \noindent Inserting these expressions in Eq.~\eqref{dg}, one obtains:
  \begin{equation*}
 \begin{split}
E_1 \underset{\underset{\frac{K}{N} \rightarrow \xi>0}{\underset{K \rightarrow \infty}{ N \rightarrow \infty }}}{\rightarrow}  -\mu  \langle  \hspace{-0.8mm} \langle  g,\nabla_{x_1} \Psi(x_1,\theta_1,\ell_1,x_2,&\theta_2,\ell_2) \cdot \big(\nabla_x U + d\log \tilde{f}\big)(x_1,\theta_1) \rangle  \hspace{-0.8mm} \rangle \\
 - \mu \frac{\xi}{2K} \sum_{k'=1}^K \bigg(\nabla_{x_1} V(&C^k_{i(k'),j(k')}) \cdot  \psi_1(X_{i(k')},\theta_{i(k')}) \\
 &+ \nabla_{x_1} V(C^k_{j(k'),i(k')}) \cdot \psi_1(X_{j(k')},\theta_{j(k')}) \bigg), \\
  \end{split}
 \end{equation*}
 \noindent where,
 \begin{equation}\label{psi1}
\psi_1(x_1,\theta_1) = \frac{1}{f(x_1,\theta_1)}\int \big( g \nabla_{x_1} \Psi \big)(x_1,\theta_1,\ell_1,x_2,\theta_2,\ell_2)  dx_2\frac{d\theta_2}{\pi} d\ell_1d\ell_2.
 \end{equation}
 \noindent Finally, we find:
  \begin{equation*}
 \begin{split}
E_1 \underset{\underset{\frac{K}{N} \rightarrow \xi>0}{\underset{K \rightarrow \infty}{ N \rightarrow \infty }}}{\rightarrow}  &-\mu  \langle  \hspace{-0.8mm} \langle  g,\nabla_{x_1} \Psi(x_1,\theta_1,\ell_1,x_2,\theta_2,\ell_2) \cdot  \big(\nabla_x U + d\log \tilde{f}\big)(x_1,\theta_1) \rangle  \hspace{-0.8mm} \rangle \\
 &- \xi \mu  \langle  \hspace{-0.8mm} \langle  g, \nabla_{x_1} V(x_1,\theta_1,\ell_1,x_2,\theta_2,\ell_2) \cdot \psi_1(x_1,\theta_1)  \rangle  \hspace{-0.8mm} \rangle .
  \end{split}
 \end{equation*}
 \noindent After the same treatment for the four other terms of Eq.~\eqref{dg1} and in the limit $K,N \rightarrow \infty, \frac{K}{N} \rightarrow \xi>0$, one obtains the final equation for $g$ (writting $X$ for $(x_1,\theta_1,\ell_1,x_2,\theta_2,\ell_2)$):
   \begin{equation}\label{simpg}
 \begin{split}
&\frac{d}{dt} \langle  \hspace{-0.8mm} \langle  g(X),\Psi(X)  \rangle  \hspace{-0.8mm} \rangle \\
&= - \mu  \langle  \hspace{-0.8mm} \langle  g(X),\nabla_{x_1} \Psi(X) \cdot  \nabla_{x} U(x_1,\theta_1) \rangle  \hspace{-0.8mm} \rangle  - \mu  \langle  \hspace{-0.8mm} \langle  g,\nabla_{x_2} \Psi(X) \cdot \nabla_{x} U(x_2,\theta_2) \rangle  \hspace{-0.8mm} \rangle \\
&-\lambda  \langle  \hspace{-0.8mm} \langle  g,\partial_{\theta_1} \Psi(X) \partial_{\theta} U(x_1,\theta_1) \rangle  \hspace{-0.8mm} \rangle 
-\lambda  \langle  \hspace{-0.8mm} \langle  g,\partial_{\theta_2} \Psi(X) \partial_{\theta} U(x_2,\theta_2) \rangle  \hspace{-0.8mm} \rangle  \\
&- d \mu  \langle  \hspace{-0.8mm} \langle  g, \nabla_{x_1} \Psi(X) \cdot \nabla_{x} \log f(x_1,\theta_1) \rangle  \hspace{-0.8mm} \rangle 
 - d \mu  \langle  \hspace{-0.8mm} \langle  g, \nabla_{x_2} \Psi(X) \cdot \nabla_{x} \log f(x_2,\theta_2) \rangle  \hspace{-0.8mm} \rangle \\
& - d \lambda  \langle  \hspace{-0.8mm} \langle  g, \partial_{\theta_1} \Psi(X) \partial_{\theta} \log f(x_1,\theta_1) \rangle  \hspace{-0.8mm} \rangle 
 - d\lambda  \langle  \hspace{-0.8mm} \langle  g, \partial_{\theta_2} \Psi(X) \partial_{\theta} \log f(x_2,\theta_2) \rangle  \hspace{-0.8mm} \rangle \\
 & -\mu \xi  \langle  \hspace{-0.8mm} \langle  g,\nabla_{x_1} V(X)  \cdot \psi_1(x_1,\theta_1) \rangle  \hspace{-0.8mm} \rangle   -\mu \xi  \langle  \hspace{-0.8mm} \langle  g,\nabla_{x_1} V(X) \cdot \psi_2(x_1,\theta_1) \rangle  \hspace{-0.8mm} \rangle  \\
& -\lambda \xi  \langle  \hspace{-0.8mm} \langle  g,\big(\partial_{\theta_1} V(X) + \partial_{\theta_1} b(\theta_1,\theta_2)\big) \chi_1(x_1,\theta_1) \rangle  \hspace{-0.8mm} \rangle \\
 &-\lambda \xi  \langle  \hspace{-0.8mm} \langle  g,\big(\partial_{\theta_1} V(X) + \partial_{\theta_1} b(\theta_1,\theta_2)\big) \chi_2(x_1,\theta_1) \rangle  \hspace{-0.8mm} \rangle ,
  \end{split}
 \end{equation}
\noindent where,
\begin{empheq}[left=\empheqlbrace]{align}
\psi_2(x_1,\theta_1) &= \frac{1}{f(x_1,\theta_1)}\int \big(g \nabla_{x_2} \Psi \big)(x_2,\theta_2,\ell_2, x_1,\theta_1,\ell_1)  dx_2\frac{d\theta_2}{\pi} d\ell_1d\ell_2,\label{psi2}\\
\chi_1(x_1,\theta_1) &= \frac{1}{f(x_1,\theta_1)}\int \big( g \partial_{\theta_1} \Psi \big) (x_1,\theta_1,\ell_1,x_2,\theta_2,\ell_2)  dx_2\frac{d\theta_2}{\pi} d\ell_1d\ell_2,\label{chi1}\\
\chi_2(x_1,\theta_1) &=\frac{1}{f(x_1,\theta_1)} \int \big( g\partial_{\theta_2}  \Psi \big) (x_2,\theta_2,\ell_2, x_1,\theta_1,\ell_1)  dx_2\frac{d\theta_2}{\pi} d\ell_1d\ell_2.\label{chi2}
 \end{empheq}
\noindent We introduce the notation $Y_1 = (x_1,\theta_1,\ell_1)$ and $Y_2 = (x_2,\theta_2,\ell_2)$, and prove the following lemma:

\begin{lemma}\label{lemapp}
For any function $h(Y_1,Y_2)$, we have:
\begin{empheq}[left=\empheqlbrace]{align}
 \langle  \hspace{-0.8mm} \langle  g, h(Y_1,Y_2) \psi_1(x_1,\theta_1)  \rangle  \hspace{-0.8mm} \rangle  &= - \langle  \hspace{-0.8mm} \langle  \nabla_{x_1} \big(g(X) \frac{F_h(x_1,\theta_1)}{f(x_1,\theta_1)}\big) , \Psi(X) \rangle  \hspace{-0.8mm} \rangle ,\nonumber\\
 \langle  \hspace{-0.8mm} \langle  g, h(Y_1,Y_2) \psi_2(x_1,\theta_1)  \rangle  \hspace{-0.8mm} \rangle  &= - \langle  \hspace{-0.8mm} \langle  \nabla_{x_2} \big(g(X) \frac{F_h(x_2,\theta_2)}{f(x_2,\theta_2)}\big) , \Psi(X) \rangle  \hspace{-0.8mm} \rangle ,\nonumber\\
 \langle  \hspace{-0.8mm} \langle  g, h(Y_1,Y_2) \chi_1(x_1,\theta_1)  \rangle  \hspace{-0.8mm} \rangle  &= - \langle  \hspace{-0.8mm} \langle  \partial_{\theta_1} \big(g(X) \frac{F_h(x_1,\theta_1)}{f(x_1,\theta_1)}\big) , \Psi(X) \rangle  \hspace{-0.8mm} \rangle ,\label{Eqspsichi}\\
 \langle  \hspace{-0.8mm} \langle  g, h(Y_1,Y_2) \chi_2(x_1,\theta_1)  \rangle  \hspace{-0.8mm} \rangle  &= - \langle  \hspace{-0.8mm} \langle  \partial_{\theta_2} \big(g(X) \frac{F_h(x_2,\theta_2)}{f(x_2,\theta_2)}\big) , \Psi(X) \rangle  \hspace{-0.8mm} \rangle ,\nonumber
\end{empheq}
\noindent where $\psi_1$, $\psi_2$, $\chi_1$ and $\chi_2$ are defined by Eq.~\eqref{psi1} and Eqs. \eqref{psi2}-\eqref{chi2}, and where :
\begin{equation}\label{Fh}
F_h(x_1,\theta_1) = \int (gh)(x_1,\theta_1,\ell_1,x_2,\theta_2,\ell_2) dx_2 \frac{d\theta_2}{\pi} d\ell_2 d\ell_1.
\end{equation}
\end{lemma}

\begin{proof}
Note that for any function $h(Y_1,Y_2)$, we have:
\begin{equation*}
\begin{split}
 \langle  \hspace{-0.8mm} \langle  g, h(Y_1,Y_2) \psi_1(x_1,\theta_1)  \rangle  \hspace{-0.8mm} \rangle  \hspace{5cm}&\\
= \int \big( gh \big)(Y_1,Y_2) \bigg(\frac{1}{f(x_1,\theta_1)} \int (g \nabla_{x_1} \Psi)(x_1,\theta_1,\ell_4,x_3,\theta_3,&\ell_3) dx_3 \frac{d\theta_3}{\pi} d\ell_4 d\ell_3 \bigg) \\
&dx_1 \frac{d\theta_1}{\pi} d\ell_1 dx_2 \frac{d\theta_2}{\pi} d\ell_2\\
 = \int \bigg( \frac{1}{f(x_1,\theta_1)}\int (gh)(Y_1,Y_2) dx_2 \frac{d\theta_2}{\pi} d\ell_2 d\ell_1 \bigg) \big(g \nabla_{x_1} \Psi \big)(&x_1,\theta_1,\ell_4,x_3,\theta_3,\ell_3) \\
 &dx_1 \frac{d\theta_1}{\pi} d\ell_4 dx_3 \frac{d\theta_3}{\pi} d\ell_3\\
= - \int \nabla_{x_1} \bigg( g(x_1,\theta_1,\ell_4,x_3,\theta_3,\ell_3) \frac{F_h(x_1,\theta_1)}{f(x_1,\theta_1)}\bigg) \Psi(x_1,\theta_1,\ell_4,&x_3,\theta_3,\ell_3)\\
& dx_1 \frac{d\theta_1}{\pi} d\ell_4 dx_3 \frac{d\theta_3}{\pi} d\ell_3\\
 = - \langle  \hspace{-0.8mm} \langle  \nabla_{x_1} \big(g(X) \frac{F_h(x_1,\theta_1)}{f(x_1,\theta_1)}\big) , \Psi(X) \rangle  \hspace{-0.8mm} \rangle \hspace{4cm}&,
\end{split}
\end{equation*}
\noindent with $F_h$ defined by~\eqref{Fh}. Similarly, we have:
\begin{equation*}
\begin{split}
 \langle  \hspace{-0.8mm} \langle  g, h(Y_1,Y_2) \psi_2(x_1,\theta_1)  \rangle  \hspace{-0.8mm} \rangle \hspace{6cm} & \\
= \int \big( gh \big)(Y_1,Y_2) \bigg(\frac{1}{f(x_1,\theta_1)} \int (g \nabla_{x_2} \Psi)(x_3,\theta_3,\ell_3,x_1,\theta_1,\ell_4)& dx_3 \frac{d\theta_3}{\pi} d\ell_4 d\ell_3 \bigg) \\
&dx_1 \frac{d\theta_1}{\pi} d\ell_1 dx_2 \frac{d\theta_2}{\pi} d\ell_2\\
 = \int \bigg( \frac{1}{f(x_1,\theta_1)}\int (gh)(Y_1,Y_2) dx_2 \frac{d\theta_2}{\pi} d\ell_2 d\ell_1 \bigg) \big(g \nabla_{x_2} \Psi \big)(x_3,&\theta_3,\ell_3,x_1,\theta_1,\ell_4)\\
 & dx_1 \frac{d\theta_1}{\pi} d\ell_4 dx_3 \frac{d\theta_3}{\pi} d\ell_3\\
 = \int \bigg( \big(g \nabla_{x_2} \Psi \big)(Y'_1,Y'_2) \frac{1}{f(x'_2,\theta'_2)}\int (gh)(x'_2,\theta'_2,\ell'_4,x'_3,\theta'_3,\ell'_3)& dx'_3 \frac{d\theta'_3}{\pi} d\ell'_3 d\ell'_4 \bigg)\\
 &  dx'_1 \frac{d\theta'_1}{\pi} d\ell'_1 dx'_2 \frac{d\theta'_2}{\pi} d\ell'_2\\
=- \int \nabla_{x'_2} \bigg( g(Y'_1,Y'_2) \frac{F_h(x'_2,\theta'_2)}{f(x'_2,\theta'_2)}\bigg) \Psi(Y'_1,Y'_2) dx'_1 \frac{d\theta'_1}{\pi} d\ell'_1 dx'_2 &\frac{d\theta'_2}{\pi} d\ell'_2\\
 = - \langle  \hspace{-0.8mm} \langle  \nabla_{x_2}  \big(g(X) \frac{F_h(x_2,\theta_2)}{f(x_2,\theta_2)}\big) , \Psi(X) \rangle  \hspace{-0.8mm} \rangle \hspace{3.5cm}&,
\end{split}
\end{equation*}
\noindent After the same computations for $\chi_1$ and $\chi_2$, we obtain Eqs.~\eqref{Eqspsichi}.
\end{proof}

Now, lemma \ref{lemapp} allows us to write the formal limit $K,N \rightarrow \infty, \frac{K}{N} \rightarrow \xi$ of Eq.~\eqref{simpg} which reads:
\begin{equation}\label{systg}
\begin{split}
\frac{dg}{dt} - \mu \nabla_{x_1} \cdot \big(g\nabla_{x} U(x_1,\theta_1) + & \xi \frac{g}{f(x_1,\theta_1)}F_1(x_1,\theta_1)\big) \\
- \lambda \partial_{\theta_1} (g\partial_{\theta} U(x_1,\theta_1)+ \xi &\frac{g}{f(x_1,\theta_1)}F_{2}(x_1,\theta_1)) \\
-\mu \nabla_{x_2} \cdot (g\nabla_{x} U(x_2,\theta_2)& + \xi \frac{g}{f(x_2,\theta_2)}F_1(x_2,\theta_2))\\
- \lambda \partial_{\theta_2}  (g\partial_{\theta} U(x_2,\theta_2)+ &\xi  \frac{g}{f(x_2,\theta_2)}F_{2}(x_2,\theta_2) )\\
-d\mu \nabla_{x_1}\cdot (\frac{g}{f(x_1,\theta_1)}&\nabla_{x} f(x_1,\theta_1) ) - d \mu \nabla_{x_2}\cdot (\frac{g}{f(x_2,\theta_2)} \nabla_{x} f(x_2,\theta_2) ) \\
- d\lambda \partial_{\theta_1}(\frac{g}{f(x_1,\theta_1)}& \partial_{\theta} f(x_1,\theta_1) ) - d \lambda \partial_{\theta_2}(\frac{g}{f(x_2,\theta_2)}\partial_{\theta} f(x_2,\theta_2) )=0,
\end{split}
\end{equation}
\noindent where $F_1$ and $F_2$ read:
\begin{empheq}[left=\empheqlbrace]{align*}
F_1(x_1,\theta_1) &= \int\limits \nabla_{x_1} V(x_1,\theta_1,\ell_1,x_2,\theta_2,\ell_2) g(x_1,\theta_1,\ell_1,x_2,\theta_2,\ell_2) dx_2 \frac{d\theta_2}{\pi}d\ell_1d\ell_2,\\
F_2(x_1,\theta_1)& = \int\limits  \bigg(g\big(\partial_{\theta_1} V +  \partial_{\theta_1} b\big) \bigg) (x_1,\theta_1,\ell_1,x_2,\theta_2,\ell_2) dx_2 \frac{d\theta_2}{\pi}d\ell_1d\ell_2.
  \end{empheq}

Finally, the link creation/deletion Poisson processes, of frequencies $\nu_f$ and $\nu_d$ respectively, classically lead to a source term $S(g)$ for Eq.~\eqref{systg}. We recall that a link between two fibers is formed only if the fibers intersect each other, whereas the link deletion process obviously acts on existing links only. This leads to the following source term:
\begin{equation*}
\begin{split}
S(g)(x_1,\theta_1,\ell_1,x_2,\theta_2,\ell_2) = &\nu_f f(x_1,\theta_1)f(x_2,\theta_2) \delta(\ell_1,\bar{\ell}(x_1,\theta_1,x_2,\theta_2)\delta(\ell_2,\bar{\ell}(x_2,\theta_2,x_1,\theta_1) \\
&- \nu_d g(x_1,\theta_1,\ell_1,x_2,\theta_2,\ell_2),
\end{split}
\end{equation*} 
\noindent where the first term corresponds to the link creation process while the second one, to the link deletion process. Here, the quantity $f(x_1,\theta_1)f(x_2,\theta_2) \delta(\ell_1,\bar{\ell}(x_1,\theta_1,x_2,\theta_2) $ $ \delta(\ell_2,\bar{\ell}(x_2,\theta_2,x_1,\theta_1) dx_1 \frac{d\theta_1}{\pi}dx_2 \frac{d\theta_2}{\pi}d\ell_1d\ell_2$  gives the probability of finding a fiber located within a volume $dx_1\frac{d\theta_1}{\pi}$ about $(x_1,\theta_1)$ and a fiber located within a volume $dx_2\frac{d\theta_2}{\pi}$ about $(x_2,\theta_2)$, such that they intersect with associated lengths within a volume $d\ell_1 d\ell_2$ about $(\ell_1,\ell_2)$. The link creation process generates a new link distribution function proportional to this probability at a rate $\nu_f$. The quantity $-\nu_d g(x_1,\theta_1,\ell_1,x_2,\theta_2,\ell_2)$ corresponds to the decay of the link distribution function with rate $\nu_d$ due to the link deletion process.


\section{Computation of the non linear term $\int\limits 
\partial_\theta (G[\rho M] \rho M)\Psi d\theta$}\label{app3}
This section is devoted to the computation of the term $X_3$ given by~\eqref{X3}. For the sake of clarity, the following notations are introduced:
 \begin{equation}\label{notations}
 M = M_{\theta_0}, \;   s_0=\sin 2(\theta - \theta_0), \;  c_0=\cos 2(\theta - \theta_0).
  \end{equation}
  \noindent By symmetry,  $\langle h(2(\theta - \theta_0)) \rangle  = 0$ for all odd functions $h$ on $[-\frac{\pi}{2},\frac{\pi}{2}]$, where $\langle \cdot   \rangle $ is the average defined in Theorem \ref{thm4}. We also note from Eq.~\eqref{coeffC1C2}, Hypothesis~\ref{hyp4} and Proposition~\ref{prop4} that we have: \begin{equation}\label{C2}
  C_2 = \frac{\alpha L^4 \gamma}{48 \eta_f} = \frac{4rd L^2}{\xi 48 \rho c(r)} = \frac{rdL^2}{12 \xi\rho c(r)}.
  \end{equation}
  \noindent Using Green's formula, Eqs.~\eqref{forceG},~\eqref{partialPsi} and the same arguments as for $X_2$, we get: 
\begin{equation*}
\begin{split}
 X_3 &= - \int\limits_{-\frac{\pi}{2}}^{\frac{\pi}{2}} (G[\rho M] \rho M) \partial_{\theta} \Psi \frac{d\theta}{\pi} \\
& = -C_2\int\limits_{-\frac{\pi}{2}}^{\frac{\pi}{2}} \bigg(\int\limits_{-\frac{\pi}{2}}^{\frac{\pi}{2}} \nabla_x^2 (\rho M(\theta')): B(\theta,\theta') \frac{d\theta'}{\pi} \bigg)\rho M(\theta)\partial_{\theta} \Psi \frac{d\theta}{\pi} \\
&= -\rho C_2 \int\limits_{-\frac{\pi}{2}}^{\frac{\pi}{2}} \nabla_x^2 (\rho M(\theta')): \bigg(\int\limits_{-\frac{\pi}{2}}^{\frac{\pi}{2}} B(\theta,\theta') M(\theta) \partial_{\theta} \Psi \frac{d\theta}{\pi} \bigg) \frac{d\theta'}{\pi}\\
&= -\frac{\rho C_2}{2r} \int\limits_{-\frac{\pi}{2}}^{\frac{\pi}{2}} \bigg(\nabla_x^2 (\rho M(\theta')): \int\limits_{-\frac{\pi}{2}}^{\frac{\pi}{2}} B(\theta,\theta') (M(\theta) - \frac{1}{Z^2}) \frac{d\theta}{\pi} \bigg) \frac{d\theta'}{\pi}.
\end{split}
\end{equation*}
\noindent Let us first compute $\nabla_x^2 (\rho M)$. We have:
\begin{equation*}
\nabla_x^2 (\rho M) = M \nabla_x^2 \rho+ \nabla_x M \otimes \nabla_x \rho + \nabla_x \rho \otimes \nabla_x M + \rho \nabla_x^2 M,
\end{equation*}
\noindent where $\nabla_x M$ is given by~\eqref{nablaxM}. A direct computation gives:
\begin{equation*}
\begin{split}
\nabla_x^2 M = 2r M \bigg[ s_0 \nabla_x^2 \theta_0+ 2\big(r s_0^2 - c_0\big) \nabla_x\theta_0 \otimes \nabla_x\theta_0 \bigg],
\end{split}
\end{equation*}
\noindent and thus:
\begin{equation*}
\begin{split}
\nabla_x^2 (\rho M) = &M \bigg[  \nabla_x^2  \rho + 2\rho r s_0 \nabla_x^2 \theta_0\\
&+  2rs_0(\nabla_x\theta_0 \otimes \nabla_x \rho+\nabla_x \rho \otimes \nabla_x \theta_0) + 4\rho r(rs_0^2 - c_0)\nabla_x \theta_0 \otimes \nabla_x\theta_0 \bigg].
\end{split}
\end{equation*}
 \noindent We now turn towards the computation of  \begin{equation*}
 \int\limits_{-\frac{\pi}{2}}^{\frac{\pi}{2}} \big( B(\theta,\theta') M(\theta) - \frac{1}{Z^2} B(\theta,\theta') \big)\frac{d\theta}{\pi},
 \end{equation*}
 \noindent where $B(\theta,\theta')$ is given by~\eqref{Bteta}.
For this purpose, we decompose: 
 \begin{equation*}
 \omega = (\omega.\omega_0) \omega_0 + (\omega.\omega_0^\perp) \omega_0^\perp = \cos(\theta - \theta_0) \omega_0 + \sin(\theta - \theta_0) \omega_0^\perp,
 \end{equation*}
 \noindent where $\omega_0 = \omega(\theta_0)$ and $\omega_0^\perp$ such that $(\omega_0,\omega_0^\perp)$ is a direct ortho-normal basis of $\mathbb{R}^2$. Using basic trigonometric formulae, one notes that:
\begin{equation*}
\begin{split}
\omega \otimes \omega & = \frac{1}{2}\bigg[ (1 + c_0) (\omega_0 \otimes\omega_0) + (1-c_0)\omega_0^\perp \otimes\omega_0^\perp + s_0 [\omega_0\otimes\omega_0^\perp + \omega_0^\perp \otimes\omega_0]\bigg]\\
& = \frac{1}{2}\bigg[ I + c_0 [\omega_0 \otimes\omega_0 - \omega_0^\perp \otimes\omega_0^\perp] + s_0 [\omega_0 \otimes\omega_0^\perp + \omega_0^\perp \otimes\omega_0]\bigg],\\
\end{split}
\end{equation*}
\noindent where $I$ is the identity matrix. Denoting $c_0 = c_0(\theta)$, $s_0 = s_0(\theta)$, $c_0' = c_0(\theta')$ and $s_0' = s_0(\theta')$, we get:
\begin{equation*}
\begin{split}
B(\theta,\theta') &= \sin 2(\theta-\theta') [\omega \otimes \omega + \omega' \otimes \omega'] \\
&= \frac{1}{2}[ s_0 c_0' - s_0' c_0] \bigg[ 2I + (c_0+c_0') [\omega_0 \otimes\omega_0 - \omega_0^\perp \otimes\omega_0^\perp] \\
&+ (s_0+ s_0') [\omega_0 \otimes\omega_0^\perp + \omega_0^\perp \otimes\omega_0]\bigg] \\
&=[s_0 c_0' - s_0' c_0]I\\
&+ \frac{1}{2}[ c_0s_0 c_0' + s_0 c_0^2(\theta') -s_0' c_0^2(\theta) - s_0' c_0'c_0] [\omega_0 \otimes\omega_0 - \omega_0^\perp \otimes\omega_0^\perp]\\
&+\frac{1}{2}[ s_0^2 c_0' + s_0 s_0'c_0' -s_0' c_0s_0 - s_0'^2 c_0][\omega_0 \otimes\omega_0^\perp + \omega_0^\perp \otimes\omega_0]. \\
\end{split}
\end{equation*}
\noindent Note that $B$ is anti-symmetric, i.e.\ $B(\theta',\theta) = - B(\theta,\theta')$. From the properties of $M$, we get: 
\begin{empheq}[left=\empheqlbrace]{align*}
\int\limits_{-\frac{\pi}{2}}^{\frac{\pi}{2}} B(\theta,\theta') \frac{d\theta}{\pi} = -[&\omega_0 \otimes\omega_0 - \omega_0^\perp \otimes\omega_0^\perp] \frac{s_0'}{4} + [\omega_0 \otimes\omega_0^\perp + \omega_0^\perp \otimes\omega_0] \frac{c_0'}{4},\\
\int\limits_{-\frac{\pi}{2}}^{\frac{\pi}{2}} M(\theta) B(\theta,\theta') \frac{d\theta}{\pi}&=- s_0' \langle c_0 \rangle  I \\
&- \frac{1}{2}[\omega_0 \otimes\omega_0 - \omega_0^\perp \otimes\omega_0^\perp] (c_0' s_0' \langle c_0 \rangle  + s_0' \langle c_0^2 \rangle  )\\
& + \frac{1}{2}[\omega_0 \otimes\omega_0^\perp + \omega_0^\perp \otimes\omega_0](c_0' \langle s_0^2 \rangle  - s_0'^2 \langle c_0 \rangle ). 
\end{empheq}
\noindent Then, we have:
\begin{equation*}
\begin{split}
\int\limits_{-\frac{\pi}{2}}^{\frac{\pi}{2}} B(\theta,\theta') \big(M(\theta)  - \frac{1}{Z^2} \big) \frac{d\theta}{\pi} &= -s_0' \langle c_0 \rangle  I \\
&+ [\omega_0 \otimes\omega_0 - \omega_0^\perp \otimes\omega_0^\perp]T_1+[\omega_0 \otimes\omega_0^\perp + \omega_0^\perp \otimes\omega_0]T_2,
\end{split}
\end{equation*}
\noindent with
$$
T_1= \frac{s_0'}{4Z^2} - \frac{c_0' s_0' \langle c_0 \rangle  + s_0' \langle c_0^2 \rangle }{2}, \quad
T_2 =  \frac{c_0' \langle s_0^2 \rangle  - s_0'^2 \langle c_0 \rangle }{2} - \frac{c_0'}{4Z^2}.
$$
\noindent Note that this expression is decomposed into an even function $T_2$ of $\theta'$ and an odd function of $\theta'$ composed of $s_0'\langle c_0 \rangle $ and $T_1$. Therefore, $\langle h, T_1 \rangle  = 0$ for all even functions $h$ and $\langle h, T_2 \rangle =0$ for all odd functions $h$. Moreover, from integration by parts, the following relations hold: 
\begin{empheq}[left=\empheqlbrace]{align}
\langle s_0^2 \rangle  &= \frac{\langle c_0 \rangle }{r},\nonumber\\
\langle c_0^2 \rangle  &= 1 - \frac{\langle c_0 \rangle }{r},\nonumber\\
\langle c_0^3 \rangle  &= \langle c_0 \rangle  - \frac{1}{r} + 2\frac{\langle c_0 \rangle }{r^2},\label{form}\\
\langle c_0^4 \rangle  &= 1 - 2 \frac{\langle c_0 \rangle }{r} + \frac{3}{r^2} - 6 \frac{\langle c_0 \rangle }{r^3},\nonumber\\
\langle c_0s_0^2 \rangle  &= \frac{1}{r}(1 - 2\frac{\langle c_0 \rangle }{r}),\nonumber\\
\langle s_0^4 \rangle  &= \frac{3}{r^2}(1 - 2\frac{\langle c_0 \rangle }{r}).\nonumber
\end{empheq}
\noindent Then,
\begin{equation}\label{F1ent}
\begin{split}
\int\limits_{-\frac{\pi}{2}}^{\frac{\pi}{2}} &\nabla_x^2 (\rho M(\theta')) : \bigg(\int\limits_{-\frac{\pi}{2}}^{\frac{\pi}{2}} (B(\theta,\theta')(M(\theta) - \frac{1}{Z^2}) )\frac{d\theta}{\pi}\bigg)\frac{d\theta'}{\pi} \\
&=  \nabla_x^2 \rho : [\omega_0 \otimes\omega_0^\perp + \omega_0^\perp \otimes\omega_0] \langle T_2 \rangle \\
&+ 4\rho r\nabla_x\theta_0 \otimes \nabla_x \theta_0 : [\omega_0 \otimes\omega_0^\perp + \omega_0^\perp \otimes\omega_0] \bigg(r \langle s_0^2 T_2 \rangle - \langle c_0 T_2 \rangle \bigg)\\
&+2r (\nabla_x \rho \otimes \nabla_x \theta_0 + \nabla_x\theta_0 \otimes \nabla_x \rho):\bigg[-\langle c_0 \rangle \langle s_0^2 \rangle I  + [\omega_0 \otimes\omega_0 - \omega_0^\perp \otimes\omega_0^\perp]\langle s_0 T_1 \rangle  \bigg]\\
&+2\rho r \nabla_x(\nabla_x \theta_0):\bigg[- \langle s_0^2 \rangle \langle c_0 \rangle  I + [\omega_0 \otimes\omega_0 - \omega_0^\perp \otimes\omega_0^\perp] \langle s_0 T_1 \rangle    \bigg],
\end{split}
\end{equation}
\noindent where (using Eqs.~\eqref{form} and integration by parts):
\begin{empheq}[left=\empheqlbrace]{align*}
\langle T_2 \rangle  &= -\frac{\langle c_0 \rangle }{4Z^2}, \\
\langle c_0 T_2 \rangle & =  \frac{\langle c_0 \rangle ^2}{2r^2} - \frac{1}{4Z^2} (1-\frac{\langle c_0 \rangle }{r}),\\
\langle s_0^2 T_2 \rangle  &= (-\frac{1}{r} + 2\frac{\langle c_0 \rangle }{r^2})\frac{1}{4Z^2} + 2\frac{\langle c_0 \rangle ^2}{r^3} - \frac{\langle c_0 \rangle }{r^2},\\
\langle s_0 T_1 \rangle  &= -\frac{\langle c_0 \rangle }{r} + \frac{3\langle c_0 \rangle ^2}{2r^2} + \frac{\langle c_0 \rangle }{4rZ^2},\\
r \langle s_0^2 T_2 \rangle  - \langle c_0 T_2 \rangle  &= \frac{\langle c_0 \rangle }{r} \big[\frac{1}{Z^2} - 1 + \frac{3 \langle c_0 \rangle }{2r}\big].
\end{empheq}
\noindent Then, after some computations and using Eq.~\eqref{C2}, Eq.~\eqref{F1ent} simplifies into:
\begin{equation}\label{F1ents}
\begin{split}
X_3&=- \frac{dL^2}{24 \xi c(r)} \int\limits_{-\frac{\pi}{2}}^{\frac{\pi}{2}} \nabla_x^2 (\rho M(\theta')) : \bigg(\int\limits_{-\frac{\pi}{2}}^{\frac{\pi}{2}} (B(\theta,\theta')(M(\theta) - \frac{1}{Z^2}) )\frac{d\theta}{\pi}\bigg)\frac{d\theta'}{\pi} \\
&=  -\frac{dL^2}{24 \xi c(r)} \bigg( -\nabla_x^2 \rho : [\omega_0 \otimes\omega_0^\perp + \omega_0^\perp \otimes\omega_0] \frac{\langle c_0 \rangle }{4Z^2}\\
&+ 4\rho \langle c_0 \rangle \nabla_x\theta_0 \otimes \nabla_x \theta_0 : [\omega_0 \otimes\omega_0^\perp + \omega_0^\perp \otimes\omega_0] (\frac{1}{4Z^2} - 1 + \frac{3\langle c_0 \rangle }{2r} )\\
&+2\langle c_0 \rangle  ( \rho \nabla_x \nabla_x \theta_0 + \nabla_x \rho \otimes \nabla_x \theta_0 + \nabla_x\theta_0 \otimes \nabla_x \rho):\bigg[-\langle c_0 \rangle I  \\
&+ [\omega_0 \otimes\omega_0 - \omega_0^\perp \otimes\omega_0^\perp] (\frac{1}{4Z^2} -1 + \frac{3\langle c_0 \rangle }{2r}  )\bigg)\bigg].
\end{split}
\end{equation}
\noindent We note that $\langle c_0 \rangle  = c(r)$. Eq.~\eqref{F1ents} leads to~\eqref{valueX3}.

\section*{Acknowledgements}
 This work was supported by the ``R\'egion Midi Pyr\'en\'ees'', under grant APRTCN 2013. PD acknowledges support from the British ``Engineering and Physical Research Council'' under grant ref: EP/M006883/1, from the Royal Society and the Wolfson foundation through a Royal Society Wolfson Research Merit Award and from NSF by NSF Grant RNMS11-07444 (KI-Net). PD is on leave from CNRS, Institut de Math\'ematiques de Toulouse, France. DP gratefully acknowledges the hospitality of Imperial College London, where part of this research was conducted.

\bibliographystyle{ieeetr}

\begin{thebibliography}{}

\end{thebibliography}


\begin{thebibliography}{10}

\bibitem{Alonso_etal_CellMolBioeng14}
R. Alonso, J. Young and Y. Cheng,  A particle interaction model for the simulation of biological, cross-linked fibers inspired from flocking theory, {\em Cellular and molecular bioengineering} {\bf 7} (2014) 58-72.

\bibitem{Alt_Dembo_MathBiosci99}
W. Alt and M. Dembo, Cytoplasm dynamics and cell motion: two phase flow models, {\em Math. Biosci.} {\bf 156} (1999) 207-228. 


\bibitem{Astrom_etal_PRE08} 
J. A. \r{A}str\"om, P. B. S. Kumar, I. Vattulainen and M. Karttunen, Strain hardening, avalanches, and strain softening in dense cross-linked actin networks, {\em Phys. Rev. E} {\bf 77} (2008) 051913.

\bibitem{Bardos_etal_TransAMS84}
C. Bardos, R. Santos and R. Sentis, Diffusion approximation and computation of the critical size, {\em Trans. Amer. Math. Soc.} {\bf 284} (1984) 617-649.

\bibitem{Baskaran_Marchetti_PRE08}
A. Baskaran and M. C. Marchetti, Hydrodynamics of self-propelled hard rods, {\em Phys. Rev. E} {\bf 77} (2008) 011920.

\bibitem{Bertin_etal_NewJPhys13}
E. Bertin, H. Chat\'e, F. Ginelli, S. Mishra, A. Peshkov and S. Ramaswamy, Mesoscopic theory for fluctuating active nematics, {\em New J. Phys.} {\bf 15} (2013) 085032.

\bibitem{Bird_etal_Wiley87}
R. Bird, C. Curtiss, R. Armstrong, and O. Hassager, Dynamics of Polymeric Liquids, Vol. 2, Kinetic Theory, John Wiley \& Sons, New York, 1987.

\bibitem{Broedersz_etal_PRL10} 
C. P. Broedersz, M. Depken, N. Y. Yao, M. R. Pollak, D. A. Weitz and F. C. MacKintosh, Cross-link-governed dynamics of biopolymer networks, {\em Phys. Rev. Lett.} {\bf 105} (2010) 238101. 

\bibitem{Buxton_etal_ExpPolLett09}
G.A. Buxton, N. Clarke and P. J. Hussey, Actin dynamics and the elasticity of cytoskeletal networks, {\em Express Polymer Letters} {\bf 3} (2009) 579-587.

\bibitem{Carlen_etal_PhysicaD13}
E. Carlen, R. Chatelin, P. Degond, and B Wennberg, Kinetic hierarchy and propagation of chaos in biological swarm models, Phys. D 260 (2013) 90-111. 

\bibitem{Carlen_etal_M3AS13}
E. Carlen, P. Degond and B Wennberg, Kinetic limits for pair-interaction driven master equations and biological swarm models, {\em Math. Models Methods Appl. Sci.} {\bf 23} (2013)1339-1376. 

\bibitem{Ciuperca_etal_DCDS12}
I. S. Ciuperca, E. Hingant, L. I. Palade and L. Pujo-Menjouet, Fragmentation and monomer lengthening of rod-like polymers, a relevant model for prion proliferation, {\em Discrete Contin. Dyn. Syst. Ser. B} {\bf 17} (2012) 775-799. 

\bibitem{Degond_etal_JSP13}
P. Degond, C. Appert-Rolland, M. Moussaid, J. Pettr\'e and G. Theraulaz, A hierarchy of heuristic-based models of crowd dynamics,  {\em J. Stat. Phys.} {\bf 152} (2013) 1033-1068.

\bibitem{Degond_etal_CMS15}
P. Degond, G Dimarco, T. B. N. Mac and N. Wang, Macroscopic models of collective motion with repulsion, {\em Commun. Math. Sci.}, to appear, arxiv preprint \# 1404.4886. 

\bibitem{Degond_etal_MAA13}
P. Degond, J-G. Liu, S. Motsch and V. Panferov, Hydrodynamic models of self-organized dynamics: derivation and existence theory, {\em Methods Appl. Anal.} {\bf 20} (2013) 089-114.

\bibitem{Degond_MasGallic_TTSP87}
P. Degond and S. Mas-Gallic, Existence of solutions and diffusion approximation for a model Fokker-Planck equation, {\em Transport Theory and Statistical Physics} {\bf 16} (1987) 589-636. 

\bibitem{Degond_Motsch_M3AS08}
P. Degond and S. Motsch, Continuum limit of self-driven particles with orientation interaction, {\em Math. Models Methods Appl. Sci.} {\bf 18 Suppl.} (2008) 1193-1215. 

\bibitem{Doi_Edwards_Oxford99}
M. Doi and S. F. Edwards, The Theory of Polymer Dynamics, International Series of Monographs on Physics, Vol 73, Oxford University Press, Oxford, 1999.

\bibitem{Frouvelle_M3AS12}
A. Frouvelle, A continuum model for alignment of self-propelled particles with anisotropy and density-dependent parameters, {\em Math. Models Methods Appl. Sci.} {\bf 22} (2012) 1250011.

\bibitem{Ginelli_etal_PRL10}
F. Ginelli, F. Peruani, M. B\"ar and H. Chat\'e, Large-scale collective properties of self-propelled rods, {\em Phys. Rev. Lett.} {\bf104} (2010) 184502.

\bibitem{Head_etal_PRE03}
D. A Head, A. J. Levine and F. C MacKintosh, Distinct regimes of elastic response and deformation modes of cross-linked cytoskeletal and semiflexible polymer networks, {\em Phys. Rev. E} {\bf 68} (2003) 061907.

\bibitem{Joanny_etal_NewJPhys07}
J. F. Joanny, F. J\"ulicher, K. Kruse and J. Prost, Hydrodynamic theory for multi-component active polar gels, {\em New J. Phys.} {\bf 9} (2007) 422. 

\bibitem{Karsher_etal_BiophysJ03}
H. Karsher, J. Lammerding, H. Huang, R. T. Lee, R. D. Kamm and M. R. Kaazempur-Mofrad, A three-dimensional viscoelastic model for cell deformation with experimental verification, {\em Biophysical Journal} {\bf 85} (2003) 3336-3349. 

\bibitem{Maier_Saupe_ZNaturforsch58}
W. Maier and A. Saupe, Eine einfache molekulare Theorie des nematischen kristallinfl\"ussigen Zustandes, {\em Z. Naturforsch.} {\bf 13} (1958) 564-566.

\bibitem{Mischler_Mouhot_InventMath13}
S. Mischler and C. Mouhot, Kac's Program in Kinetic Theory, {\em Invent. Math.} {\bf 193} (2013) 1-147. 

\bibitem{Mischler_etal_PTRF15}
S. Mischler, C. Mouhot and B. Wennberg, A new approach to quantitative propagation
of chaos for drift, diffusion and jump processes, {\em Probab. Theory Related Fields} {\bf 161} (2015) 1-59. 

\bibitem{Oelz_etal_CellAdhMigr08}
D. Oelz, C. Schmeiser and J. V. Small, Modeling of the actin-cytoskeleton in symmetric
lamellipodial fragments, {\em Cell Adhesion and Migration} {\bf 2} (2008) 117-126.

\bibitem{Onsager_AnnNYAcadSci49}
L. Onsager, The effects of shape on the interaction of colloidal particles, {\em Ann. New York Acad. Sci.} {\bf 51} (1949) 627-659.

\bibitem{Peruani_etal_PRE06}
F. Peruani, A. Deutsch and M. B\"ar, Nonequilibrium clustering of self-propelled rods, {\em Phys. Rev. E} {\bf 74} (2006) 030904(R).

\bibitem{Peurichard_etal_preprint15}
D. Peurichard, F. Delebecque, A. Lorsignol, C. Barreau, J. Rouquette, X. Descombes, L. Casteilla and P. Degond, Simple mechanical cues could explain adipose tissue morphology, submitted. 

\bibitem{Poupaud_AsymptAnal91}
F. Poupaud, Diffusion approximation of the linear semiconductor Boltzmann equation: analysis of boundary layers, {\em Asymptot. Anal.}	{\bf 4} (1991) 293-317. 

\bibitem{Sone2002}
Y. Sone, Kinetic Theory and Fluid Dynamics, Birkhausser, 2002.

\bibitem{Taber_etal_JMechMatStruct11}
L. A. Taber, Y. Shi, L. Yang and P. V. Bayly, A poroelastic model for cell crawling including mechanical coupling between cytoskeletal contraction and actin polymerization, {\em Journal of  Mechanics of Materials and Structures} {\bf 6} (2011) 569-589.

\bibitem{Taylor1996}
M.E. Taylor,  Partial Differential Equations III: Nonlinear Equations, Applied Mathematical Sciences, vol 117, Springer, 1996. 

\bibitem{Vicsek_etal_PRL95}
T. Vicsek, A. Czir\'ok, E. Ben-Jacob, I. Cohen and O. Shochet, Novel type of phase transition in a system of self-driven particles, {\em Phys. Rev. Lett.} {\bf 75} (1995) 1226-1229.

\bibitem{Vicsek_Zafeiris_PhysRep12}
T. Vicsek and A. Zafeiris, Collective motion, {\em Phys. Rep.} {\bf 517} (2012) 71-140.


\end{thebibliography}

\end{document}